\documentclass[11pt,a4paper]{article}

\usepackage{cite}
\usepackage{tikz}\usetikzlibrary{matrix,decorations.pathreplacing,positioning}
\usepackage{hyperref}
\usepackage{amsmath,amsthm,amssymb,bm,relsize}

\usepackage{setspace}
\usepackage{multirow}
\usepackage{array}
\newcolumntype{x}[1]{>{\centering\arraybackslash\hspace{0pt}}p{#1}}

\usepackage{wasysym}

\usepackage{caption}
\usepackage{subcaption}

\usepackage{enumitem}
\setitemize{itemsep=-1pt}
\setenumerate{itemsep=-1pt}

\usepackage{titling}
\usepackage[margin=2.6cm]{geometry}
\usepackage{pgfplots}
\pgfplotsset{width=10cm,compat=1.9}

\definecolor{myg}{RGB}{220,220,220}

\theoremstyle{definition}
\newtheorem{theorem}{Theorem}[section]
\newtheorem{corollary}[theorem]{Corollary}
\newtheorem{proposition}[theorem]{Proposition}
\newtheorem{lemma}[theorem]{Lemma}

\newtheorem{definition}[theorem]{Definition}
\newtheorem{example}[theorem]{Example}

\usepackage{etoolbox}
\AtEndEnvironment{example}{\null\hfill$\blacktriangle$}%

\newtheorem{notation}[theorem]{Notation}
\newtheorem{remark}[theorem]{Remark}

\newcommand*{\myproofname}{Proof of the claim}

\usepackage{caption}

\newcommand{\numberset}{\mathbb}

\newcommand{\F}{\numberset{F}}

\newcommand{\mS}{\mathcal{S}}
\newcommand{\mC}{\mathcal{C}}

\newcommand{\mA}{\mathcal{A}}
\newcommand{\mN}{\mathcal{N}}

\newcommand{\mF}{\mathcal{F}}

\newcommand{\mU}{\mathcal{U}}

\newcommand{\mX}{\mathcal{X}}
\newcommand{\mY}{\mathcal{Y}}
\newcommand{\dH}{d^\textnormal{H}}

\newcommand{\st}{\, : \,}

\newcommand{\mE}{\mathcal{E}}

\newcommand{\mV}{\mathcal{V}}

\newcommand{\inn}{\textnormal{in}}
\newcommand{\out}{\textnormal{out}}

\newcommand{\lin}{\textnormal{lin}}

\newcommand{\CC}{\textnormal{C}}

\newcommand{\HH}{\textnormal{H}}

\newcommand{\dto}{\dashrightarrow}

\newcommand{\bd}[1]{{\bf #1}}

\newcommand{\bfT}{\bf T}

\newcommand{\mincut}{\textnormal{min-cut}}

\newlength{\mynodespace}
\newcommand{\degin}{\partial^-}
\newcommand{\degout}{\partial^+}

\newtheorem{family}{Family}

\title{\textbf{\huge Network Decoding}}

\usepackage{authblk}

\author[1]{Allison Beemer}
\affil[1]{Department of Mathematics, University of Wisconsin-Eau Claire, U.S.A.}

\author[2]{Altan B. K\i l\i\c{c}\thanks{A. B. K. is supported by the Dutch Research Council through grant VI.Vidi.203.045.}}

\author[3]{Alberto Ravagnani\thanks{A. R. is supported by the Dutch Research Council through grants VI.Vidi.203.045, 
OCENW.KLEIN.539, 
and by the Royal Academy of Arts and Sciences of the Netherlands.}}
\affil[2,3]{Department of Mathematics and Computer Science, Eindhoven University of Technology, the Netherlands}

\date{}

\begin{document}


\maketitle

\begin{abstract}
We consider the problem of 
error control
in a coded, multicast network,
focusing on the scenario where the errors can 
occur only on a \textit{proper subset} of the network edges.
We model this problem via an adversarial noise, presenting 
a formal framework and a series of techniques to obtain upper and lower bounds on the network's (1-shot) capacity, improving on the best 
currently known
results.
In particular,
we show that
traditional cut-set bounds are not tight in general in the presence of a restricted adversary,
and that the non-tightness of these is caused precisely 
by the restrictions imposed on the noise (and not, as one may expect, by the alphabet size).
We also show that, in sharp contrast
with the typical situation within network coding,
capacity
cannot be achieved in general by combining linear network coding with end-to-end channel coding, not even when the underlying network has a single source and a single terminal.
We finally illustrate how network
\textit{decoding} techniques 
are necessary to achieve capacity in the scenarios we examine, exhibiting capacity-achieving schemes 
and lower bounds for various classes of networks.\unskip\parfillskip 0pt \par
\end{abstract}



\bigskip

\section*{Introduction}

Global propagation of interconnected devices and the ubiquity of communication demands in unsecured settings signify the importance of a unified understanding of the limits of communications in networks.
The correction of errors modeled by an adversarial noise
has been studied in a number of previous works, with results ranging from those concerning network capacity to specific code design (see~\cite{YC06,CY06, YY07, byzantine, M07, YNY07, YYZ08,MANIAC, RK18,kosut14,Zhang,nutmanlangberg} among many others). In this paper, we focus on the effects of a small, and yet crucial,
modification of previous models, where a malicious actor can access and alter transmissions across a \textit{proper} subset of edges within a network. We show that not only does this modification disrupt the sharpness of known results on network capacity, but that more specialized network coding (in fact, network \textit{decoding}) becomes necessary to achieve the capacity.

We adopt the setting of a communication network whose inputs are drawn from a finite alphabet and whose internal (hereafter referred to as intermediate) nodes may process incoming information before forwarding (this is known as network \textit{coding}; see e.g.~\cite{ahlswede,linearNC,koettermedard}).
We phrase our work in terms of adversarial noise, but we note that our treatment truly addresses \textit{worst-case} errors, also providing guarantees in networks where there may be random noise, or a combination of random and adversarial noise. We assume that the adversary is omniscient in the sense that they may design their attacks given full knowledge of the network topology, of the symbols sent along all its edges,
and of the operations performed at the intermediate nodes.

Again, in contrast to most previous work in the area, we concentrate on networks with noise occurring only on a proper subset of the network edges: for example, an adversary who may see all transmitted symbols, but has limited access to edges in the network when actively corrupting symbols. We examine the {1-shot capacity} of adversarial networks with these restricted adversaries, which roughly measures the number of alphabet symbols that can be sent with zero error during a single transmission round. The case of multi-shot capacity, where more than one transmission round is considered, is also interesting, but will involve a separate treatment and different techniques from our work here.  Compellingly, the simple act of restricting possible error locations fundamentally alters the problem of computing the 1-shot capacity, as well as the manner in which this capacity may be achieved. This is discussed in further detail below.

This paper expands upon (and disrupts) the groundwork laid in \cite{RK18}, where a combinatorial framework for adversarial networks and a generalized method for porting point-to-point coding-theoretic results to the network setting are established. The work in \cite{RK18} makes a start on addressing the case of restricted adversaries, also unifying the best-known upper bounds on~1-shot capacity for such networks. These upper bounds fall under the category of cut-set bounds, where an edge-cut is taken between network source(s) and terminal(s) and a bound on capacity is derived from an induced channel that takes only the cut-set edges into account. We note that Cai and Yeung previously gave generalizations of the Hamming, Singleton, and Gibert-Varshamov bounds to the network setting using edge-cuts in \cite{YC06,CY06}. The work in \cite{RK18} allows for any point-to-point bound to be ported via an edge-cut.

In the case of random noise in a single-source, single-terminal network, it is well-known that the cut-set upper bound on capacity given by the Max-Flow Min-Cut Theorem may be achieved simply by forwarding information at the intermediate network nodes; see e.g.~\cite{elgamalkim}. It has also been shown that in the scenarios of multiple terminals or where a malicious adversary may choose among \textit{all} network edges, cut-set bounds may be achieved via  \textit{linear} operations~(over the relevant finite field) at intermediate nodes (combined with rank-metric codes in the presence of a malicious adversary); see~\cite{SKK,MANIAC,epfl2}. In prior, preliminary work in this area \cite{beemer2021curious}, we demonstrated a minimal example of a network,
which we called the Diamond Network,
with a restricted adversary, where both (1) the Singleton cut-set bound (the best cut-set bound for this network) cannot be achieved, and (2) the actual 1-shot capacity cannot be achieved using linear operations at intermediate nodes. In fact, this example of network requires network \textit{decoding} in order to achieve capacity. Via an extension of the Diamond Network, termed the {Mirrored Diamond Network}, we found that it is possible with restricted adversaries for the 1-shot capacity to meet the Singleton cut-set bound, but still be forced to use non-linear network codes to achieve this rate.

Generally speaking, the act of limiting the location of adversarial action is enough to put an end to the tightness of the best currently known cut-set bounds, and also to destroy the ability to achieve capacity with linear network codes in combination with end-to-end coding/decoding. This paper is the first stepping stone aimed at understanding the effects of restricting an adversary to a proper subset of the networks' edges, and how to compute~(1-shot) capacities of such adversarial networks. 


The paper is organized as follows. Sections \ref{sec:motiv} and \ref{sec:channel} introduce our problem set-up and notation. Section \ref{sec:diamond} is devoted to the Diamond Network mentioned above, which gives the smallest example illustrating our results. In Subsection \ref{sec:info} we present an information-theoretic view of our problem for the case of random rather than adversarial noise, in order to generate further intuition.
To address the more general problem at hand, we will later in the paper~(in Section \ref{sec:double-cut-bd}) present a technique showing that the capacity of any network is upper bounded by the capacity of an induced, much less complicated network. The induced network utilized borrows from the original idea of cut-sets, but instead of studying the information flow through a single edge-cut, it considers the information flow between \textit{two} edge-cuts. We call the result we obtain as an application of this idea the Double-Cut-Set Bound.
The less-complicated networks induced by the bound's statement are introduced in Section \ref{sec:net-2-and-3} and the collection of families we study throughout the remainder of the paper are introduced in Subsection \ref{sec:families}. In Section \ref{sec:upper}, we propose new combinatorial techniques to obtain upper bounds for the 1-shot capacities of these families, and in Section \ref{sec:2level_lower}, we establish lower bounds. Section \ref{sec:linear} shows that there is a strong separation between linear and non-linear capacities, solidifying the necessity of network decoding in this setting. Finally, Section \ref{sec:open} is devoted to conclusions, a discussion of open problems, and possible future works.

\section{Problem Statement and Motivation}
\label{sec:motiv}

We focus on the typical scenario studied within the context of network coding; see~\cite{ahlswede,CY06,YC06,randomHo,linearNC,KK1,koettermedard,epfl1,SKK,WSK,YY07,YNY07,YangYeung2,jaggi2005,RK18} among many others.
Namely, one source of information attempts to transmit information packets to a collection of terminals through a network of intermediate nodes. The packets are 
elements of a finite alphabet $\mA$ and the network is acyclic and delay-free.
We are interested in solving the multicast problem (that is, each terminal demands all packets)
under the assumption that an omniscient adversary acts on the network edges according to some restrictions. 
These concepts will be formalized later 
in the paper; see Section~\ref{sec:channel}.
The goal of this section is to illustrate
the motivation behind our paper through an example.

Consider the single-source, two-terminal network $\mN$ depicted in Figure \ref{fig:ie1}. We want to compute the number of alphabet packets that the source $S$ can transmit to the terminal $T$ in a single transmission round, called the \textit{1-shot capacity}.

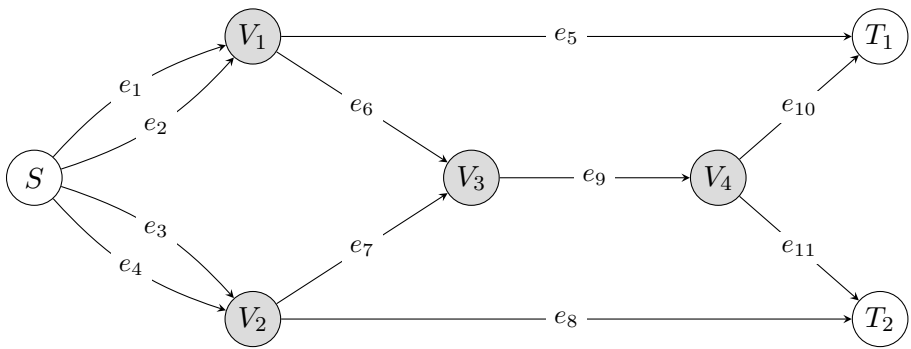
\begin{figure}[h!]
\centering
\begin{tikzpicture}
\tikzset{vertex/.style = {shape=circle,draw,inner sep=0pt,minimum size=1.9em}}
\tikzset{nnode/.style = {shape=circle,fill=myg,draw,inner sep=0pt,minimum
size=1.9em}}
\tikzset{edge/.style = {->,> = stealth}}
\tikzset{dedge/.style = {densely dotted,->,> = stealth}}
\tikzset{ddedge/.style = {dashed,->,> = stealth}}

\node[vertex] (S1) {$S$};

\node[shape=coordinate,right=\mynodespace of S1] (K) {};

\node[nnode,above=0.6\mynodespace of K] (V1) {$V_1$};

\node[nnode,below=0.6\mynodespace of K] (V2) {$V_2$};

\node[nnode,right=\mynodespace of K] (V3) {$V_3$};

\node[nnode,right=\mynodespace of V3] (V4) {$V_4$};


\node[vertex,right=3\mynodespace of V1 ] (T1) {$T_1$};

\node[vertex,right=3\mynodespace of V2] (T2) {$T_2$};

\draw[edge,bend left=15] (S1)  to node[fill=white, inner sep=3pt]{\small $e_1$} (V1);

\draw[edge,bend right=15] (S1)  to node[fill=white, inner sep=3pt]{\small $e_2$} (V1);

\draw[edge,bend left=15] (S1) to  node[fill=white, inner sep=3pt]{\small $e_3$} (V2);

\draw[edge,bend right=15] (S1) to  node[fill=white, inner sep=3pt]{\small $e_4$} (V2);

\draw[edge,bend left=0] (V1) to  node[fill=white, inner sep=3pt]{\small $e_6$} (V3);

\draw[edge,bend left=0] (V4)  to node[fill=white, inner sep=3pt]{\small $e_{10}$} (T1);

\draw[edge,bend left=0] (V4)  to node[fill=white, inner sep=3pt]{\small $e_{11}$} (T2);

\draw[edge,bend left=0] (V1)  to node[fill=white, inner sep=3pt]{\small $e_{5}$} (T1);

\draw[edge,bend left=0] (V2)  to node[fill=white, inner sep=3pt]{\small $e_{8}$} (T2);

\draw[edge,bend left=0] (V2) to  node[fill=white, inner sep=3pt]{\small $e_7$} (V3);

\draw[edge,bend left=0] (V3) to  node[fill=white, inner sep=3pt]{\small $e_{9}$} (V4);


\end{tikzpicture}
\caption{An example of a network.\label{fig:ie1}}
\end{figure}

If no adversary acts on the network,
then the traditional cut-set bounds 
are sharp, if the network alphabet is sufficiently large; see~\cite{ahlswede,linearNC,koettermedard}. Since the (edge) min-cut
between $S$ and any $T \in \{T_1,T_2\}$ is~2, no more than 2 packets can be sent in a single transmission round. Furthermore, a strategy that achieves this bound is obtained by routing packet~1 across paths
$e_1 \to e_5$ and $e_4 \to e_8$,
and packet 2
 across paths
$e_2 \to e_6 \to e_9 \to e_{10}$ and $e_3 \to e_7 \to e_9 \to e_{11}$.

This paper focuses on an adversarial model. That is, 
a malicious ``outside actor'' can corrupt up to $t$ information packets 
sent via the edges.
Note that the
term adversary has no cryptographic meaning in our setting and it simply models 
the situation where \textit{any} pattern of~$t$ errors needs to be corrected.

Now suppose that the network $\mN$
is vulnerable, with an adversary able to change the value of up to $t=1$ of the 
network edges.
Figure \ref{fig:ie2} represents the same network as Figure~\ref{fig:ie1}, but with vulnerable (dashed) edges.
This scenario has been extensively investigated within network coding with possibly multiple terminals and multiple sources; see for instance~\cite{MANIAC,RK18,YC06,CY06,SKK,KK1}. In particular, it follows from the Network Singleton Bound of \cite{CY06,YC06,RK18} that 
the network has capacity 0, meaning that the largest unambiguous code has size 1 (the terminology will be formalized later).

\begin{figure}[h!]
\centering
\begin{tikzpicture}
\tikzset{vertex/.style = {shape=circle,draw,inner sep=0pt,minimum size=1.9em}}
\tikzset{nnode/.style = {shape=circle,fill=myg,draw,inner sep=0pt,minimum
size=1.9em}}
\tikzset{edge/.style = {->,> = stealth}}
\tikzset{dedge/.style = {densely dotted,->,> = stealth}}
\tikzset{ddedge/.style = {dashed,->,> = stealth}}

\node[vertex] (S1) {$S$};

\node[shape=coordinate,right=\mynodespace of S1] (K) {};

\node[nnode,above=0.6\mynodespace of K] (V1) {$V_1$};

\node[nnode,below=0.6\mynodespace of K] (V2) {$V_2$};

\node[nnode,right=\mynodespace of K] (V3) {$V_3$};

\node[nnode,right=\mynodespace of V3] (V4) {$V_4$};


\node[vertex,right=3\mynodespace of V1 ] (T1) {$T_1$};

\node[vertex,right=3\mynodespace of V2] (T2) {$T_2$};

\draw[ddedge,bend left=15] (S1)  to node[fill=white, inner sep=3pt]{\small $e_1$} (V1);

\draw[ddedge,bend right=15] (S1)  to node[fill=white, inner sep=3pt]{\small $e_2$} (V1);

\draw[ddedge,bend left=15] (S1) to  node[fill=white, inner sep=3pt]{\small $e_3$} (V2);

\draw[ddedge,bend right=15] (S1) to  node[fill=white, inner sep=3pt]{\small $e_4$} (V2);

\draw[ddedge,bend left=0] (V1) to  node[fill=white, inner sep=3pt]{\small $e_6$} (V3);

\draw[ddedge,bend left=0] (V4)  to node[fill=white, inner sep=3pt]{\small $e_{10}$} (T1);

\draw[ddedge,bend left=0] (V4)  to node[fill=white, inner sep=3pt]{\small $e_{11}$} (T2);

\draw[ddedge,bend left=0] (V1)  to node[fill=white, inner sep=3pt]{\small $e_{5}$} (T1);

\draw[ddedge,bend left=0] (V2)  to node[fill=white, inner sep=3pt]{\small $e_{8}$} (T2);

\draw[ddedge,bend left=0] (V2) to  node[fill=white, inner sep=3pt]{\small $e_7$} (V3);

\draw[ddedge,bend left=0] (V3) to  node[fill=white, inner sep=3pt]{\small $e_{9}$} (V4);

\end{tikzpicture}
\caption{\label{fig:ie2} The network of Figure \ref{fig:ie1}, where all edges are vulnerable (dashed).}
\end{figure}
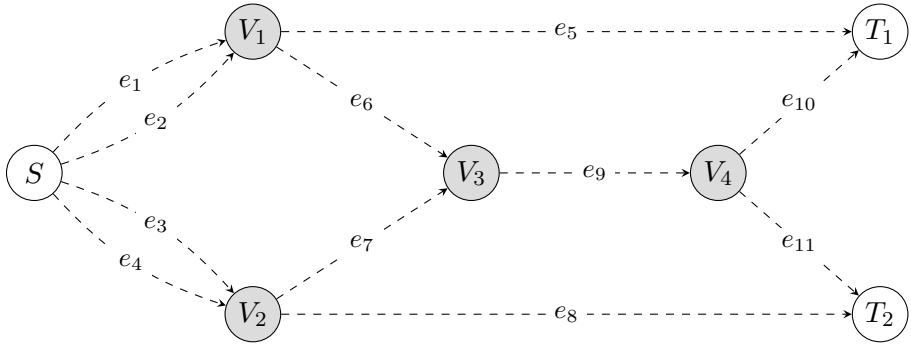

We recall that in the case of a network with multiple sources, multiple terminals, and an adversary able to corrupt up to $t$ of the network edges, the capacity region was computed in~\cite{MANIAC,RK18,epfl2}. In the scenario just described, a capacity-achieving scheme can be obtained by combining linear network coding with rank-metric end-to-end coding.
We will comment on this again in Theorem \ref{thm:mcm}.

\begin{figure}[h!]
\centering
\begin{tikzpicture}
\tikzset{vertex/.style = {shape=circle,draw,inner sep=0pt,minimum size=1.9em}}
\tikzset{nnode/.style = {shape=circle,fill=myg,draw,inner sep=0pt,minimum
size=1.9em}}
\tikzset{edge/.style = {->,> = stealth}}
\tikzset{dedge/.style = {densely dotted,->,> = stealth}}
\tikzset{ddedge/.style = {dashed,->,> = stealth}}

\node[vertex] (S1) {$S$};

\node[shape=coordinate,right=\mynodespace of S1] (K) {};

\node[nnode,above=0.6\mynodespace of K] (V1) {$V_1$};

\node[nnode,below=0.6\mynodespace of K] (V2) {$V_2$};

\node[nnode,right=\mynodespace of K] (V3) {$V_3$};

\node[nnode,right=\mynodespace of V3] (V4) {$V_4$};


\node[vertex,right=3\mynodespace of V1 ] (T1) {$T_1$};

\node[vertex,right=3\mynodespace of V2] (T2) {$T_2$};

\draw[ddedge,bend left=15] (S1)  to node[fill=white, inner sep=3pt]{\small $e_1$} (V1);

\draw[ddedge,bend right=15] (S1)  to node[fill=white, inner sep=3pt]{\small $e_2$} (V1);

\draw[ddedge,bend left=15] (S1) to  node[fill=white, inner sep=3pt]{\small $e_3$} (V2);

\draw[ddedge,bend right=15] (S1) to  node[fill=white, inner sep=3pt]{\small $e_4$} (V2);

\draw[ddedge,bend left=0] (V1) to  node[fill=white, inner sep=3pt]{\small $e_6$} (V3);

\draw[edge,bend left=0] (V4)  to node[fill=white, inner sep=3pt]{\small $e_{10}$} (T1);

\draw[edge,bend left=0] (V4)  to node[fill=white, inner sep=3pt]{\small $e_{11}$} (T2);

\draw[edge,bend left=0] (V1)  to node[fill=white, inner sep=3pt]{\small $e_{5}$} (T1);

\draw[edge,bend left=0] (V2)  to node[fill=white, inner sep=3pt]{\small $e_{8}$} (T2);

\draw[ddedge,bend left=0] (V2) to  node[fill=white, inner sep=3pt]{\small $e_7$} (V3);

\draw[ddedge,bend left=0] (V3) to  node[fill=white, inner sep=3pt]{\small $e_{9}$} (V4);


\end{tikzpicture} 
\caption{{The network of Figure \ref{fig:ie1}, where only the dashed edges are vulnerable.}\label{fig:introex1}}
\end{figure}
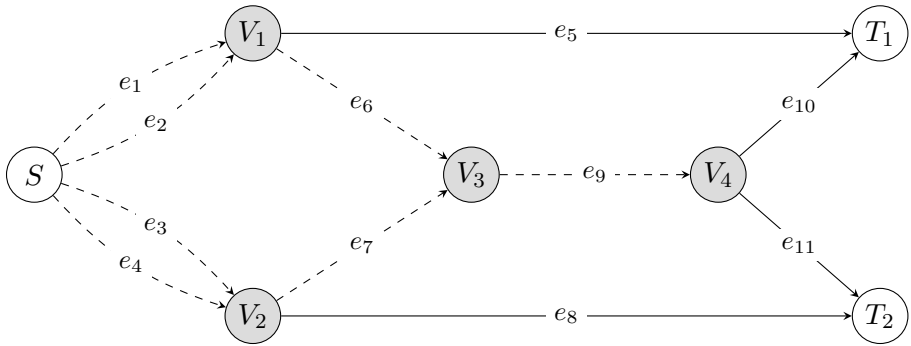

We now turn to the scenario that motivates this paper. The network remains vulnerable, but this time the adversary can only corrupt at most one of the \textit{dashed} edges in Figure~\ref{fig:introex1}; the solid edges are \textit{not} vulnerable.
Our main question is the same (what is the largest number of alphabet packets that can be transmitted?), but the answer is less obvious/known than before. By restricting the adversary to operate on a proper subset of the edges, one expects the capacity to increase with respect to the ``unrestricted'' situation. Therefore a natural question is whether the rate of one packet in a single channel use can be achieved.
As we will show in Theorem~\ref{computC},
this is not possible: instead, the partially vulnerable network has capacity $\log_{|\mA|} (|\mA|-1)$,
where~$\mA$ denotes the network alphabet.
As we will see, capacity can be achieved by making nodes $V_1$, $V_2$ and $V_3$ partially 
\textit{decode} received information
(which explains the title of our paper).
This is in sharp contrast with the case of an unrestricted adversary, where capacity can be achieved with end-to-end encoding/decoding.

The goal of this paper is to develop the theoretical framework needed to study networks that are partially vulnerable to adversarial noise, comparing results and strategies 
with those currently available in contexts that are to date much better understood.

\section{Channels and Networks}
\label{sec:channel}

In this section we include some preliminary definitions and results that will be needed throughout the entire paper. This will also allow us to establish the notation and to state the problems we will investigate in rigorous mathematical language.
This section is divided into two subsections. The first is devoted to arbitrary~(adversarial) channels, while in the second we focus our attention on communication networks and their capacities.

\subsection{Adversarial Channels}
\label{subsect:channels}

In our treatment, we will use the definition of (adversarial) channels proposed in~\cite{RK18}, based on the notion of~\textit{fan-out sets}.
This concept essentially dates back to Shannon's fundamental paper on the zero-error capacity of a channel~\cite{shannon_zero}. 
Under this approach, a channel is fully specified by the collection of symbols $y$ that can be received when a given symbol $x$ is transmitted. Considering transition probabilities does not make sense in this model since the noise is assumed to be of an ``adversarial'' nature. The latter assumption conveniently models all those scenarios where \textit{any} error pattern of a given type must be corrected
and no positive probability of unsuccessful decoding is tolerated. 
This is the scenario considered (often implicitly) in standard network coding references such as~\cite{SKK,YC06,CY06,KK1} among many others.
We will briefly consider a probabilistic regime in Subsection~\ref{sec:info} with the goal of forming intuition about certain information theory phenomena we will observe.

\begin{definition}
\label{dd1}
A (\textbf{discrete}, \textbf{adversarial}) \textbf{channel} is a map
$\Omega: \mX \to 2^{\mY} \setminus \{\emptyset\}$, where~$\mX$ and~$\mY$ are finite non-empty sets called the \textbf{input} and \textbf{output alphabets} respectively.
The notation for such a channel is $\Omega: \mX \dashrightarrow \mY$. We say that $\Omega$ is  \textbf{deterministic} if $|\Omega(x)|=1$ for all $x \in \mX$.
We call~$\Omega(x)$ the \textbf{fan-out set} of $x$.
\end{definition}

Note that a deterministic channel $\Omega: \mX \dashrightarrow \mY$ can be naturally identified with a function~$\mX \to \mY$, which we denote by $\Omega$ as well.

\begin{definition}
\label{dd2}
A \textbf{(outer) code} for a channel $\Omega: \mX \dashrightarrow \mY$ is a non-empty subset $\mC \subseteq \mX$. We say that~$\mC$ is
\textbf{unambiguous} if $\Omega(x) \cap \Omega(x') =\emptyset$
for all $x, x' \in \mC$ with $x \neq x'$. 
\end{definition}

The base-two logarithm of the largest size of an unambiguous code for a given channel is its~\textit{1-shot capacity}. In the graph theory representation of channels proposed by Shannon in~\cite{shannon_zero}, the~1-shot capacity coincides with the base-two logarithm of the largest cardinality of an independent set in the corresponding graph. We refer to~\cite[Section~II]{RK18}
for a more detailed discussion.

\begin{definition}
\label{def:future}
The (\textbf{$1$-shot}) \textbf{capacity} of a channel $\Omega:\mX \dashrightarrow \mY$ is the real number
$$\CC_1(\Omega)=\max\left\{\log_2 |\mC| \; : \; \mC \subseteq \mX \mbox{ is an unambiguous code for $\Omega$}\right\}.$$
\end{definition}

We give an example to illustrate the previous definitions.

\begin{example}
Let $\mX=\mY=\{0,1,2,3,4,5,6,7\}$. Define a channel $\Omega: \mX \dto \mY$ by setting
$$\Omega(x)=
\begin{cases}
      \{0,2\}  & \text{if} \ \ x = 0, \\
      \{0,1,4,6\}  & \text{if} \ \ x = 1, \\
      \{2,3,5\}  & \text{if} \ \ x = 2, \\
      \{2,3,4,7\}  & \text{if} \ \ x = 3,
\end{cases}
\qquad  \qquad
\Omega(x)=
\begin{cases}      
      \{2,3,4,6\}  & \text{if} \ \ x = 4, \\
      \{0,1,5\}  & \text{if} \ \ x = 5, \\
      \{6\}  & \text{if} \ \ x = 6, \\
      \{0,1,5,7\}  & \text{if} \ \ x = 7. \\
\end{cases}$$
Clearly, $\Omega$ is not deterministic. It can be checked that the only unambigious code for $\Omega$ of size~3 is~$\mC=\{3,5,6\}$, and that there are no unambiguous codes of size 4. Therefore we have~$\CC_1(\Omega)=\log_2 3$.
\end{example}

In this paper we focus solely on the  
1-shot capacity of 
channels. 
While other capacity notions can be considered as well (e.g., the \textit{zero-error capacity}), in the networking context these are significantly more technical to treat
than the 1-shot capacity, especially when focusing on restricted noise. We therefore omit them
in this first paper on network decoding and leave them as a future research direction; see Section \ref{sec:open}.

We next describe how channels can be compared and combined with each other, referring to~\cite{RK18} for a more complete treatment.
 
\begin{definition} \label{deffiner}
Let $\Omega_1,\Omega_2 : \mX \dashrightarrow \mY$ be channels.
We say that $\Omega_1$ is \textbf{finer}
than $\Omega_2$ (or that~$\Omega_2$ is \textbf{coarser} than $\Omega_1$) if $\Omega_1(x) \subseteq \Omega_2(x)$ for all $x \in \mX$.
The notation is $\Omega_1 \le \Omega_2$.
\end{definition}

Finer channels have larger capacity, as the following simple result states.

\begin{proposition}
\label{prop:finer}
Let $\Omega_1,\Omega_2 : \mX \dashrightarrow \mY$ be channels with $\Omega_1 \le \Omega_2$. We have $\CC_1(\Omega_1) \ge \CC_1(\Omega_2)$. 
\end{proposition}

Channels with compatible output/input alphabets can be concatenated with each other via the following construction. 

\begin{definition}
Let $\Omega_1:\mX_1 \dashrightarrow \mY_1$ and $\Omega_2:\mX_2 \dashrightarrow \mY_2$
be channels, with $\mY_1 \subseteq \mX_2$.
The \textbf{concatenation} of
$\Omega_1$ and $\Omega_2$ is the channel
$\Omega_1 \blacktriangleright \Omega_2 : \mX_1 \dashrightarrow \mY_2$ defined by
$$(\Omega_1 \blacktriangleright \Omega_2)(x):= \bigcup_{y \in \Omega_1(x)} \Omega_2(y).$$
\end{definition}

The concatenation of channels is associative in the following precise sense.

\begin{proposition}
\label{prop:11}
Let $\Omega_i:\mX_i \dashrightarrow \mY_i$ be channels, for $i \in \{1,2,3\}$,
with $\mY_i \subseteq \mX_{i+1}$ for 
$i \in \{1,2\}$. We have
$(\Omega_1 \blacktriangleright \Omega_2) \blacktriangleright \Omega_3 = 
\Omega_1 \blacktriangleright (\Omega_2 \blacktriangleright \Omega_3)$. 
\end{proposition}

The previous result allows us to 
write expressions such as $\Omega_1 \blacktriangleright \Omega_2 \blacktriangleright \Omega_3$
without parentheses, when all concatenations are defined. 

We conclude this subsection with a 
discrete
version 
of the \textit{data processing inequality} from classical information theory; see e.g.~\cite[Section 2.8]{coverthomas}.

\begin{proposition}
\label{dpi}
Let $\Omega_1:\mX_1 \dashrightarrow \mY_1$ and $\Omega_2:\mX_2 \dashrightarrow \mY_2$
be channels, with $\mY_1 \subseteq \mX_2$.
We have~$\CC_1(\Omega_1 \blacktriangleright \Omega_2) \le \min\{\CC_1(\Omega_1), \, 
\CC_1(\Omega_2)\}$. 
\end{proposition}

\subsection{Networks and Their Capacities}

In this subsection we formally define 
communication networks, network codes, and the channels they induce.
Our approach is inspired by~\cite{RK18}, even though the notation used in this paper differs slightly. We omit some of the details and refer the interested reader directly to~\cite{RK18}.

\begin{definition}
\label{def:network}
A (\textbf{single-source}) \textbf{network}  is a 4-tuple $\mN=(\mV,\mE, S, \bfT)$ where:
\begin{enumerate}[label=(\Alph*)] 
\item $(\mV,\mE)$ is a finite, directed, acyclic multigraph,
\item $S \in \mV$ is the \textbf{source},
\item ${\bf T} \subseteq \mV$ is the set of \textbf{terminals} or \textbf{sinks},
\end{enumerate}
Note that we allow multiple parallel directed edges. We also assume that the following hold.
\begin{enumerate}[label=(\Alph*)] \setcounter{enumi}{3}
 \item $|{\bf T}| \ge 1$, $S \notin  {\bf T}$.
\item \label{prnE} For any $T \in {\bf T}$, there exists a directed path from $S$ to $T$.
\item The source does not have incoming edges, and terminals do not have outgoing edges. \label{prnF}
\item For every vertex $V \in \mV \setminus (\{S\} \cup \bfT)$, there exists a directed path from
$S$ to $V$ and a directed path from $V$ to $T$ for some $T \in {\bf T}$. \label{prnG}
\end{enumerate}
The elements of $\mV$ are called \textbf{vertices} or \textbf{nodes}. The elements of 
$\mV \setminus (\{S\} \cup {\bf T})$ are called \textbf{intermediate} nodes. 
A (\textbf{network}) \textbf{alphabet} is a finite set $\mA$ with $|\mA| \ge 2$.
The elements of~$\mA$ are called \textbf{symbols} or \textbf{packets}. We say that $\mN$ is a \textbf{single-terminal} network if $|\bfT| = 1$.
\end{definition}

The network alphabet is interpreted as the set of symbols that can be sent over the edges of the network.

\begin{notation} \label{not:fixN}
Throughout the paper, $\mN=(\mV,\mE,S,{\bf T})$ will always denote a network and $\mA$ an alphabet, as in Definition~\ref{def:network}, unless otherwise stated. We let
$$\mincut_\mN(V,V')$$
be the minimum cardinality of an edge-cut (a set of edges that cuts the connection) between vertices $V,V' \in \mV$. We denote the
set of incoming and outgoing edges of $V \in \mV$ by $\inn(V)$ and~$\out(V)$, respectively, and their cardinalities by $\degin(V)$ and $\degout(V)$.
These cardinalities are called the \textbf{in-degree} and \textbf{out-degree} of the vertex $V$, respectively.
\end{notation}

The following concepts will be crucial in our approach.

\begin{definition} \label{def:prece}
The edges of a network $\mN=(\mV,\mE, S, \bfT)$ can be partially ordered as follows. For $e,e' \in \mE$, we say that $e$ \textbf{precedes} $e'$ if 
there exists a directed path in $\mN$ that starts with $e$ and ends with~$e'$. The notation is 
$e \preccurlyeq e'$.
\end{definition}

\begin{notation} \label{not:ext}
Following the notation of Definition~\ref{def:prece}, it is well known in graph theory that
the  partial order on $\mE$ can be extended to a (not necessarily unique) total order. By definition, such an extension $\le$ satisfies the following property:
$e \preccurlyeq e'$ implies  $e \le e'$.
Throughout the paper we will assume that such an order extension has been fixed in the various networks and we denote it by $\le$ (none of the results of this paper depend on the particular choice of the total order).
Moreover, we illustrate the chosen total order via the labeling of the edges; see, for example, Figure~\ref{fig:ie1}.
\end{notation}

In our model, the intermediate nodes of a network process incoming packets according to prescribed functions. We do not assume any restrictions on these functions. In particular, even when $\mA$ is a linear space over a given finite field, we do \textit{not} require the functions to be linear over the underlying field.
This is in strong contrast with the most common approach taken in the context of network coding; see, for instance, 
\cite{YNY07,Zhang,randomHo,randomlocations,ho2008,linearNC}. In fact, as we will argue in Section~\ref{sec:linear},
using non-linear network codes (e.g. \textit{decoding} at intermediate nodes) is often \textit{needed} to achieve capacity in the scenarios studied in this paper.

\begin{definition}
\label{def:nc}
Let $\mN=(\mV,\mE, S, \bfT)$ be a network and $\mA$ an alphabet. A \textbf{network code} for  $(\mN,\mA)$ is a family $\mF=\{\mF_V \st V \in \mV \setminus (\{S\} \cup {\bf T})\}$ of functions,
where $$\mF_V : \mA^{\degin(V)} \to \mA^{\degout(V)} \quad \mbox{for all $V \in \mV \setminus (\{S\} \cup {\bf T})$}.$$
\end{definition}

A network code $\mF$ fully specifies how the intermediate nodes of a network process information packets. Note that the interpretation of each function $\mF_V$ is unique precisely thanks to the choice of the total order $\le$; see Notation~\ref{not:ext}.

\begin{example}
\label{ex:adv}
Consider the network depicted in Figure~\ref{easyex}, consisting of one source, one terminal, and one intermediate node. The edges are ordered according to their indices.
\begin{figure}[h!]
\centering
\begin{tikzpicture}
\tikzset{vertex/.style = {shape=circle,draw,inner sep=0pt,minimum size=1.9em}}
\tikzset{nnode/.style = {shape=circle,fill=myg,draw,inner sep=0pt,minimum
size=1.9em}}
\tikzset{edge/.style = {->,> = stealth}}
\tikzset{dedge/.style = {densely dotted,->,> = stealth}}
\tikzset{ddedge/.style = {dashed,->,> = stealth}}

\node[vertex] (S1) {$S$};

\node[shape=coordinate,right=\mynodespace*0.5 of S1] (K) {};

\node[nnode,right=\mynodespace of K] (V) {$V$};

\node[vertex,right=\mynodespace*0.5 of V3] (T) {$T$};

\draw[edge,bend left=25] (S1)  to node[fill=white, inner sep=3pt]{\small $e_1$} (V);

\draw[edge,bend right=25] (S1)  to node[fill=white, inner sep=3pt]{\small $e_3$} (V);

\draw[edge,bend right=0] (S1)  to node[fill=white, inner sep=3pt]{\small $e_2$} (V);

\draw[edge,bend left=0] (V)  to node[fill=white, inner sep=3pt]{\small $e_{4}$} (T);
\end{tikzpicture} 
\caption{Network for Example~\ref{ex:adv}.\label{easyex}}
\end{figure}
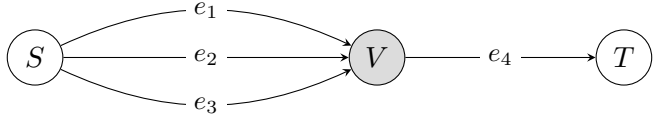
The way vertex $V$ processes information is fully specified by a function $\mF_V:\mA^3 \to \mA$, thanks to the choice of the total order.
For example, if $\mA=\F_5$ and $\mF_V(x_1,x_2,x_3)=x_1+2x_2+3x_3$
for all~$(x_1,x_2,x_3) \in \mA^3$, then vertex $V$ sends over edge $e_4$ the field element obtained by summing the field element collected on edge $e_1$ with twice the field element collected on edge $e_2$ and three times the field element collected on edge $e_3$.
\end{example}

This paper focuses on networks $\mN=(\mV,\mE, S, \bfT)$ affected by \textit{potentially restricted} adversarial noise. More precisely, we assume that at most $t$ of the alphabet symbols on a given edge set $\mU \subseteq \mE$ can be changed into any other alphabet symbols. We are interested in computing the largest number of packets that can be multicasted to all terminals in a single channel use.
As already mentioned,
the notion of error probability does not make sense in this context, since the noise is adversarial in nature.
We can make the described problem 
rigorous with the aid of the following notation, which connects networks and network codes to the notion of (adversarial) channels introduced in
 Subsection~\ref{subsect:channels}.

\begin{notation} \label{not:netch}
Let $\mN=(\mV,\mE, S, \bfT)$ be a network, $\mA$ an alphabet, $T \in \bd{T}$ a terminal, $\mF$ a network code for $(\mN,\mA)$, $\mU \subseteq \mE$ an edge set, and $t \ge 0$ an integer. We denote by
$$\Omega[\mN, \mA, \mF, S \to T,\mU,t] : \mA^{\degout(S)} \dashrightarrow \mA^{\degin(T)}$$
the channel representing the transfer from $S$ to terminal $T \in \bd{T}$, when the network code $\mF$ is used by the vertices and at most $t$ packets from the edges in $\mU$ are corrupted. In this context, we call $t$ the \textbf{adversarial power}.
\end{notation}
The following example illustrates how to formally describe the channel introduced in Notation~\ref{not:netch}. 
\begin{example}
\label{ex:ad}
Let $\mN$ be the network in Figure~\ref{fig:ad} and $\mA$ be an alphabet. We consider an adversary capable of corrupting up to one of the dashed edges, that is, one of the edges in~$\mU =\{e_1,e_2,e_3\}$. 
Let $\mF_{V_1}: \mA \to \mA$ be the identity function and let $\mF_{V_2}: \mA \to \mA$ be a function returning a constant value $a \in \mA$. 
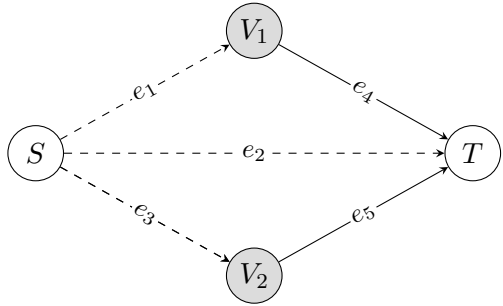
\begin{figure}[htbh]
\centering
\begin{tikzpicture}
\tikzset{vertex/.style = {shape=circle,draw,inner sep=0pt,minimum size=1.9em}}
\tikzset{nnode/.style = {shape=circle,fill=myg,draw,inner sep=0pt,minimum
size=1.9em}}
\tikzset{edge/.style = {->,> = stealth}}
\tikzset{dedge/.style = {densely dotted,->,> = stealth}}
\tikzset{ddedge/.style = {dashed,->,> = stealth}}

\node[vertex] (S1) {$S$};

\node[shape=coordinate,right=\mynodespace of S1] (K) {};

\node[nnode,above=0.5\mynodespace of K] (V1) {$V_1$};

\node[nnode,below=0.5\mynodespace of K] (V2) {$V_2$};

\node[vertex,right=\mynodespace of K] (T) {$T$};

\draw[ddedge,bend left=0] (S1)  to node[sloped,fill=white, inner sep=1pt]{\small $e_1$} (V1);

\draw[ddedge,bend left=0] (S1) to  node[sloped,fill=white, inner sep=1pt]{\small $e_2$} (T);

\draw[ddedge,bend right=0] (S1)  to node[sloped,fill=white, inner sep=1pt]{\small $e_3$} (V2);

\draw[ddedge,bend right=0] (S1)  to node[sloped,fill=white, inner sep=1pt]{\small $e_3$} (V2);

\draw[edge,bend left=0] (V1)  to node[sloped,fill=white, inner sep=1pt]{\small $e_4$} (T);

\draw[edge,bend left=0] (V2)  to node[sloped,fill=white, inner sep=1pt]{\small $e_{5}$} (T);

\end{tikzpicture} 
\caption{{{Network for Example~\ref{ex:ad}.}}}\label{fig:ad}
\end{figure}
This scenario is fully modeled by the channel
$\Omega[\mN, \mA, \mF, S \to T,\mU,1] : \mA^{3} \dashrightarrow \mA^{2}$,
which we now describe.
 For $x=(x_1,x_2,x_3) \in \mA^3$, we have that
$\Omega[\mN, \mA, \mF, S \to T,\mU,1](x)$  is the set of all alphabet vectors $y=(y_1,y_2,a) \in \mA^3$ for which 
$\dH((y_1,y_2),(x_1,x_2)) \le 1$, where~$\dH$ denotes the Hamming distance; see~\cite{macwilliams1977theory}.
\end{example}

We are finally ready to give a rigorous definition for the 1-shot capacity of a network. This is the main quantity we are concerned with in this paper.

\begin{definition} \label{def:capacities}
Let $\mN=(\mV,\mE, S, \bfT)$ be a network, $\mA$ an alphabet,
$\mU \subseteq \mE$ an edge set,
and~$t \ge 0$ an integer. 
The (\textbf{1-shot}) \textbf{capacity}
of $(\mN,\mA,\mU,t)$
is the largest 
real number~$\kappa$ for which there exists an \textbf{outer code} $$\mC \subseteq \mA^{\degout(S)}$$
and a network code $\mF$ for~$(\mN,\mA)$ 
with $\kappa=\log_{|\mA|}(|\mC|)$ such that
$\mC$ is unambiguous for each channel
$\Omega[\mN,\mA,\mF,S \to T,\mU,t]$, $T \in \bd{T}$. The notation for this largest $\kappa$ is
$$\CC_1(\mN,\mA,\mU,t).$$
The elements of the outer code $\mC$
are called \textbf{codewords}.
\end{definition}

The following bound, which is not sharp in general, is an immediate consequence of the definitions. 

\begin{proposition} \label{prop:aux}
Following the notations of Definition \ref{def:future} and Definition~\ref{def:capacities}, we have
$$\CC_1(\mN,\mA,\mU,t) \le \, \min_T \, \max_\mF \, \CC_1(\Omega[\mN,\mA,\mF,S \to T, \mU,t]),$$
where the minimum is taken over all network terminals $T \in \bfT$ and the maximum is taken over all network codes $\mF$ for $(\mN,\mA)$. 
\end{proposition}

The main goal of this paper is to initiate the study of the quantity $\CC_1(\mN,\mA,\mU,t)$ for an arbitrary tuple $(\mN,\mA,\mU,t)$,
where $t \ge 0$ is an integer, $\mA$ is an alphabet and $\mU$ is a \textit{proper} subset of the edges of the network $\mN$.

The main difference between this work and previous work in the same field lies precisely in the restriction of the noise to the set $\mU$. Recall moreover that when all edges are vulnerable, i.e., when $\mE=\mU$, 
the problem 
of computing~$\CC_1(\mN,\mA,\mE,t)$
can be completely solved
by combining cut-set bounds
with \textit{linear} network coding and rank-metric codes (for the achievability).
More precisely, the following hold.

\begin{theorem}[\text{see \cite{SKK}}]
\label{thm:mcm}
Let $\mN=(\mV,\mE, S, \bfT)$ be a network, $\mA$ an alphabet,
and $t \ge 0$ an integer. 
Let $\mu=\min_{T \in \bfT} \mincut_{\mN}(S,T)$.
Suppose that~$\mA=\F_{q^m}$,
with~$m \ge \mu$ and $q$ sufficiently large ($q \ge |{\bf T}|-1$ suffices).
Then
$$\CC_1(\mN,\mA,\mE,t) = \max\{0, \, \mu-2t\}.$$
\end{theorem}
Moreover, it has been proven in \cite{SKK} that the capacity value $\max\{0, \, \mu-2t\}$ can be attained by taking as a network code 
$\mF$ a collection of $\F_q$-linear functions. In the case of an adversary having access to only a proper subset of the network edges, the generalization of Theorem \ref{thm:mcm} can be derived and is stated as follows.  

\begin{theorem}[Generalized Network Singleton Bound; see~\cite{RK18}] \label{sbound}
Let $\mN=(\mV,\mE, S, \bfT)$ be a network, $\mA$ an alphabet,~$\mU \subseteq \mE$,
and $t \ge 0$ an integer. 
We have
\begin{equation*}
    \CC_1(\mN,\mA,\mU,t) \le \min_{T \in \bfT} \, \min_{\mE'} \left(
|\mE'\setminus \mU| + \max\{0,|\mE' \cap \mU|-2t\} \right),
\end{equation*}
where $\mE' \subseteq \mE$ ranges over edge-cuts between $S$ and $T$.
\end{theorem}

Next, we give an example to illustrate Definition \ref{def:capacities} that also makes use of the Generalized Network Singleton Bound of Theorem \ref{sbound}. 
\begin{example}
Consider the network $\mN$ depicted in Figure \ref{fig:ad}. We will show that $$\CC_1(\mN,\mA,\mU,t)=1.$$ We choose the network code to consist of identity functions $\mF_{V_1}$ and $\mF_{V_2}$. Let $\mC$ be the $3$-times repetition code, that is,
$\mC=\{(x,x,x) \mid x \in \mA\}.$ Since at most $1$ symbol from the edges of $\mU$ is corrupted, it can easily be seen that $\mC$ is unambiguous for the channel $\Omega[\mN,\mA,\mF,S \to T,\mU,1]$. Since~$|\mC|=|\mA|$, this shows that $\CC_1(\mN,\mA,\mU,t) \ge 1$. Choosing $\mE'=\{e_1,e_2,e_3\}$ in Theorem~\ref{sbound}, we have $\CC_1(\mN,\mA,\mU,t) \le 1$, yielding the desired result.
\end{example}

\section{The Curious Case of the Diamond Network}
\label{sec:diamond}

This section is devoted to the smallest example of network that illustrates the problem we focus on in this paper. We call it the \textbf{Diamond Network}, due to its shape, and denote it by~$\mathfrak{A}_1$. The choice for the notation will become clear in Subsection~\ref{sec:families}, where we will introduce a family of networks (Family~\ref{fam:a})
of which the Diamond Network is the ``first'' member.
The Diamond Network is depicted in Figure~\ref{fig:diamond}.

\begin{figure}[htbh]
    \centering
        \centering
\begin{tikzpicture}
\tikzset{vertex/.style = {shape=circle,draw,inner sep=0pt,minimum size=1.9em}}
\tikzset{nnode/.style = {shape=circle,fill=myg,draw,inner sep=0pt,minimum
size=1.9em}}
\tikzset{edge/.style = {->,> = stealth}}
\tikzset{dedge/.style = {densely dotted,->,> = stealth}}
\tikzset{ddedge/.style = {dashed,->,> = stealth}}

\node[vertex] (S1) {$S$};

\node[shape=coordinate,right=\mynodespace of S1] (K) {};

\node[nnode,above=0.35\mynodespace of K] (V1) {$V_1$};

\node[nnode,below=0.35\mynodespace of K] (V2) {$V_2$};

\node[vertex,right=\mynodespace of K] (T) {$T$};

\draw[edge,bend left=0] (S1)  to node[sloped,fill=white, inner sep=1pt]{\small $e_1$} (V1);

\draw[edge,bend left=15] (S1) to  node[sloped,fill=white, inner sep=1pt]{\small $e_2$} (V2);


\draw[edge,bend right=15] (S1)  to node[sloped,fill=white, inner sep=1pt]{\small $e_3$} (V2);

\draw[edge,bend left=0] (V1)  to node[sloped,fill=white, inner sep=1pt]{\small $e_4$} (T);

\draw[edge,bend left=0] (V2)  to node[sloped,fill=white, inner sep=1pt]{\small $e_{5}$} (T);
\end{tikzpicture} 
    \caption{The Diamond Network $\mathfrak{A}_1$.}
    \label{fig:diamond}
\end{figure}
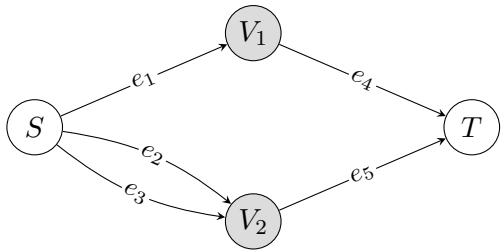

\subsection{The Capacity of the Diamond Network}

We are interested in computing the capacity of the Diamond Network $\mathfrak{A}_1$ of Figure~\ref{fig:diamond} when $\mA$ is an arbitrary alphabet,
$\mU=\{e_1,e_2,e_3\}$ is the set of vulnerable edges, and 
$t=1$.
Previously, the best known upper bound for the capacity of $\mathfrak{A}_1$ in this context was
\begin{equation} \label{boundDN}
\CC_1(\mathfrak{A}_1,\mA,\mU,1) \le 1,
\end{equation}
which follows from Theorem~\ref{sbound}. It was however shown in the preliminary work~\cite{beemer2021curious} that the upper bound
in~\eqref{boundDN} is not tight, regardless of the choice of alphabet $\mA$. 
More precisely, the following hold.

\begin{theorem}
\label{thm:diamond_cap}
For the Diamond Network $\mathfrak{A}_1$ of Figure \ref{fig:diamond}, any alphabet $\mA$, and $\mU=\{e_1,e_2,e_3\}$, we have
$$\CC_1(\mathfrak{A}_1,\mA,\mU,1) = \log_{|\mA|} \, (|\mA|-1).$$
\end{theorem}

An intuitive idea of why the capacity of the Diamond Network is strictly less than one can be seen by considering that information arriving at the terminal through $e_4$ is completely useless without information arriving through $e_5$. The opposite is also true. Thus, we must have some cooperation between the two different ``routes'' in order to achieve a positive capacity. Unfortunately, because $V_2$ has one more incoming than outgoing edge, the cooperation implicitly suggested by the Generalized Network Singleton Bound of Theorem~\ref{sbound} is impossible: a repetition code sent across $e_1$, $e_2$, and $e_3$ will fall short of guaranteed correction at the terminal. Cooperation is still possible, but it will come at a cost. To see achievability of Theorem \ref{thm:diamond_cap}, consider sending a repetition code across $e_1$, $e_2$, and $e_3$. Intermediate node~$V_1$ forwards its received symbol; $V_2$ forwards if the two incoming symbols match, and sends a special reserved~``alarm'' symbol if they do not. The terminal looks to see whether the alarm symbol was sent across~$e_5$. If so, it trusts $e_4$. If not, it trusts $e_5$. The (necessary) sacrifice of one alphabet symbol to be used as an alarm, or locator of the adversary, results in a rate~of $\log_{|\mA|} (|\mA|-1)$. A proof that this is the best rate possible was first presented in~\cite{beemer2021curious}, and the result is also shown in a new, more general way in Section \ref{sec:upper}.

Interestingly, when we add an additional edge to the Diamond Network, resulting in the so-called \textbf{Mirrored Diamond Network} $\mathfrak{D}_1$ of Figure~\ref{fig:mirrored}, the Generalized Network Singleton Bound of Theorem~\ref{sbound} becomes tight (with the
analogous adversarial action). More precisely, the following holds. It should be noted that the case of adding an extra incoming edge to $V_2$ of Figure \ref{fig:diamond} is covered by Corollary \ref{cor:conf}.

\begin{theorem}
For the Mirrored Diamond Network $\mathfrak{D}_1$ of Figure \ref{fig:mirrored}, any network alphabet $\mA$, and~$\mU=\{e_1,e_2,e_3,e_4\}$, we have
$$\CC_1(\mathfrak{D}_1,\mA,\mU,1) = 1.$$
\end{theorem}

By Theorem~\ref{sbound}, the previous result may be shown by simply exhibiting an explicit scheme achieving the upper bound of 1. We send a repetition code across $e_1$, $e_2$, $e_3$, and $e_4$. Each of~$V_1$ and~$V_2$ forwards if the two incoming symbols match, and sends a reserved alarm symbol if they do not. The terminal trusts the edge without the alarm symbol; if both send the alarm symbol, the terminal decodes to that symbol. Notice that we again make use of an alarm symbol, but that this symbol could also be sent as easily as any other alphabet symbol. This marks a striking difference from the alarm symbol used in the achievability of the Diamond Network capacity, which is instead sacrificed.

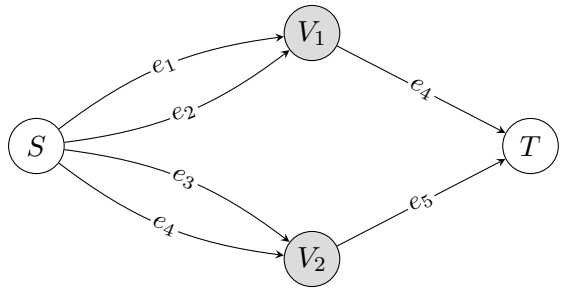
\begin{figure}[htbh]
    \centering
        \centering
\begin{tikzpicture}
\tikzset{vertex/.style = {shape=circle,draw,inner sep=0pt,minimum size=1.9em}}
\tikzset{nnode/.style = {shape=circle,fill=myg,draw,inner sep=0pt,minimum
size=1.9em}}
\tikzset{edge/.style = {->,> = stealth}}
\tikzset{dedge/.style = {densely dotted,->,> = stealth}}
\tikzset{ddedge/.style = {dashed,->,> = stealth}}

\node[vertex] (S1) {$S$};

\node[shape=coordinate,right=\mynodespace*1.3 of S1] (K) {};

\node[nnode,above=0.45\mynodespace of K] (V1) {$V_1$};

\node[nnode,below=0.45\mynodespace of K] (V2) {$V_2$};

\node[vertex,right=\mynodespace of K] (T) {$T$};

\draw[edge,bend left=15] (S1)  to node[sloped,fill=white, inner sep=1pt]{\small $e_1$} (V1);

\draw[edge,bend right=15] (S1)  to node[sloped,fill=white, inner sep=1pt]{\small $e_2$} (V1);

\draw[edge,bend left=15] (S1) to  node[sloped,fill=white, inner sep=1pt]{\small $e_3$} (V2);


\draw[edge,bend right=15] (S1)  to node[sloped,fill=white, inner sep=1pt]{\small $e_4$} (V2);

\draw[edge,bend left=0] (V1)  to node[sloped,fill=white, inner sep=1pt]{\small $e_4$} (T);

\draw[edge,bend left=0] (V2)  to node[sloped,fill=white, inner sep=1pt]{\small $e_{5}$} (T);
\end{tikzpicture} 
    \caption{The Mirrored Diamond Network $\mathfrak{D}_1$.}
    \label{fig:mirrored}
\end{figure}

The example of the Mirrored Diamond Network of Figure \ref{fig:mirrored} indicates that the \textit{bottleneck} vertex $V_2$ of the Diamond Network of Figure \ref{fig:diamond} does not tell the whole story about whether the Network Singleton Bound is achievable. Instead, something more subtle is occurring with the manner in which information ``streams'' are split within the network. One may naturally wonder about the impact of adding additional edges to $V_1$ and/or $V_2$; we leave an exploration of a variety of such families to later sections. We next look at the scenario of random noise to build further intuition.

\subsection{Information Theory Intuition}
\label{sec:info}

Partial information-theoretic intuition for the fact that the Generalized Network Singleton Bound of Theorem \ref{sbound} cannot be achieved in some cases when an adversary is restricted to a particular portion of the network can be seen even in the case of random noise (rather than adversarial). In this subsection, we will briefly consider the standard information-theoretic definition of capacity. That is, capacity will be defined as the supremum of rates for which an asymptotically vanishing decoding error probability (as opposed to zero error) is achievable.
We do not include all the fundamental information theory definitions, instead referring the reader to e.g.~\cite{coverthomas} for more details.

\begin{example}
\label{ex:info_thy}
Consider a unicast network with a single terminal and one intermediate node as illustrated in Figure \ref{fig:inf-thy-1a}. Suppose that the first three edges of the network experience random, binary symmetric noise. That is, in standard information theory notation and terminology, each dashed edge indicates a Binary Symmetric Channel with transition probability $p$, denoted BSC($p$), while the two edges from the intermediate node to the terminal are noiseless.
In the sequel, we also let
$$H(p)=-p\log_2(p)-(1-p)\log_2(1-p)$$ denote the entropy function associated with the transition probability~$p$.

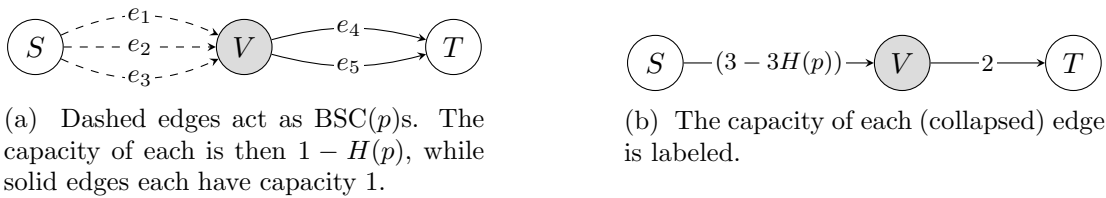
\begin{figure}[hbtp]
    \centering
\begin{subfigure}{.4\textwidth}
\centering
\begin{tikzpicture}
\tikzset{vertex/.style = {shape=circle,draw,inner sep=0pt,minimum size=1.9em}}
\tikzset{nnode/.style = {shape=circle,fill=myg,draw,inner sep=0pt,minimum
size=1.9em}}
\tikzset{edge/.style = {->,> = stealth}}
\tikzset{dedge/.style = {densely dotted,->,> = stealth}}
\tikzset{ddedge/.style = {dashed,->,> = stealth}}

\node[vertex] (S1) {$S$};

\node[nnode,right=\mynodespace*0.8 of S1] (V1) {$V$};

\node[vertex,right=\mynodespace*0.6 of K] (T) {$T$};

\draw[ddedge,bend left=0] (S1)  to node[sloped,fill=white, inner sep=1pt]{\small $e_2$} (V1);

\draw[ddedge,bend left=25] (S1) to  node[sloped,fill=white, inner sep=1pt]{\small $e_1$} (V1);


\draw[ddedge,bend right=25] (S1)  to node[sloped,fill=white, inner sep=1pt]{\small $e_3$} (V1);

\draw[edge,bend left=15] (V1)  to node[sloped,fill=white, inner sep=1pt]{\small $e_4$} (T);

\draw[edge,bend right=15] (V1)  to node[sloped,fill=white, inner sep=1pt]{\small $e_{5}$} (T);
\end{tikzpicture} 
\caption{\label{fig:inf-thy-1a} Dashed edges act as BSC($p$)s. The capacity of each is then $1-H(p)$, while solid edges each have capacity~$1$.}
\end{subfigure}
\hspace{0.1\textwidth}
\begin{subfigure}{.4\textwidth}
\centering
\begin{tikzpicture}
\tikzset{vertex/.style = {shape=circle,draw,inner sep=0pt,minimum size=1.9em}}
\tikzset{nnode/.style = {shape=circle,fill=myg,draw,inner sep=0pt,minimum
size=1.9em}}
\tikzset{edge/.style = {->,> = stealth}}
\tikzset{dedge/.style = {densely dotted,->,> = stealth}}
\tikzset{ddedge/.style = {dashed,->,> = stealth}}

\node[vertex] (S1) {$S$};

\node[nnode,right=\mynodespace of S1] (V1) {$V$};

\node[vertex,right=\mynodespace*0.6 of K] (T) {$T$};

\draw[edge,bend left=0] (S1)  to node[sloped,fill=white, inner sep=1pt]{\footnotesize $(3-3H(p))$} (V1);

\draw[edge,bend left=0] (V1)  to node[sloped,fill=white, inner sep=1pt]{\footnotesize $2$} (T);
\end{tikzpicture}
\caption{\label{fig:inf-thy-1b} The capacity of each (collapsed) edge is labeled.}
\end{subfigure}
\caption{\label{fig:inf-thy-1} A network with multi-edges, along with its simplified version with collapsed multi-edges labeled with their capacities.}
\end{figure}

 Each of the multi-edge-sets $\{e_1,e_2,e_3\}$ and $\{e_4,e_5\}$ can then be considered as collapsed to a single edge of capacity $3(1-H(p))$ and $2$, respectively (see Figure \ref{fig:inf-thy-1b}). Recall that the Max-Flow Min-Cut Theorem (see e.g. \cite[Theorem 15.2]{elgamalkim}) 
 states that the capacity of the network is equal to the minimum over all edge-cuts of the network of the sum of the capacities of the edges in the cut. Thus, the capacity of our network is equal to $\min\{2,3-3H(p)\}$.

Next, we split the intermediate node into two nodes, as in Figure \ref{fig:inf-thy-2}. Again making use of the Max-Flow Min-Cut Theorem, the new network's capacity is equal to $$\min\{1,1-H(p)\}+\min\{1,2(1-H(p))\}=1-H(p)+\min\{1,2-2H(p)\}.$$ One can easily determine that this value is upper bounded by the capacity of the network in Figure~\ref{fig:inf-thy-1} for all $0\leq p\leq 0.5$. Furthermore, when $0< H(p) < 0.5$ (i.e., when $p$ is positive and less than approximately $0.11$), this bound is strict. In other words, splitting the intermediate node reduces capacity for an interval of small transition probabilities. 

\begin{figure}[hbtp]
    \centering
\begin{tikzpicture}
\tikzset{vertex/.style = {shape=circle,draw,inner sep=0pt,minimum size=1.9em}}
\tikzset{nnode/.style = {shape=circle,fill=myg,draw,inner sep=0pt,minimum
size=1.9em}}
\tikzset{edge/.style = {->,> = stealth}}
\tikzset{dedge/.style = {densely dotted,->,> = stealth}}
\tikzset{ddedge/.style = {dashed,->,> = stealth}}

\node[vertex] (S1) {$S$};

\node[shape=coordinate,right=\mynodespace of S1] (K) {};

\node[nnode,above=0.25\mynodespace of K] (V1) {$V_1$};

\node[nnode,below=0.25\mynodespace of K] (V2) {$V_2$};

\node[vertex,right=\mynodespace of K] (T) {$T$};

\draw[ddedge,bend left=0] (S1)  to node[sloped,fill=white, inner sep=1pt]{\small $e_1$} (V1);

\draw[ddedge,bend left=15] (S1) to  node[sloped,fill=white, inner sep=1pt]{\small $e_2$} (V2);


\draw[ddedge,bend right=15] (S1)  to node[sloped,fill=white, inner sep=1pt]{\small $e_3$} (V2);

\draw[edge,bend left=0] (V1)  to node[sloped,fill=white, inner sep=1pt]{\small $e_4$} (T);

\draw[edge,bend left=0] (V2)  to node[sloped,fill=white, inner sep=1pt]{\small $e_{5}$} (T);
\end{tikzpicture} 
\caption{\label{fig:inf-thy-2} Vertex $V$ of Figure \ref{fig:inf-thy-1a} is split into two intermediate nodes.}
\end{figure}
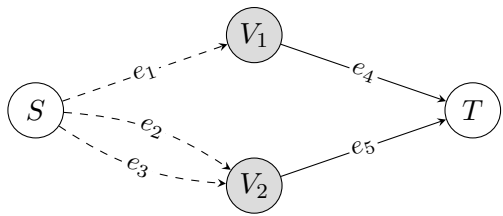

As a first natural generalization, suppose the network of Figure \ref{fig:inf-thy-1} had $n+1$ edges from source to the intermediate node, $n$ edges from intermediate node to terminal, and that the~(extension of the) network shown in Figure~\ref{fig:inf-thy-2} peeled off just a single edge from each layer, resulting in~$\degin(V_1)=\degout(V_1)=1$, $\degin(V_2)=n-1$, and $\degout(V_2)=n-2$. Then the capacity gap between the first and second networks would be non-zero for all $0<H(p)<1/(n-1)$. This gap is illustrated for $n \in \{3,5,7\}$ in Figure~\ref{fig:inf-thy-3}. Denote by ``Scenario 1'' the original network for the given value of~$n$, and by ``Scenario 2'' the corresponding network with split intermediate node. In Section~\ref{sec:families} we return to this generalization with adversarial as opposed to random noise; there, it is termed Family~\ref{ex:s}.

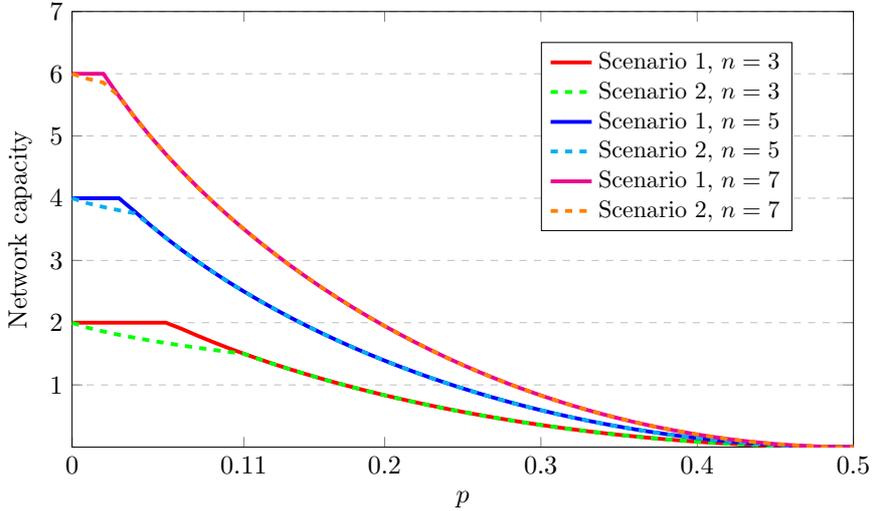
\begin{figure}[h]
\centering
\begin{tikzpicture}[scale=0.9]
\begin{axis}[legend style={at={(0.6,0.93)}, anchor = north west},
legend cell align={left},
width=13cm,height=8cm,
xlabel={$p$},
ylabel={Network capacity},
xmin=0, xmax=0.5,
ymin=0, ymax=7,
xtick={0,0.11,0.2,0.3, 0.4, 0.5},
ytick={1,2,3,4,5,6,7},
ymajorgrids=true,
grid style=dashed,
every axis plot/.append style={thick}, yticklabel style={/pgf/number format/fixed}
]
\addplot[color=red,style={ultra thick}]
coordinates { ( 0.0, 2 )
( 0.01, 2.0 )
( 0.02, 2.0 )
( 0.03, 2.0 )
( 0.04, 2.0 )
( 0.05, 2.0 )
( 0.06, 2.0 )
( 0.07, 1.902229047 )
( 0.08, 1.793462429 )
( 0.09, 1.690590549 )
( 0.1, 1.593013219 )
( 0.11, 1.500252126 )
( 0.12, 1.411917404 )
( 0.13, 1.327685445 )
( 0.14, 1.247283565 )
( 0.15, 1.170479086 )
( 0.16, 1.097071336 )
( 0.17, 1.026885664 )
( 0.18, 0.959768863 )
( 0.19, 0.89558562 )
( 0.2, 0.834215715 )
( 0.21, 0.77555178 )
( 0.22, 0.719497491 )
( 0.23, 0.665966089 )
( 0.24, 0.614879162 )
( 0.25, 0.566165627 )
( 0.26, 0.519760883 )
( 0.27, 0.475606091 )
( 0.28, 0.433647568 )
( 0.29, 0.393836261 )
( 0.3, 0.356127302 )
( 0.31, 0.320479625 )
( 0.32, 0.286855627 )
( 0.33, 0.255220882 )
( 0.34, 0.225543885 )
( 0.35000000000000003, 0.197795834 )
( 0.36, 0.171950432 )
( 0.37, 0.147983722 )
( 0.38, 0.125873933 )
( 0.39, 0.105601354 )
( 0.4, 0.087148217 )
( 0.41000000000000003, 0.070498594 )
( 0.42, 0.055638315 )
( 0.43, 0.042554888 )
( 0.44, 0.031237436 )
( 0.45, 0.021676638 )
( 0.46, 0.013864684 )
( 0.47000000000000003, 0.007795234 )
( 0.48, 0.003463392 )
( 0.49, 0.000865675 )
( 0.5, 0.0 )
};
\addplot[color=green,dashed,style={ultra thick}]
coordinates {( 0.0, 2 )
( 0.01, 1.919206864 )
( 0.02, 1.858559457 )
( 0.03, 1.805608142 )
( 0.04, 1.757707811 )
( 0.05, 1.713603043 )
( 0.06, 1.672555081 )
( 0.07, 1.634076349 )
( 0.08, 1.59782081 )
( 0.09, 1.563530183 )
( 0.1, 1.531004406 )
( 0.11, 1.500084042 )
( 0.12, 1.411917404 )
( 0.13, 1.327685445 )
( 0.14, 1.247283565 )
( 0.15, 1.170479086 )
( 0.16, 1.097071336 )
( 0.17, 1.026885664 )
( 0.18, 0.959768863 )
( 0.19, 0.89558562 )
( 0.2, 0.834215715 )
( 0.21, 0.77555178 )
( 0.22, 0.719497491 )
( 0.23, 0.665966089 )
( 0.24, 0.614879162 )
( 0.25, 0.566165627 )
( 0.26, 0.519760883 )
( 0.27, 0.475606091 )
( 0.28, 0.433647568 )
( 0.29, 0.393836261 )
( 0.3, 0.356127302 )
( 0.31, 0.320479625 )
( 0.32, 0.286855627 )
( 0.33, 0.255220882 )
( 0.34, 0.225543885 )
( 0.35000000000000003, 0.197795834 )
( 0.36, 0.171950432 )
( 0.37, 0.147983722 )
( 0.38, 0.125873933 )
( 0.39, 0.105601354 )
( 0.4, 0.087148217 )
( 0.41000000000000003, 0.070498594 )
( 0.42, 0.055638315 )
( 0.43, 0.042554888 )
( 0.44, 0.031237436 )
( 0.45, 0.021676638 )
( 0.46, 0.013864684 )
( 0.47000000000000003, 0.007795234 )
( 0.48, 0.003463392 )
( 0.49, 0.000865675 )
( 0.5, 0.0 )
};

\addplot[color=blue,style={ultra thick}]
coordinates {
( 0.0, 4 )
( 0.01, 4.0 )
( 0.02, 4.0 )
( 0.03, 4.0 )
( 0.04, 3.788539055 )
( 0.05, 3.568015214 )
( 0.06, 3.362775404 )
( 0.07, 3.170381745 )
( 0.08, 2.989104049 )
( 0.09, 2.817650915 )
( 0.1, 2.655022032 )
( 0.11, 2.500420209 )
( 0.12, 2.353195674 )
( 0.13, 2.212809075 )
( 0.14, 2.078805942 )
( 0.15, 1.950798476 )
( 0.16, 1.828452227 )
( 0.17, 1.711476106 )
( 0.18, 1.599614771 )
( 0.19, 1.492642701 )
( 0.2, 1.390359526 )
( 0.21, 1.2925863 )
( 0.22, 1.199162485 )
( 0.23, 1.109943482 )
( 0.24, 1.024798603 )
( 0.25, 0.943609378 )
( 0.26, 0.866268138 )
( 0.27, 0.792676819 )
( 0.28, 0.722745947 )
( 0.29, 0.656393768 )
( 0.3, 0.593545504 )
( 0.31, 0.534132708 )
( 0.32, 0.478092711 )
( 0.33, 0.425368136 )
( 0.34, 0.375906475 )
( 0.35000000000000003, 0.329659723 )
( 0.36, 0.286584054 )
( 0.37, 0.246639537 )
( 0.38, 0.209789889 )
( 0.39, 0.176002257 )
( 0.4, 0.145247028 )
( 0.41000000000000003, 0.117497656 )
( 0.42, 0.092730525 )
( 0.43, 0.070924814 )
( 0.44, 0.052062394 )
( 0.45, 0.03612773 )
( 0.46, 0.023107806 )
( 0.47000000000000003, 0.012992057 )
( 0.48, 0.00577232 )
( 0.49, 0.001442791 )
( 0.5, 0.0 )
};

\addplot[color=cyan,dashed,style={ultra thick}]
coordinates {( 0.0, 4 )
( 0.01, 3.919206864 )
( 0.02, 3.858559457 )
( 0.03, 3.805608142 )
( 0.04, 3.757707811 )
( 0.05, 3.568015214 )
( 0.06, 3.362775404 )
( 0.07, 3.170381745 )
( 0.08, 2.989104049 )
( 0.09, 2.817650915 )
( 0.1, 2.655022032 )
( 0.11, 2.500420209 )
( 0.12, 2.353195674 )
( 0.13, 2.212809075 )
( 0.14, 2.078805942 )
( 0.15, 1.950798476 )
( 0.16, 1.828452227 )
( 0.17, 1.711476106 )
( 0.18, 1.599614771 )
( 0.19, 1.492642701 )
( 0.2, 1.390359526 )
( 0.21, 1.2925863 )
( 0.22, 1.199162485 )
( 0.23, 1.109943482 )
( 0.24, 1.024798603 )
( 0.25, 0.943609378 )
( 0.26, 0.866268138 )
( 0.27, 0.792676819 )
( 0.28, 0.722745947 )
( 0.29, 0.656393768 )
( 0.3, 0.593545504 )
( 0.31, 0.534132708 )
( 0.32, 0.478092711 )
( 0.33, 0.425368136 )
( 0.34, 0.375906475 )
( 0.35000000000000003, 0.329659723 )
( 0.36, 0.286584054 )
( 0.37, 0.246639537 )
( 0.38, 0.209789889 )
( 0.39, 0.176002257 )
( 0.4, 0.145247028 )
( 0.41000000000000003, 0.117497656 )
( 0.42, 0.092730525 )
( 0.43, 0.070924814 )
( 0.44, 0.052062394 )
( 0.45, 0.03612773 )
( 0.46, 0.023107806 )
( 0.47000000000000003, 0.012992057 )
( 0.48, 0.00577232 )
( 0.49, 0.001442791 )
( 0.5, 0.0 )
};

\addplot[color=magenta,style={ultra thick}]
coordinates {( 0.0, 6 )
( 0.01, 6.0 )
( 0.02, 6.0 )
( 0.03, 5.639256995 )
( 0.04, 5.303954676 )
( 0.05, 4.9952213 )
( 0.06, 4.707885566 )
( 0.07, 4.438534444 )
( 0.08, 4.184745669 )
( 0.09, 3.944711281 )
( 0.1, 3.717030845 )
( 0.11, 3.500588293 )
( 0.12, 3.294473943 )
( 0.13, 3.097932705 )
( 0.14, 2.910328319 )
( 0.15, 2.731117867 )
( 0.16, 2.559833118 )
( 0.17, 2.396066549 )
( 0.18, 2.23946068 )
( 0.19, 2.089699781 )
( 0.2, 1.946503336 )
( 0.21, 1.80962082 )
( 0.22, 1.678827479 )
( 0.23, 1.553920875 )
( 0.24, 1.434718044 )
( 0.25, 1.321053129 )
( 0.26, 1.212775393 )
( 0.27, 1.109747547 )
( 0.28, 1.011844326 )
( 0.29, 0.918951276 )
( 0.3, 0.830963705 )
( 0.31, 0.747785791 )
( 0.32, 0.669329796 )
( 0.33, 0.595515391 )
( 0.34, 0.526269065 )
( 0.35000000000000003, 0.461523612 )
( 0.36, 0.401217675 )
( 0.37, 0.345295351 )
( 0.38, 0.293705844 )
( 0.39, 0.24640316 )
( 0.4, 0.203345839 )
( 0.41000000000000003, 0.164496719 )
( 0.42, 0.129822735 )
( 0.43, 0.09929474 )
( 0.44, 0.072887351 )
( 0.45, 0.050578822 )
( 0.46, 0.032350928 )
( 0.47000000000000003, 0.01818888 )
( 0.48, 0.008081248 )
( 0.49, 0.002019908 )
( 0.5, 0.0 )
};

\addplot[color=orange,dashed,style={ultra thick}]
coordinates {( 0.0, 6 )
( 0.01, 5.919206864 )
( 0.02, 5.858559457 )
( 0.03, 5.639256995 )
( 0.04, 5.303954676 )
( 0.05, 4.9952213 )
( 0.06, 4.707885566 )
( 0.07, 4.438534444 )
( 0.08, 4.184745669 )
( 0.09, 3.944711281 )
( 0.1, 3.717030845 )
( 0.11, 3.500588293 )
( 0.12, 3.294473943 )
( 0.13, 3.097932705 )
( 0.14, 2.910328319 )
( 0.15, 2.731117867 )
( 0.16, 2.559833118 )
( 0.17, 2.396066549 )
( 0.18, 2.23946068 )
( 0.19, 2.089699781 )
( 0.2, 1.946503336 )
( 0.21, 1.80962082 )
( 0.22, 1.678827479 )
( 0.23, 1.553920875 )
( 0.24, 1.434718044 )
( 0.25, 1.321053129 )
( 0.26, 1.212775393 )
( 0.27, 1.109747547 )
( 0.28, 1.011844326 )
( 0.29, 0.918951276 )
( 0.3, 0.830963705 )
( 0.31, 0.747785791 )
( 0.32, 0.669329796 )
( 0.33, 0.595515391 )
( 0.34, 0.526269065 )
( 0.35000000000000003, 0.461523612 )
( 0.36, 0.401217675 )
( 0.37, 0.345295351 )
( 0.38, 0.293705844 )
( 0.39, 0.24640316 )
( 0.4, 0.203345839 )
( 0.41000000000000003, 0.164496719 )
( 0.42, 0.129822735 )
( 0.43, 0.09929474 )
( 0.44, 0.072887351 )
( 0.45, 0.050578822 )
( 0.46, 0.032350928 )
( 0.47000000000000003, 0.01818888 )
( 0.48, 0.008081248 )
( 0.49, 0.002019908 )
( 0.5, 0.0 )
};
\legend{\small{Scenario 1, $n=3$}, \small{Scenario 2, $n=3$},\small{Scenario 1, $n=5$},\small{Scenario 2, $n=5$},\small{Scenario 1, $n=7$},\small{Scenario 2, $n=7$}}
\end{axis}
\end{tikzpicture}
\caption{\label{fig:inf-thy-3} Capacity gaps between the first generalized networks for $n \in \{3,5,7\}$.}
\end{figure}

A second possible generalization is as follows: suppose the network of Figure \ref{fig:inf-thy-1} had $3n$ edges from source to intermediate node, and $2n$ edges from intermediate node to terminal, and that the network of Figure \ref{fig:inf-thy-2} peeled off $n$ edges from each layer so that~$\degin(V_1)=\degout(V_1)=n$, $\degin(V_2)=2n$, and $\degout(V_2)=n$. Interestingly, the capacity gap between the first and second networks would be non-zero for all $0<H(p)<0.5$, regardless of the value of $n$. This is illustrated for~$n \in \{3,5,7\}$ in Figure~\ref{fig:inf-thy-4}. Denote by ``Scenario 1'' the original network for the given value of~$n$, and by ``Scenario 2'' the corresponding network with split intermediate node. In Section~\ref{sec:families} we return to this generalization with adversarial as opposed to random noise; there, it is termed Family~\ref{fam:a}.
\end{example}

\begin{figure}[h!]
\centering
\begin{tikzpicture}[scale=0.9]
\begin{axis}[legend style={at={(0.6,0.93)}, anchor = north west},
legend cell align={left},
width=13cm,height=8cm,
xlabel={$p$},
ylabel={Network capacity},
xmin=0, xmax=0.5,
ymin=0, ymax=16,
xtick={0,0.11,0.2,0.3, 0.4, 0.5},
ytick={2,4,6,8,10,12,14,16},
ymajorgrids=true,
grid style=dashed,
every axis plot/.append style={thick}, yticklabel style={/pgf/number format/fixed}
]
\addplot[color=red,style={ultra thick}]
coordinates { ( 0.0, 6 )
( 0.01, 6.0 )
( 0.02, 6.0 )
( 0.03, 6.0 )
( 0.04, 6.0 )
( 0.05, 6.0 )
( 0.06, 6.0 )
( 0.07, 5.706687142 )
( 0.08, 5.380387288 )
( 0.09, 5.071771646 )
( 0.1, 4.779039658 )
( 0.11, 4.500756377 )
( 0.12, 4.235752212 )
( 0.13, 3.983056335 )
( 0.14, 3.741850695 )
( 0.15, 3.511437258 )
( 0.16, 3.291214008 )
( 0.17, 3.080656991 )
( 0.18, 2.879306588 )
( 0.19, 2.686756861 )
( 0.2, 2.502647146 )
( 0.21, 2.326655341 )
( 0.22, 2.158492473 )
( 0.23, 1.997898268 )
( 0.24, 1.844637486 )
( 0.25, 1.69849688 )
( 0.26, 1.559282648 )
( 0.27, 1.426818274 )
( 0.28, 1.300942705 )
( 0.29, 1.181508783 )
( 0.3, 1.068381907 )
( 0.31, 0.961438875 )
( 0.32, 0.86056688 )
( 0.33, 0.765662645 )
( 0.34, 0.676631655 )
( 0.35000000000000003, 0.593387502 )
( 0.36, 0.515851297 )
( 0.37, 0.443951166 )
( 0.38, 0.3776218 )
( 0.39, 0.316804063 )
( 0.4, 0.26144465 )
( 0.41000000000000003, 0.211495781 )
( 0.42, 0.166914945 )
( 0.43, 0.127664665 )
( 0.44, 0.093712309 )
( 0.45, 0.065029914 )
( 0.46, 0.041594051 )
( 0.47000000000000003, 0.023385703 )
( 0.48, 0.010390176 )
( 0.49, 0.002597024 )
( 0.5, 0.0 )
};
\addplot[color=green,dashed,style={ultra thick}]
coordinates {
( 0.0, 6 )
( 0.01, 5.757620592 )
( 0.02, 5.575678372 )
( 0.03, 5.416824427 )
( 0.04, 5.273123433 )
( 0.05, 5.140809129 )
( 0.06, 5.017665243 )
( 0.07, 4.902229047 )
( 0.08, 4.793462429 )
( 0.09, 4.690590549 )
( 0.1, 4.593013219 )
( 0.11, 4.500252126 )
( 0.12, 4.235752212 )
( 0.13, 3.983056335 )
( 0.14, 3.741850695 )
( 0.15, 3.511437258 )
( 0.16, 3.291214008 )
( 0.17, 3.080656991 )
( 0.18, 2.879306588 )
( 0.19, 2.686756861 )
( 0.2, 2.502647146 )
( 0.21, 2.326655341 )
( 0.22, 2.158492473 )
( 0.23, 1.997898268 )
( 0.24, 1.844637486 )
( 0.25, 1.69849688 )
( 0.26, 1.559282648 )
( 0.27, 1.426818274 )
( 0.28, 1.300942705 )
( 0.29, 1.181508783 )
( 0.3, 1.068381907 )
( 0.31, 0.961438875 )
( 0.32, 0.86056688 )
( 0.33, 0.765662645 )
( 0.34, 0.676631655 )
( 0.35000000000000003, 0.593387502 )
( 0.36, 0.515851297 )
( 0.37, 0.443951166 )
( 0.38, 0.3776218 )
( 0.39, 0.316804063 )
( 0.4, 0.26144465 )
( 0.41000000000000003, 0.211495781 )
( 0.42, 0.166914945 )
( 0.43, 0.127664665 )
( 0.44, 0.093712309 )
( 0.45, 0.065029914 )
( 0.46, 0.041594051 )
( 0.47000000000000003, 0.023385703 )
( 0.48, 0.010390176 )
( 0.49, 0.002597024 )
( 0.5, 0.0 )
};

\addplot[color=blue,style={ultra thick}]
coordinates {

( 0.0, 10 )
( 0.01, 10.0 )
( 0.02, 10.0 )
( 0.03, 10.0 )
( 0.04, 10.0 )
( 0.05, 10.0 )
( 0.06, 10.0 )
( 0.07, 9.511145236 )
( 0.08, 8.967312147 )
( 0.09, 8.452952744 )
( 0.1, 7.965066096 )
( 0.11, 7.501260628 )
( 0.12, 7.059587021 )
( 0.13, 6.638427225 )
( 0.14, 6.236417825 )
( 0.15, 5.852395429 )
( 0.16, 5.48535668 )
( 0.17, 5.134428319 )
( 0.18, 4.798844314 )
( 0.19, 4.477928102 )
( 0.2, 4.171078577 )
( 0.21, 3.877758901 )
( 0.22, 3.597487456 )
( 0.23, 3.329830447 )
( 0.24, 3.074395809 )
( 0.25, 2.830828133 )
( 0.26, 2.598804413 )
( 0.27, 2.378030457 )
( 0.28, 2.168237842 )
( 0.29, 1.969181305 )
( 0.3, 1.780636512 )
( 0.31, 1.602398124 )
( 0.32, 1.434278134 )
( 0.33, 1.276104408 )
( 0.34, 1.127719425 )
( 0.35000000000000003, 0.988979169 )
( 0.36, 0.859752161 )
( 0.37, 0.73991861 )
( 0.38, 0.629369667 )
( 0.39, 0.528006772 )
( 0.4, 0.435741083 )
( 0.41000000000000003, 0.352492969 )
( 0.42, 0.278191574 )
( 0.43, 0.212774442 )
( 0.44, 0.156187182 )
( 0.45, 0.10838319 )
( 0.46, 0.069323418 )
( 0.47000000000000003, 0.038976171 )
( 0.48, 0.01731696 )
( 0.49, 0.004328374 )
( 0.5, 0.0 )
};

\addplot[color=cyan,dashed,style={ultra thick}]
coordinates {
( 0.0, 10 )
( 0.01, 9.596034321 )
( 0.02, 9.292797287 )
( 0.03, 9.028040711 )
( 0.04, 8.788539055 )
( 0.05, 8.568015214 )
( 0.06, 8.362775404 )
( 0.07, 8.170381745 )
( 0.08, 7.989104049 )
( 0.09, 7.817650915 )
( 0.1, 7.655022032 )
( 0.11, 7.500420209 )
( 0.12, 7.059587021 )
( 0.13, 6.638427225 )
( 0.14, 6.236417825 )
( 0.15, 5.852395429 )
( 0.16, 5.48535668 )
( 0.17, 5.134428319 )
( 0.18, 4.798844314 )
( 0.19, 4.477928102 )
( 0.2, 4.171078577 )
( 0.21, 3.877758901 )
( 0.22, 3.597487456 )
( 0.23, 3.329830447 )
( 0.24, 3.074395809 )
( 0.25, 2.830828133 )
( 0.26, 2.598804413 )
( 0.27, 2.378030457 )
( 0.28, 2.168237842 )
( 0.29, 1.969181305 )
( 0.3, 1.780636512 )
( 0.31, 1.602398124 )
( 0.32, 1.434278134 )
( 0.33, 1.276104408 )
( 0.34, 1.127719425 )
( 0.35000000000000003, 0.988979169 )
( 0.36, 0.859752161 )
( 0.37, 0.73991861 )
( 0.38, 0.629369667 )
( 0.39, 0.528006772 )
( 0.4, 0.435741083 )
( 0.41000000000000003, 0.352492969 )
( 0.42, 0.278191574 )
( 0.43, 0.212774442 )
( 0.44, 0.156187182 )
( 0.45, 0.10838319 )
( 0.46, 0.069323418 )
( 0.47000000000000003, 0.038976171 )
( 0.48, 0.01731696 )
( 0.49, 0.004328374 )
( 0.5, 0.0 )
};

\addplot[color=magenta,style={ultra thick}]
coordinates {
( 0.0, 14 )
( 0.01, 14.0 )
( 0.02, 14.0 )
( 0.03, 14.0 )
( 0.04, 14.0 )
( 0.05, 14.0 )
( 0.06, 14.0 )
( 0.07, 13.31560333 )
( 0.08, 12.55423701 )
( 0.09, 11.83413384 )
( 0.1, 11.15109253 )
( 0.11, 10.50176488 )
( 0.12, 9.883421829 )
( 0.13, 9.293798114 )
( 0.14, 8.730984956 )
( 0.15, 8.193353601 )
( 0.16, 7.679499353 )
( 0.17, 7.188199646 )
( 0.18, 6.71838204 )
( 0.19, 6.269099342 )
( 0.2, 5.839510007 )
( 0.21, 5.428862461 )
( 0.22, 5.036482438 )
( 0.23, 4.661762626 )
( 0.24, 4.304154133 )
( 0.25, 3.963159386 )
( 0.26, 3.638326178 )
( 0.27, 3.32924264 )
( 0.28, 3.035532978 )
( 0.29, 2.756853827 )
( 0.3, 2.492891116 )
( 0.31, 2.243357374 )
( 0.32, 2.007989388 )
( 0.33, 1.786546172 )
( 0.34, 1.578807196 )
( 0.35000000000000003, 1.384570837 )
( 0.36, 1.203653026 )
( 0.37, 1.035886054 )
( 0.38, 0.881117533 )
( 0.39, 0.739209481 )
( 0.4, 0.610037516 )
( 0.41000000000000003, 0.493490156 )
( 0.42, 0.389468204 )
( 0.43, 0.297884219 )
( 0.44, 0.218662054 )
( 0.45, 0.151736466 )
( 0.46, 0.097052785 )
( 0.47000000000000003, 0.05456664 )
( 0.48, 0.024243744 )
( 0.49, 0.006059723 )
( 0.5, 0.0 )
};

\addplot[color=orange,dashed,style={ultra thick}]
coordinates {
( 0.0, 14 )
( 0.01, 13.43444805 )
( 0.02, 13.0099162 )
( 0.03, 12.639257 )
( 0.04, 12.30395468 )
( 0.05, 11.9952213 )
( 0.06, 11.70788557 )
( 0.07, 11.43853444 )
( 0.08, 11.18474567 )
( 0.09, 10.94471128 )
( 0.1, 10.71703084 )
( 0.11, 10.50058829 )
( 0.12, 9.883421829 )
( 0.13, 9.293798114 )
( 0.14, 8.730984956 )
( 0.15, 8.193353601 )
( 0.16, 7.679499353 )
( 0.17, 7.188199646 )
( 0.18, 6.71838204 )
( 0.19, 6.269099342 )
( 0.2, 5.839510007 )
( 0.21, 5.428862461 )
( 0.22, 5.036482438 )
( 0.23, 4.661762626 )
( 0.24, 4.304154133 )
( 0.25, 3.963159386 )
( 0.26, 3.638326178 )
( 0.27, 3.32924264 )
( 0.28, 3.035532978 )
( 0.29, 2.756853827 )
( 0.3, 2.492891116 )
( 0.31, 2.243357374 )
( 0.32, 2.007989388 )
( 0.33, 1.786546172 )
( 0.34, 1.578807196 )
( 0.35000000000000003, 1.384570837 )
( 0.36, 1.203653026 )
( 0.37, 1.035886054 )
( 0.38, 0.881117533 )
( 0.39, 0.739209481 )
( 0.4, 0.610037516 )
( 0.41000000000000003, 0.493490156 )
( 0.42, 0.389468204 )
( 0.43, 0.297884219 )
( 0.44, 0.218662054 )
( 0.45, 0.151736466 )
( 0.46, 0.097052785 )
( 0.47000000000000003, 0.05456664 )
( 0.48, 0.024243744 )
( 0.49, 0.006059723 )
( 0.5, 0.0 )
};
\legend{\small{Scenario 1, $n=3$}, \small{Scenario 2, $n=3$},\small{Scenario 1, $n=5$},\small{Scenario 2, $n=5$},\small{Scenario 1, $n=7$},\small{Scenario 2, $n=7$}}
\end{axis}
\end{tikzpicture}
\caption{\label{fig:inf-thy-4} Capacity gaps between the second generalized networks for $n \in \{3,5,7\}$.}
\end{figure}
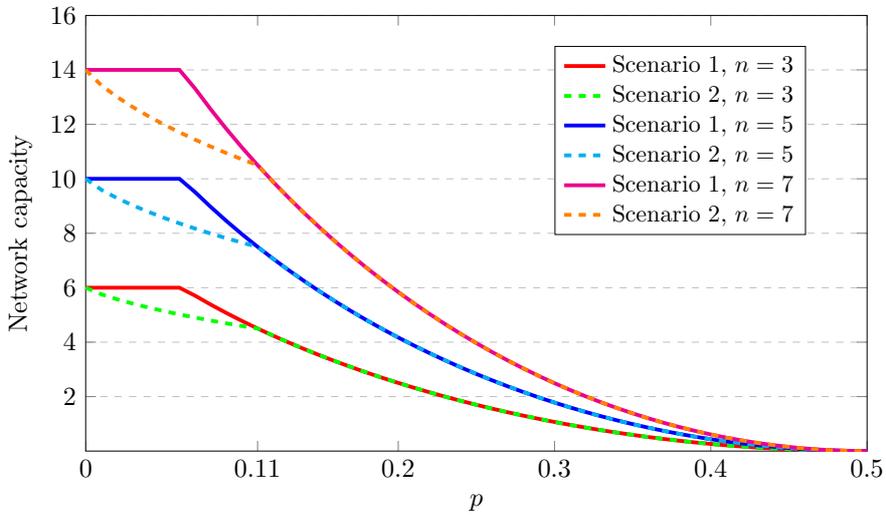

In both generalizations, we see that the intermediate node split in Example \ref{ex:info_thy} prevents edges $\{e_{1},e_{2},e_{3}\}$ from cooperating to send messages when the BSC transition probability is small: in the random noise scenario, the difference is due to a lower-capacity mixed-vulnerability edge-cut being present in the split network. This gap is mirrored by the gap in the 1-shot capacity we observe in adversarial networks with restricted adversaries, though the reason for the gap differs. In the case of a restricted adversary on these two networks who may corrupt up to one vulnerable edge (with no random noise), the Generalized Network Singleton Bound of Theorem~\ref{sbound} is achievable for the network in Scenario 1, while it is not achievable for the network in Scenario~2 (see Theorem~\ref{thm:diamond_cap}). For higher adversarial power, as for higher transition probability in the case of random noise, the two have matching capacities. In the case of adversarial noise, the difference in capacities for limited adversarial power is due to something beyond edge-cut differences, which are already baked into the Generalized Network Singleton Bound. Capacities with restricted adversaries are no longer additive, and so we must preserve the split structure by looking beyond both the Max-Flow Min-Cut Theorem and the Generalized Network Singleton Bound in order to establish improved upper bounds on capacity.

\section{Networks with Two and Three Levels}
\label{sec:net-2-and-3}

In this section we focus on families of networks having 2 or 3 levels, a property which is defined formally below. Throughout this section, we assume basic graph theory knowledge; see e.g.~\cite{west2001introduction}. We show that one can upper bound the capacity of a special class of 3-level networks by the capacity of a corresponding 2-level network. We then define five families of 2-level networks, which will be key players of this paper and whose capacities will be computed or estimated in later sections. 

In Section \ref{sec:double-cut-bd} we will show how the results of this section can be ``ported'' to arbitrary networks using a general method that 
describes the information transfer from one edge-set to another; see in particular the Double-Cut-Set Bound of Theorem~\ref{thm:dcsb}.
All of this will also allow us to 
compute the capacity of the network
that opened the paper.

\subsection{$m$-Level Networks}

\begin{definition}
\label{def:n-level}
Let $\mN=(\mV,\mE,S,\mathbf{T})$ be a network and let $V,V' \in \mV$. We say that $V'$ \textbf{covers}~$V$ if
$(V,V') \in \mE$. 
We call $\mN$ an \textbf{$m$-level} network 
if $\mV$ can be partitioned into $m+1$ sets $\mV_{0},\ldots,\mV_{m}$ such that $\mV_{0}=\{S\}$, $\mV_{m}=\mathbf{T}$, and each node in $\mV_{k}$,  for $k\in \{1,...,m-1\}$,
is only covered by elements of
$\mV_{k+1}$ and only covers elements of $\mV_{k-1}$. We call $\mV_k$ the \textbf{$k$-th layer} of $\mN$.
\end{definition}

Notice that in an $m$-level network, any path from $S$ to any $T\in \mathbf{T}$ is of length $m$. 
Moreover, the value of $m$ and the layers $\mV_k$, for $k \in \{0,...,m\}$,
in Definition~\ref{def:n-level} are uniquely determined by the network~$\mN$.
Many of the results of this paper rely on  particular classes of 2-level and~3-level networks, which we now define.

\begin{definition}
\label{def:special_3level}
\label{def:special_2level}
A 2-level network is 
\textbf{simple}
if it has a single terminal.
A 3-level network is 
\textbf{simple} 
if it has 
a single terminal,
each intermediate node at
distance 1 from the source has in-degree equal to~1, and each intermediate node at distance~1 from the terminal has out-degree equal to~1. \end{definition}

In order to denote 2- and 3-level networks more compactly, we will utilize (simplified) adjacency matrices of the bipartite subgraphs induced by subsequent node layers. First we present the most general notation for an $m$-level network, then discuss the particular cases of~2- and~3-level networks. 

\begin{notation} \label{notmtx}
Let $\mN_m=(\mV,\mE,S,\bfT)$ be an $m$-level network and let $\mV_0,...,\mV_m$ be as in Definition~\ref{def:n-level} (the subscript in $\mN_m$ has the sole function of stressing the number of levels).
Fix an enumeration of the elements of each $\mV_k$, $k \in \{0,...,m\}$.
We denote by $\smash{M^{m,k}}$ the matrix representing the graph induced by the nodes in layers $k-1$ and $k$ of $\mN_m$, for $k\in \{1,...,m\}$.
Specifically, $\smash{M^{m,k}}$ has dimensions $\smash{|\mV_{k-1}|\times |\mV_{k}|}$, and $\smash{M^{m,k}_{ij}=\ell}$ if and only if there are $\ell$ edges from node $i$ of $\mV_{k-1}$ to node $j$ of $\mV_{k}$. We can then denote the network by $\smash{(M^{m,1},M^{m,2},\ldots,M^{m,m})}$. It is easy to check that $\mN_m$ is uniquely determined by this representation.

In a 2-level network, we have two adjacency matrices $\smash{M^{2,1}}$ and $\smash{M^{2,2}}$, and in a 3-level network we have three adjacency matrices $\smash{M^{3,1}}$ and $\smash{M^{3,2}}$, and $\smash{M^{3,3}}$. Notice that in a simple~3-level network, $\smash{M^{3,1}}$ and $\smash{M^{3,3}}$ will always be all-ones vectors, and so we may drop them from the notation. With a slight abuse of notation, we will denote $\smash{M^{2,2}}$ as a row vector in a simple~2-level network (instead of as a column vector).
\end{notation}

We give two examples to illustrate the previous notation.

\begin{example} \label{ex:simplified}
Consider the Diamond Network of Section \ref{sec:diamond}; see Figure~\ref{fig:diamond}. Following Notation~\ref{notmtx}, we may represent this network as $\smash{([1,2],[1,1]^\top)}$. Because the network is simple, we may abuse notation and simplify this to $\smash{([1,2],[1,1])}$. Similarly, the Mirrored Diamond Network (see Figure~\ref{fig:mirrored}) may be represented as $\smash{([2,2],[1,1])}$.
\end{example}

\begin{example} \label{ex:Hexagon}
Consider the 3-level network shown in Figure \ref{fig:Hexagon}. Following Notation~\ref{notmtx}, we may represent this network as
\[ 
\left(\begin{bmatrix} 
1 & 1 & 1 & 1 & 1 & 1
\end{bmatrix},
\begin{bmatrix} 
1 & 1 & 0 & 0  \\ 
1 & 1 & 0 & 0  \\ 
1 & 1 & 0 & 0  \\ 
1 & 1 & 0 & 0 \\ 
0 & 0 & 1 & 1  \\ 
0 & 0 & 1 & 1
\end{bmatrix},
\begin{bmatrix} 
1 \\ 1 \\ 1 \\ 1  
\end{bmatrix}\right)
.\]
More simply, the simple 3-level network may be represented using only the center matrix.
\end{example}

\begin{figure}[htbp]
\centering
\scalebox{0.90}{
\begin{tikzpicture}

\tikzset{vertex/.style = {shape=circle,draw,inner sep=0pt,minimum size=1.9em}}
\tikzset{nnode/.style = {shape=circle,fill=myg,draw,inner sep=0pt,minimum
size=1.9em}}
\tikzset{edge/.style = {->,> = stealth}}
\tikzset{dedge/.style = {densely dotted,->,> = stealth}}
\tikzset{ddedge/.style = {dashed,->,> = stealth}}

\node[vertex] (S1) {$S$};

\node[shape=coordinate,right=\mynodespace of S1] (K) {};

\node[nnode,above=1.2\mynodespace of K] (V1) {$V_1$};

\node[nnode,above=0.7\mynodespace of K] (V2) {$V_2$};

\node[nnode,above=0.3\mynodespace of K] (V3) {$V_3$};

\node[nnode,below=0\mynodespace of K] (V4) {$V_4$};

\node[nnode,below=0.4\mynodespace of K] (V5) {$V_5$};

\node[nnode,below=0.8\mynodespace of K] (V6) {$V_6$};

\node[vertex,right=3.5\mynodespace of S1] (T) {$T$};

\node[nnode,right=1.3\mynodespace of V1 ] (V7) {$V_7$};

\node[nnode,right=1.3\mynodespace of V3] (V8) {$V_8$};

\node[nnode,right=1.3\mynodespace of V4] (V9) {$V_9$};

\node[nnode,right=1.3\mynodespace of V6] (V10) {$V_{10}$};

\draw[ddedge,bend left=0] (S1)  to node[sloped,fill=white, inner sep=0pt]{} (V1);

\draw[ddedge,bend right=0] (S1)  to node[sloped,fill=white, inner sep=0pt]{} (V2);

\draw[ddedge,bend left=0] (S1) to  node[sloped,fill=white, inner sep=0pt]{} (V3);

\draw[ddedge,bend right=0] (S1) to  node[sloped,fill=white, inner sep=0pt]{} (V4);

\draw[ddedge,bend right=0] (S1) to  node[sloped,fill=white, inner sep=0pt]{} (V5);

\draw[ddedge,bend right=0] (S1) to  node[sloped,fill=white, inner sep=0pt]{} (V6);

\draw[edge,bend left=0] (V1) to  node{} (V7);

\draw[edge,bend left=0] (V1) to  node{} (V8);

\draw[edge,bend left=0] (V2) to  node{} (V7);

\draw[edge,bend left=0] (V2) to  node{} (V8);

\draw[edge,bend left=0] (V3) to  node{} (V7);

\draw[edge,bend left=0] (V3) to  node{} (V8);

\draw[edge,bend left=0] (V4) to  node{} (V7);

\draw[edge,bend left=0] (V4) to  node{} (V8);

\draw[edge,bend right=0] (V5) to  node{} (V9);

\draw[edge,bend right=0] (V5) to  node{} (V10);

\draw[edge,bend right=0] (V6) to  node{} (V9);

\draw[edge,bend right=0] (V6) to  node{} (V10);

\draw[edge,bend left=0] (V7) to  node{} (T);

\draw[edge,bend left=0] (V8) to  node{} (T);

\draw[edge,bend left=0] (V9) to  node{} (T);

\draw[edge,bend left=0] (V10) to  node{} (T);

\end{tikzpicture} 

}
\caption{Network for Examples~\ref{ex:Hexagon} and \ref{ex:vulne}. \label{fig:Hexagon}}
\end{figure}
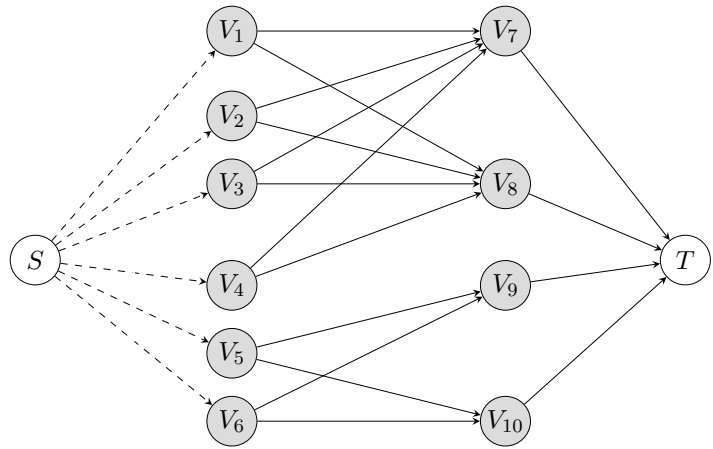

\subsection{Reduction from 3- to 2-Level Networks}
\label{sec:3to2reduc}

In this subsection we describe a procedure to obtain a simple 2-level network from a simple~3-level network. In Section \ref{sec:double-cut-bd} we will show that, under certain assumptions, the capacity of any network can be upper bounded by the capacity of a simple 3-level network constructed from it. Using the procedure described in this subsection, we will be able to upper bound
the capacity of an arbitrary network with that of an \textit{induced} simple 2-level network (obtaining sharp bounds in some cases).

Let $\mN_3$ be a simple 3-level network
defined by matrix $\smash{M^{3,2}}$, along with all-ones matrices~$\smash{M^{3,1}}$~and~$\smash{M^{3,3}}$ (see Notation~\ref{notmtx}). 
We construct a simple
2-level network $\mN_2$, defined via~$\smash{M^{2,1}}$ and~$\smash{M^{2,2}}$ as follows. Consider the bipartite graph~$\smash{G^{3,2}}$ corresponding to adjacency matrix~$\smash{M^{3,2}}$; if~$\smash{G^{3,2}}$ has $\ell$ connected components, then let~$\smash{M^{2,1}}$ and~$\smash{M^{2,2}}$ both have dimensions~$1\times \ell$ (where we are considering the simplified representation for a simple 2-level network; see Example~\ref{ex:simplified}). Let~$\smash{M^{2,1}_{1i}=a}$ if and only if the $i$th connected component of~$\smash{G^{3,2}}$ has~$a$ vertices in~$\mV_{1}$, and let~$\smash{M^{2,2}_{1i}=b}$ if and only if the $i$th connected component of~$\smash{G^{3,2}}$ has~$b$ vertices in~$\mV_{2}$.
Observe that the sum of the entries of~$\smash{M^{2,1}}$ is equal to the sum of the entries of~$\smash{M^{3,1}}$, and similarly with~$\smash{M^{2,2}}$ and~$\smash{M^{3,3}}$.

\begin{definition}\label{def:associated}
We call the network $\mN_2$ constructed above 
the 2-level network \textbf{associated} with the 3-level network $\mN_3$.
\end{definition}

\begin{example} 
Consider the network of Figure \ref{fig:Hexagon}. The corresponding 2-level network is depicted in Figure~\ref{fig:Hexagon-2level}.
\end{example}

\begin{figure}[htbp]
\centering

\begin{tikzpicture}
\tikzset{vertex/.style = {shape=circle,draw,inner sep=0pt,minimum size=1.9em}}
\tikzset{nnode/.style = {shape=circle,fill=myg,draw,inner sep=0pt,minimum
size=1.9em}}
\tikzset{edge/.style = {->,> = stealth}}
\tikzset{dedge/.style = {densely dotted,->,> = stealth}}
\tikzset{ddedge/.style = {dashed,->,> = stealth}}

\node[vertex] (S1) {$S$};

\node[shape=coordinate,right=\mynodespace of S1] (K) {};

\node[nnode,above=0.5\mynodespace of K] (V1) {$V_1$};
\node[nnode,below=0.5\mynodespace of K] (V2) {$V_2$};

\node[vertex,right=2\mynodespace of S1] (T) {$T$};

\draw[ddedge,bend left=10] (S1)  to node[sloped,fill=white, inner sep=0pt]{$e_2$} (V1);
\draw[ddedge,bend right=10] (S1)  to node[sloped,fill=white, inner sep=0pt]{$e_3$} (V1);
\draw[ddedge,bend left=30] (S1) to  node[sloped,fill=white, inner sep=0pt]{$e_1$} (V1);
\draw[ddedge,bend right=30] (S1) to  node[sloped,fill=white, inner sep=0pt]{$e_4$} (V1);

\draw[edge,bend right=15] (V1) to  node[sloped,fill=white, inner sep=0pt]{$e_8$} (T);
\draw[edge,bend left=15] (V1) to  node[sloped,fill=white, inner sep=0pt]{$e_7$} (T);

\draw[ddedge,bend right=15] (S1) to  node[sloped,fill=white, inner sep=0pt]{$e_6$} (V2);
\draw[ddedge,bend left=15] (S1) to  node[sloped, fill=white, inner sep=0pt]{$e_5$} (V2);

\draw[edge,bend right=15] (V2) to  node[sloped,fill=white, inner sep=0pt]{$e_{10}$} (T);
\draw[edge,bend left=15] (V2) to  node[sloped,fill=white, inner sep=0pt]{$e_9$} (T);
\end{tikzpicture} 

\caption{{{The simple 2-level network corresponding to the (simple, 3-level) network 
of Example~\ref{ex:Hexagon}.}}}\label{fig:Hexagon-2level}
\end{figure}
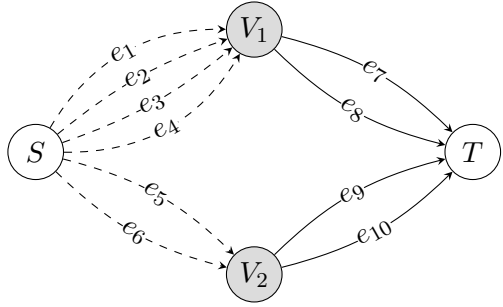

While this deterministic process results in a unique simple 2-level network given a simple~3-level network, there may be multiple 3-level networks that result in the same 2-level network. 
This however does not affect the following statement, which gives an upper bound for the capacity of a 3-level network 
in terms of the capacity of the corresponding 2-level network, when the vulnerable edges are in both cases those directly connected with the source. While the argument 
extends to more generalized choices of the vulnerable edges, in this paper we concentrate on this simplified scenario for ease of exposition.


\begin{theorem}
\label{thm:channel}
Let $\mN_3$ be a simple 3-level network, and let $\mN_2$ be the simple 2-level network associated to it.
Let $\mU_3$ and $\mU_2$ be the set of edges directly connected to the sources of $\mN_3$ and~$\mN_2$, respectively. Then for all network alphabets $\mA$ and for all $t \ge 0$ we have 
\[\CC_1(\mN_3,\mA,\mU_3,t)\leq \CC_1(\mN_2,\mA,\mU_2,t).\]
\end{theorem}

\begin{example} \label{ex:vulne}
Consider the network of Figure \ref{fig:Hexagon} and
the corresponding 2-level network in Figure~\ref{fig:Hexagon-2level}.
Suppose that the vulnerable edges in both networks are those 
directly connected to the source,
and in both cases we allow up to $t$ corrupted edges.
Then Theorem~\ref{thm:channel} implies that the capacity of the network
of Figure \ref{fig:Hexagon}
is upper bounded by the capacity of the network of Figure \ref{fig:Hexagon-2level}.
\end{example}

\begin{proof}[Proof of Theorem~\ref{thm:channel}]
Let $\mC_3$ be an outer code and $\mF_3$ be a network code for $(\mN_3,\mA)$ such that~$\mC_3$ is unambiguous for the channel $\Omega[\mN_3,\mA,\mF_3,S \to T,\mU_3,t]$. 

Let $\smash{M^{3,2}}$ be the matrix defining $\smash{\mN_3}$, and let $\smash{G^{3,2}}$ denote the bipartite graph with (simplified) adjacency matrix $\smash{M^{3,2}}$. Let $\smash{V^3_{ij}}$ denote the $j$th node in the right part of the $i$th connected component of~$\smash{G^{3,2}}$, and let $\smash{\mF_{V^3_{ij}}}$ denote the function at $\smash{V^3_{ij}}$ defined by $\smash{\mF_3}$. Let the neighborhood of $\smash{V^3_{ij}}$ in~$\smash{G^{3,2}}$ contain the (ordered) set of vertices \[\smash{V^3_{ij1},\, V^3_{ij2}, \, \ldots,\, V^3_{ij\degin(V^3_{ij})}},\] where $\smash{\degin(V^3_{ij})}$ is the in-degree of $\smash{V^3_{ij}}$. Then, for each $\smash{1\leq k\leq \degin(V^3_{ij})}$, denote the network code function at $\smash{V^3_{ijk}}$ by $\smash{\mF_{V^3_{ijk}}}$. 
Notice that every $\smash{\mF_{V^3_{ijk}}}$ is a function with domain $\mA$ and codomain~$\mA^{\degout(V^3_{ijk})}$, while each function  $\smash{\mF_{V^3_{ij}}}$ has domain $\mA^{\degin(V^3_{ij})}$ and codomain $\mA$.
Note that every node in the left part of~$\smash{G^{3,2}}$ has a label $V^{3}_{ijk}$ for some $i$ and $j$ due to assumption~\ref{prnG} of Definition~\ref{def:network}. Each such node can have multiple labels~$\smash{V^3_{ijk}}$ and~$\smash{V^3_{ij'k'}}$ where $(j,k)\neq (j',k')$; of course, we stipulate that $$\mF_{V^3_{ijk}}=\mF_{V^3_{ij'k'}}.$$ 
We claim there exists $\mF_2$ such that $\mC_3$ is also unambiguous for $\Omega[\mN_2,\mA,\mF_2,S \to T,\mU_2,t]$. Indeed, define $\mF_2$ for each intermediate node $V_i$ in $\mN_2$ that corresponds to connected component~$i$ of~$G^{3,2}$ as an appropriate composition of functions at nodes in $\mV_1$ and $\mV_2$ of the 3-level network. More technically, we define $\mF_{V_i}$ as follows (here the product symbols denote the Cartesian product): 
\begin{align*}
    \mF_{V_i} : \mA^{\degin(V_i)} &\to \mA^{\degout(V_i)},\\
    x&\mapsto \prod_{j=1}^{\degout(V_i)}
    \mF_{V^3_{ij}}\left(\prod_{k=1}^{\degin(V^3_{ij})} \mF_{V^3_{ijk}}(x_{ijk})\vert_{V_{ij}^3}\right),
\end{align*}
where $x_{ijk}$ is the coordinate of the vector $x$ corresponding to node $\smash{V^{3}_{ijk}}$ in the left part of the~$i$th connected component of $\smash{G^{3,2}}$, and  $\smash{\mF_{V^3_{ijk}}(x_{ijk})\vert_{V_{ij}^3}}$ is the restriction of $\smash{\mF_{V^3_{ijk}}(x_{ijk})}$ to the coordinate corresponding to $\smash{V_{ij}^3}$.

We claim that the fan-out set of any $x\in \mC_3$ over the channel $\Omega[\mN_2,\mA,\mF_2,S \to T,\mU_2,t]$ is exactly equal to the fan-out set of $x$ over the channel $\Omega[\mN_3,\mA,\mF_3,S \to T,\mU_3,t]$. This follows directly from the definitions of $\mF_2$ and $\mF_3$, and the fact that both networks are corrupted in up to $t$ positions from their first layers.
Suppose then, by way of contradiction, that $\mC_3$ is not unambiguous for $\Omega[\mN_2,\mA,\mF_2,S \to T,\mU_2,t]$.  That is, there exist $x,x'\in \mC_3$ such that~$x\neq x'$ but the intersection of the fan-out sets of $x$ and $x'$ is nonempty. Then $\mC_3$ was not unambiguous for $\Omega[\mN_3,\mA,\mF_3,S \to T,\mU_3,t]$ to begin with. We conclude that $\mC_3$ is unambiguous for~$\Omega[\mN_2,\mA,\mF_2,S \to T,\mU_2,t]$.
\end{proof}

The proof above contains rather heavy notation and terminology. We therefore illustrate it with an example.

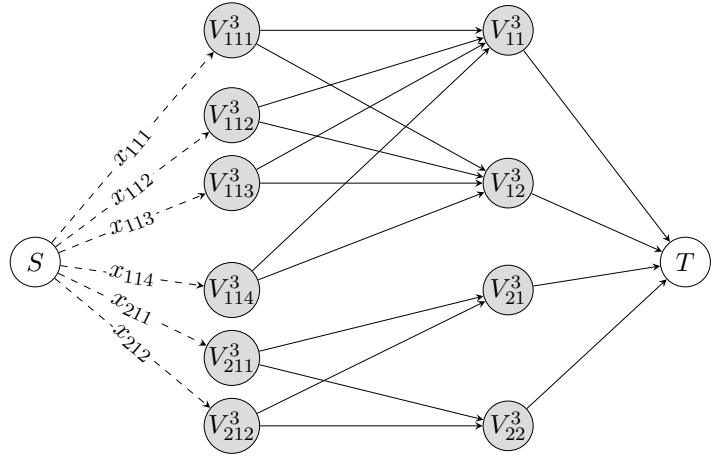
\begin{figure}[htbp]
\centering
\scalebox{0.90}{
\begin{tikzpicture}

\tikzset{vertex/.style = {shape=circle,draw,inner sep=0pt,minimum size=1.9em}}
\tikzset{nnode/.style = {shape=circle,fill=myg,draw,inner sep=0pt,minimum
size=1.9em}}
\tikzset{edge/.style = {->,> = stealth}}
\tikzset{dedge/.style = {densely dotted,->,> = stealth}}
\tikzset{ddedge/.style = {dashed,->,> = stealth}}

\node[vertex] (S1) {$S$};

\node[shape=coordinate,right=\mynodespace of S1] (K) {};

\node[nnode,above=1.2\mynodespace of K] (V1) {$V_{111}^3$};

\node[nnode,above=0.7\mynodespace of K] (V2) {$V_{112}^3$};

\node[nnode,above=0.3\mynodespace of K] (V3) {$V_{113}^3$};

\node[nnode,below=0\mynodespace of K] (V4) {$V_{114}^3$};

\node[nnode,below=0.4\mynodespace of K] (V5) {$V_{211}^3$};

\node[nnode,below=0.8\mynodespace of K] (V6) {$V_{212}^3$};

\node[vertex,right=3.5\mynodespace of S1] (T) {$T$};

\node[nnode,right=1.3\mynodespace of V1 ] (V7) {$V_{11}^3$};

\node[nnode,right=1.3\mynodespace of V3] (V8) {$V_{12}^3$};

\node[nnode,right=1.3\mynodespace of V4] (V9) {$V_{21}^3$};

\node[nnode,right=1.3\mynodespace of V6] (V10) {$V_{22}^3$};

\draw[ddedge,bend left=0] (S1)  to node[sloped,fill=white, inner sep=1pt]{$x_{111}$} (V1);

\draw[ddedge,bend right=0] (S1)  to node[sloped,fill=white, inner sep=1pt]{$x_{112}$} (V2);

\draw[ddedge,bend left=0] (S1) to  node[sloped,fill=white, inner sep=1pt]{$x_{113}$} (V3);

\draw[ddedge,bend right=0] (S1) to  node[sloped,fill=white, inner sep=1pt]{$x_{114}$} (V4);

\draw[ddedge,bend right=0] (S1) to  node[sloped,fill=white, inner sep=1pt]{$x_{211}$} (V5);

\draw[ddedge,bend right=0] (S1) to  node[sloped,fill=white, inner sep=1pt]{$x_{212}$} (V6);

\draw[edge,bend left=0] (V1) to  node{} (V7);

\draw[edge,bend left=0] (V1) to  node{} (V8);

\draw[edge,bend left=0] (V2) to  node{} (V7);

\draw[edge,bend left=0] (V2) to  node{} (V8);

\draw[edge,bend left=0] (V3) to  node{} (V7);

\draw[edge,bend left=0] (V3) to  node{} (V8);

\draw[edge,bend left=0] (V4) to  node{} (V7);

\draw[edge,bend left=0] (V4) to  node{} (V8);

\draw[edge,bend right=0] (V5) to  node{} (V9);

\draw[edge,bend right=0] (V5) to  node{} (V10);

\draw[edge,bend right=0] (V6) to  node{} (V9);

\draw[edge,bend right=0] (V6) to  node{} (V10);

\draw[edge,bend left=0] (V7) to  node{} (T);

\draw[edge,bend left=0] (V8) to  node{} (T);

\draw[edge,bend left=0] (V9) to  node{} (T);

\draw[edge,bend left=0] (V10) to  node{} (T);

\end{tikzpicture} 

}
\caption{Network for Example \ref{ex:hex}. \label{fig:Hex}}
\end{figure}

\begin{example}
\label{ex:hex}
Consider the labeling of vertices in Figure \ref{fig:Hex} of Network $\mN_3$. Suppose the capacity is achieved by a network code $$\{\mF_{V_{111}^3},
\, \mF_{V_{112}^3},
\, \mF_{V_{113}^3},
\, \mF_{V_{114}^3},
\, \mF_{V_{211}^3},
\, \mF_{V_{211}^3},
\, \mF_{V_{11}^3},
\, \mF_{V_{12}^3},
\, \mF_{V_{21}^3},
\, \mF_{V_{22}^3}\}$$ for $(\mN_3,\mA)$ and by an outer code $\mC_3\subseteq \mA^6$ unambiguous for the channel~$\Omega[\mN_3,\mA,\mF_3,S \to T,\mU_3,t]$. Let $x=(x_{111},x_{112},x_{113},x_{114},x_{211},x_{212}) \in \mC_3$ and consider the scheme in Figure \ref{fig:Hex}. The way functions $\mF_{V_1}$ and $\mF_{V_2}$ are defined in the proof of Theorem~\ref{thm:channel} gives that the alphabet symbols
\begin{gather*} 
\mF_{V_{11}^3}\big(\mF_{V_{111}^3}(x_{111})\vert_{V_{11}^3},\mF_{V_{112}^3}(x_{112})\vert_{V_{11}^3},\mF_{V_{113}^3}(x_{113})\vert_{V_{11}^3},\mF_{V_{114}^3}(x_{114})\vert_{V_{11}^3}\big), \\ 
\mF_{V_{12}^3}\big(\mF_{V_{111}^3}(x_{111})\vert_{V_{12}^3},\mF_{V_{112}^3}(x_{112})\vert_{V_{12}^3},\mF_{V_{113}^3}(x_{113})\vert_{V_{12}^3},\mF_{V_{114}^3}(x_{114})\vert_{V_{12}^3}\big), \\ 
\mF_{V_{21}^3}\big(\mF_{V_{211}^3}(x_{211})\vert_{V_{21}^3},\mF_{V_{212}^3}(x_{212})\vert_{V_{21}^3}\big), \mbox{ and } \\ 
\mF_{V_{22}^3}\big(\mF_{V_{211}^3}(x_{211})\vert_{V_{22}^3},\mF_{V_{212}^3}(x_{212})\vert_{V_{22}^3}\big)
\end{gather*}
are carried over the edges $e_7$, $e_8$, $e_9$ and $e_{10}$ respectively in Figure \ref{fig:Hexagon-2level}. Observe that the the fan-out set of any $x\in \mC_3$ over the channel $\Omega[\mN_2,\mA,\mF_2,S \to T,\mU_2,t]$ is exactly equal to the fan-out set of $x$ over the channel $\Omega[\mN_3,\mA,\mF_3,S \to T,\mU_3,t]$, as desired.
\end{example}

\subsection{Some Families of Simple 2-Level Networks}
\label{sec:families}

In this section, we introduce five families of simple 2-level networks. Thanks to Theorem~\ref{thm:channel},
any upper bound for the capacities of these translates into an upper bound for the capacities of a 3-level networks associated with them; see Definition~\ref{def:associated}.
The five families of networks introduced in this subsection should be regarded as the ``building blocks'' of the theory developed in this paper,
since in Section~\ref{sec:double-cut-bd} 
we will argue how to use them to obtain upper bounds for the capacities of larger networks.

We focus our attention on the scenario where the adversary acts on the edges directly connected to the source $S$, which we denote by $\mU_S$ throughout this section.
The families we introduce will be studied in detail in later sections, but are collected here for preparation and ease of reference. Each family is parametrized by a positive integer (denoted by $t$ or $s$ for reasons that will become clear below).

\begin{family}
\label{fam:a}
Define the simple 2-level networks
$$\mathfrak{A}_t=([t,2t],[t,t]), \quad t \ge 1.$$ 
Note that they reduce to
the Diamond Network of Section~\ref{sec:diamond} for $t=1$. The Generalized Network Singleton Bound of Theorem~\ref{sbound} reads
$\CC_1(\mathfrak{A}_t,\mA,\mU_S,t) \le t$
for any alphabet $\mA$. Results related to this family can be found in Theorem \ref{thm:meta} and Proposition \ref{prop:atleasta}.
\end{family}

\begin{family}
\label{ex:s} \label{fam:b}
Define the simple 2-level networks $$\mathfrak{B}_s=([1,s+1],[1,s]), \quad s \ge 1.$$
The case where $s=1$ yields the Diamond Network of Section~\ref{sec:diamond}. The Generalized Network Singleton Bound of Theorem~\ref{sbound} for $t=1$
reads
$\CC_1(\mathfrak{B}_s,\mA,\mU_S,1) \le s$
for any alphabet $\mA$. Note that for this family we will always take $t=1$, which explains our choice of using a different index,~$s$, for the family members. Results related to this family can be found in Theorem \ref{thm:notmet} and Corollary \ref{cor:sbs}.
\end{family}

\begin{family}
\label{ex:u} \label{fam:c}
Define the simple 2-level networks $$\mathfrak{C}_t=([t,t+1],[t,t]), \quad t \ge 2.$$ 
The case $t=1$ is covered in $\mathfrak{A}_1$ of Family \ref{fam:a} and thus formally excluded here for a reason we will explain in Remark \ref{rem:exclude}. The Generalized Network Singleton Bound of Theorem~\ref{sbound} 
reads
$\CC_1(\mathfrak{C}_t,\mA,\mU_S,t) \le 1$
for any alphabet $\mA$. Our result related to this family can be found in Theorem \ref{thm:metc}.
\end{family}

\begin{family}
\label{fam:d}
Define the simple 2-level networks $$\mathfrak{D}_t=([2t,2t],[1,1]), \quad t \ge 1.$$
The case where $t=1$ yields the Mirrored Diamond Network of Section~\ref{sec:diamond}. The Generalized Network Singleton Bound of Theorem~\ref{sbound} 
reads
$\CC_1(\mathfrak{D}_t,\mA,\mU_S,t) \le 1$
for any alphabet $\mA$. Results related to this family can be found in Theorems \ref{thm:metd} and \ref{thm:linmirr}.
\end{family}

\begin{family}
\label{fam:e}
Define the simple 2-level networks $$\mathfrak{E}_t=([t,t+1],[1,1]), \quad t \ge 1.$$
The case where $t=1$ yields the Diamond Network of Section~\ref{sec:diamond}.
The Generalized Network Singleton Bound of Theorem~\ref{sbound} 
reads
$\CC_1(\mathfrak{E}_t,\mA,\mU_S,t) \le 1$
for any alphabet $\mA$. Results related to this family can be found in Theorems \ref{thm:mete} and \ref{thm:8.4}.
\end{family}

As we will see, the results of this section and of the next show that the Generalized Network Singleton Bound of Theorem~\ref{sbound}
is \textit{never} sharp for Families~\ref{fam:a},~\ref{fam:b}, and~\ref{fam:e}, \textit{no matter} what the alphabet is.
The bound is, however, sharp for Families~\ref{fam:c} and~\ref{fam:d} under the assumption that the alphabet is a sufficiently large finite field, as we will show in Section~\ref{sec:2level_lower}.

\section{Simple 2-level Networks: Upper Bounds}
\label{sec:upper}

In this section we present upper bounds on the capacity of simple 2-level networks based on a variety of techniques.
We then apply these bounds to the 
families introduced in Subsection~\ref{sec:families}. We start by establishing the notation that we will follow in the sequel.

\begin{notation} \label{not:s5}
Throughout this section, $n \ge 2$ is an integer and $$\mN=(\mV,\mE,S,\{T\})=([a_1,\ldots,a_n],[b_1,\ldots,b_n])$$ is a simple 2-level network; see Definition~\ref{def:special_2level}.
We denote by $\mU_S \subseteq \mE$ the set of edges directly connected to the source $S$ and let  $\mA$ be a network alphabet.
We denote the intermediate nodes of $\mN$ by $V_1,\ldots,V_n$, which correspond to the structural parameters $a_1, \ldots, a_n$ and 
$b_1, \ldots, b_n$.
If~$\mF$ is a network code for $(\mN,\mA)$,
then we simply write $\mF_i$ for $\smash{\mF_{V_i}}$.
Observe moreover that the network code $\mF$ can be ``globally'' interpreted as a function
$$\mF: \mA^{a_1+\ldots +a_n} \to \mA^{b_1+\ldots +b_n},$$
although of course allowed functions $\mF$ are restricted by the topology of the underlying 2-level network.
\end{notation}

We start by spelling out the Generalized Network Singleton Bound of Theorem~\ref{sbound} specifically for simple 2-level networks.

\begin{corollary}[Generalized Network Singleton Bound for simple 2-level networks]
\label{cor:sing}
Following Notation~\ref{not:s5}, for all $t \ge 0$ we have
\[\CC_1(\mN,\mA,\mU_S,t) \leq
\min_{P_1\sqcup P_2=\{1,\ldots,n\}}\left(\sum_{i\in P_1} b_i+\max\left\{0,\sum_{i\in P_2} a_i - 2t\right\}\right),\] 
where the minimum is taken over all 2-partitions $P_1, P_2$ of the set $\{1,\ldots,n\}$.
\end{corollary}

The upper bounds we derive in this section use a ``mix'' of projection and packing arguments. We therefore continue by reminding the reader of the notions of the Hamming metric \textit{ball} and~\textit{shell}.

\begin{notation} \label{not:ballsetc}
Given an outer code $\smash{\mC \subseteq \mA^{a_1+a_2+\ldots+a_n}}$, we let $\smash{\pi_i(\mC)}$ be the projection of $\mC$ onto the $a_i$ coordinates corresponding to the edges to intermediate node $V_{i}$ of the simple 2-level network $\mN$. For example, $\smash{\pi_1(\mC)}$ is the projection onto the first $a_1$ coordinates of the codeword. Moreover, for a given~$x \in \mC$, we denote by $B^\HH_t(x)$ the \textbf{Hamming ball} of radius $t$ with center~$x$, and by $S^\HH_t(x)$ the \textbf{shell} of that ball. In symbols,
$$B^\HH_t(x)=\{y \in \mA^{a_1+a_2+\ldots+a_n} \st d^\HH(x,y) \le t\}, \qquad S^\HH_t(x)=\{y \in B^\HH_t(x) \st d^\HH(x,y) = t\},$$
where $d^\HH$ denotes the usual Hamming distance.
\end{notation}

The next observation focuses on the case $n=2$ to illustrate an idea that will be generalized immediately after in Remark \ref{rem:id} below.

\begin{remark} \label{rem:packing}
If $n=2$, then for all $t \ge 0$ 
an outer code $\mC \subseteq \mA^{a_1+a_2}$ is unambiguous for the channel $\Omega[\mN,\mA,\mF,S\to T,\mU_S,t]$ if and only if
$$\mF(x+e)=(\mF_1(\pi_1(x+e)), \mF_2(\pi_2(x+e))) \neq (\mF'_1(\pi_1(x'+e')), \mF_2(\pi_2(x'+e'))) =\mF(x'+e')$$ for all 
$x,x' \in \mC$ with $x \neq x'$ and for all $e,e' \in \mA^{a_1+a_2}$ of Hamming weight at most $t$. 
Therefore via a packing argument 
we obtain
\begin{equation}
\label{eqn:packing_2}
    \sum_{x \in \mC} |\mF\left(B^\HH_t(x)\right)| \le {|\mA|}^{b_1+b_2}.
\end{equation}
\end{remark}

We will work towards extending the packing bound idea outlined in Remark~\ref{rem:packing} to higher numbers of intermediate nodes using the properties of simple 2-level networks. We start with 
the following result.

\begin{lemma}
\label{lem:id}
Following Notation~\ref{not:s5}, suppose $n=2$ and $b_1 \ge a_1$.
Let $t \ge 0$ and suppose that~$\mC \subseteq \mA^{a_1+a_2}$ is unambiguous for the channel $\Omega[\mN,\mA,\mF,S\to T,\mU_S,t]$. Let $\mF'_1:\mA^{a_1} \to \mA^{b_1}$ be any injective map. Then $\mC$ is unambiguous for the channel $\Omega[\mN,\mA,\{\mF_1',\mF_2\},S\to T,\mU_S,t]$ as well.
\end{lemma}

\begin{proof}
Let $\mF=\{\mF_1,\mF_2\}$ and $\mF'=\{\mF'_1,\mF_2\}$. Towards a contradiction and using Remark~\ref{rem:packing}, suppose that
there are $x,x' \in \mC$ with $x \neq x'$ and $e,e' \in \mA^{a_1+a_2}$ of Hamming weight at most $t$ with~$\mF'(x+e)=\mF'(x'+e')$. This implies
$$(\mF'_1(\pi_1(x+e)), \mF_2(\pi_2(x+e)))=
(\mF'_1(\pi_1(x'+e')), \mF_2(\pi_2(x'+e'))).$$
Since $\mF'_1$ is injective, we have
$\pi_1(x+e) = \pi_1(x'+e')$ and thus
$$(\mF_1(\pi_1(x+e)), \mF_2(\pi_2(x+e)))=
(\mF_1(\pi_1(x'+e')), \mF_2(\pi_2(x'+e'))).$$
That is,
$\mF(x+e)=\mF(x'+e')$, a contradiction.
\end{proof}

\begin{remark}
\label{rem:id}
Lemma \ref{lem:id} extends easily to an arbitrary number of intermediate nodes, showing that whenever~$b_i \ge a_i$ for some $i$, without loss of generality we can assume $b_i=a_i$ and take the corresponding function $\mF_i$ to be the identity (ignoring extraneous outgoing edges). In other words, the capacity obtained by maximizing over all possible choices of $\mF$ is the same as the capacity obtained by maximizing over the $\mF$'s where $\mF_i$ is equal to the identity whenever $b_i \ge a_i$ (again, where some edges can be disregarded). This will simplify the analysis of the network families we study. In particular,
we can reduce the study of a 
simple 2-level network~$\mN=([a_1,...,a_n],[b_1,...,b_n])$ 
as in Notation~\ref{not:s5} to the study of
$$\mN'=([a_1,\ldots,a_r,a_{r+1},\ldots,a_n],[a_1,\ldots,a_r,b_{r+1},\ldots,b_n]),$$
where, up to a permutation of the vertices and edges,
\begin{equation*}
    r=\max\{i \st a_i \le b_i \mbox{ for all } 1 \le i \le r\}.
\end{equation*}
\end{remark}

The next result combines the observation
of Remark~\ref{rem:id} with a packing argument to derive an upper bound on the capacity of certain simple 2-level networks.

\begin{theorem}[First Packing Bound]
\label{thm:down} 
Following Notation~\ref{not:s5},
suppose that $a_i \le b_i$ for all $1\leq i \leq r$. Let $t \ge 0$ and let $\mC$ be an unambigious code for $\Omega[\mN,\mA,\mF,S \to T,\mU_S,t]$. Then
$$\sum_{\substack{t_1,\ldots,t_r \ge 0 \\ t_1+\ldots+t_r \le t}} \, \prod_{i=1}^r\binom{a_i}{t_i}(|\mA|-1)^{t_i} \, \sum_{x \in \mC} \, \prod_{j=r+1}^n \left|\mF_j\left(B^\HH_{t-(t_1+\ldots+t_r)}(\pi_j(x))\right)\right| \le |\mA|^{b_1+b_2+\ldots +b_n}.$$
\end{theorem}

\begin{proof} Let $\mF=\{\mF_1,\ldots,\mF_r,\ldots,\mF_n\}$ be a network code for $(\mN,\mA)$, where the first $r$ functions correspond to the pairs~$(a_1,b_1),\ldots,(a_r,b_r)$. By Lemma \ref{lem:id}, we shall assume without loss of generality that $\mF_i$ 
is an injective map for $1 \le i \le r.$ By Remark \ref{rem:id}, we can assume them to be identity by ignoring the extraneous outgoing edges. For $x \in \mC$, we have 
$$B^\HH_t(x) = \bigsqcup_{t_1+\ldots +t_n \le t} \left[ S^\HH_{t_1}(\pi_1(x)) \times \cdots \times 
S^\HH_{t_n}(\pi_n(x)) \right],$$
where $\sqcup$ emphasizes that the union is disjoint. Then,
{\small
\begin{align*}
\mF(B^\HH_t(x)) &= \bigcup_{t_1+\ldots +t_n \le t} \mF \left[S^\HH_{t_1}(\pi_1(x)) \times \cdots \times 
S^\HH_{t_n}(\pi_n(x)) \right]\\
&= \bigcup_{t_1+\ldots +t_n \le t} \left[S^\HH_{t_1}(\pi_1(x)) \times \cdots \times 
S^\HH_{t_r}(\pi_r(x)) \times \mF_{r+1}(S^\HH_{t_{r+1}}(\pi_{r+1}(x)) \times \cdots \times \mF_{n}(S^\HH_{t_{n}}(\pi_{n}(x)) \right]\\
&= \bigcup_{t_1 + \ldots +t_r \le t} \  \bigcup_{t_{r+1}+\ldots+t_n\le t - (t_1 + \ldots +t_r)}  \Bigl[ S^\HH_{t_1}(\pi_1(x)) \times \cdots \times S^\HH_{t_r}(\pi_r(x)) \, \times \\ & \qquad \qquad 
\qquad 
\times \mF_{r+1}(S^\HH_{t_{r+1}}(\pi_{r+1}(x))  
\times \cdots \times  \mF_{n}(S^\HH_{t_{n}}(\pi_{n}(x))\Bigr].
\end{align*}
}
The first union in the last equality is disjoint due to~$\mF_1,\ldots,\mF_r$ being the identity and the fact that we are considering the shells.
So,
{\small
\begin{align*}
    \mF(B^\HH_t(x)) &= \bigsqcup_{t_1+\ldots+t_r \le t} \Bigl[S^\HH_{t_1}(\pi_1(x)) \times \cdots \times S^\HH_{t_r}(\pi_r(x)) \times \\ &
    \qquad \qquad \qquad 
    \bigcup_{t_{r+1}+\ldots+t_n \le t - (t_1 + \ldots + t_r)}\biggl( \mF_{r+1}(S^\HH_{t_{r+1}}(\pi_{r+1}(x))  
\times \cdots \times \mF_{n}(S^\HH_{t_{n}}(\pi_{n}(x))\biggr)\Bigr].
\end{align*}
}
By taking cardinalities in the previous identity 
we obtain
\[
    |\mF(B^\HH_t(x))|= \sum_{ t_1 + \ldots +t_r \le t} \, \prod_{i=1}^r\binom{a_i}{t_i}(|\mA|-1)^{t_i} \prod_{j=r+1}^n |\mF_j(B^\HH_{t-(t_1 + \ldots + t_r)}(\pi_j(x)))|,\]
where the terms in the second product come from considering that all shells up to the shell of radius $t-(t_1 + \ldots + t_r)$ will be included for each projection. Summing over $x \in \mC$, exchanging summations, and using the same argument as in \eqref{eqn:packing_2}, we obtain the desired result.
\end{proof}

In this paper, we 
are mainly interested in the case $(n,r)=(2,1)$ of the previous result. 
For convenience of the reader, we state this as a corollary of Theorem~\ref{thm:down}.

\begin{corollary}
\label{cor:ub}
Following Notation~\ref{not:s5}, suppose $n=2$ and $a_1\le b_1$.
Let $t \ge 1$ and let
$\mC$ be an unambiguous code for the channel $\Omega[\mN,\mA,\mF,S \to T,\mU_S,t]$. Then
$$\sum_{t_1=0}^t
\binom{a_1}{t_1}(|\mA|-1)^{t_1} \sum_{x \in \mC} \left|\mF_2\left(B^\HH_{t-t_1}(\pi_2(x)))\right)\right| \le |\mA|^{b_1+b_2}.$$
\end{corollary}

We next apply the bound of Theorem~\ref{thm:down} (in the form of Corollary~\ref{cor:ub}) to some of the network families 
introduced in Subsection~\ref{sec:families}.
With some additional information, Theorem~\ref{thm:down} gives us that the Generalized Network Singleton Bound of Corollary~\ref{cor:sing} is in general not achievable, no matter what the network alphabet~$\mA$~is. This is in strong contrast with what is commonly observed in classical coding theory, where the Singleton bound is always the sharpest (and in fact sharp) when the alphabet is sufficiently large.

\begin{theorem}
\label{thm:notmet}
Let $\mathfrak{B}_s=(\mV,\mE,S,\{T\})$ be a member of
Family~\ref{ex:s}. Let $\mA$ be any network alphabet and let $\mU_S$ be the  set of edges of $\mathfrak{B}_s$ directly connected to $S$.
We have
$$\CC_1(\mathfrak{B}_s,\mA,\mU_S,1) <s.$$
In particular, the Generalized Network Singleton Bound of Corollary~\ref{cor:sing} is not met.
\end{theorem}

\begin{proof} 
Let $\mF$ be a network code for the pair $(\mathfrak{B}_s,\mA)$ and let
$q:=|\mA|$ to simplify the notation.
Suppose that $\mC$ is an unambiguous code for the channel 
$\Omega[\mathfrak{B}_s,\mA,\mF,S \to T,\mU_S,1]$. We want to show that $|\mC|<q^s$.
Corollary \ref{cor:ub} gives
\begin{equation} \label{touse}
    (q-1) \cdot |\mC| + \sum_{x \in \mC} \left|\mF_2\left(B^\HH_1(\pi_2(x))\right)\right| \le q^{s+1}.
\end{equation}
Suppose towards a contradiction that $|\mC|=q^s$.
We claim that there exists a codeword $x \in \mC$ for which the cardinality of 
$$\left\{\mF_2(\pi_2(z)) \st \dH(\pi_2(x),\pi_2(z))\leq 1,\ z\in \mA^{s+2}\right\}$$ 
is at least 2. To see this, we start by observing that $\mF_2$ restricted to ${\pi_2(\mC)}$ must be injective, as otherwise the fan-out sets of at least two codewords would intersect non-trivially, making~$\mC$ ambiguous for the channel $\Omega[\mathfrak{B}_s,\mA,\mF,S \to T,\mU_S,1]$.
Therefore it is now enough to show that~$B^\HH_1(\pi_2(x)) \cap B^\HH_1(\pi_2(y)) \neq \emptyset$ for some distinct $x,y\in\mC$,
which would imply that the cardinality of one among 
$$\left\{\mF_2(\pi_2(z)) \st \dH(\pi_2(x),\pi_2(z))\leq 1,\ z\in \mA^{s+2}\right\}, \ \left\{\mF_2(\pi_2(z)) \st \dH(\pi_2(y),\pi_2(z))\leq 1,\ z\in \mA^{s+2}\right\}$$ 
is at least 2, since $\mF_2(\pi_2(x)) \neq \mF_2(\pi_2(y))$. 
Observe that $|B^\HH_1(\pi_2(x))| = (s+1)(q-1)+1$ for any $x \in \mC,$ and that $$\sum_{x \in \mC} |B^\HH_1(\pi_2(x))| = q^s(s+1)(q-1)+q^s = q^{s+1}(s+1) - sq^s > q^{s+1}= |\mA^{a_2}|,$$ where we use the fact that $q > 1$. If $B^\HH_1(\pi_2(x)) \cap B^\HH_1(\pi_2(y)) = \emptyset$ for all distinct $x,y\in\mC$, then~$\sum_{x \in \mC} |B^\HH_1(\pi_2(x))| \le |\mA^{a_2}|$, a contradiction. Therefore $$B^\HH_1(\pi_2(x)) \cap B^\HH_1(\pi_2(y)) \neq \emptyset,$$ proving our claim.
We finally combine the said claim with the inequality in~\eqref{touse}, obtaining
$$q|\mC|+1=(q-1)|\mC| + 2 +(|\mC|-1)  \le (q-1) \, |\mC| + \sum_{x \in \mC} |\mF_2(B^\HH_1(\pi_2(x)))| \le q^{s+1}.$$ In particular, $|\mC| < q^s$, establishing the theorem.
\end{proof}

\begin{remark}
While we cannot compute the exact capacity of the networks of Family~\ref{fam:b}, preliminary experimental results and case-by-case analyses seem to indicate that
\begin{equation} \label{ourc}
\CC_1(\mathfrak{B}_s,\mA,\mU_S,1) = \log_{|\mA|} \left(\frac{|\mA|^s + |\mA|}{2}-1\right),
\end{equation}
which would be consistent
with Corollary \ref{cor:sbs}.
At the time of writing this paper
proving (or disproving) the equality in~\eqref{ourc} remains an open problem.
\end{remark}

Notice that Theorem \ref{thm:down} uses a packing argument ``downstream'' of the intermediate nodes in a simple 2-level network. Next, we work toward an upper bound on capacity that utilizes a packing argument ``upstream'' of the intermediate nodes. 
Later we will show that the packing argument ``upstream'' acts similarly to the Hamming Bound from classical coding theory and can sometimes give better results in the networking context.

We start with the following lemma, which will be the starting point for the Second Packing Bound below.

\begin{lemma}
\label{lem:goodup}
We follow Notation~\ref{not:s5}
and let $t \ge 0$ be an integer.
Let $\mC$ be an outer code for $(\mN,\mA)$.  Then $\mC$ is unambiguous 
for the channel $\Omega[\mN,\mA,\mF,S\to T,\mU_S,t]$
if and only if~$\mF^{-1}(\mF(B^\HH_t(x))) \cap \mF^{-1}(\mF(B^\HH_t(x'))) = \emptyset$ for all distinct $x,x' \in \mC$.
\end{lemma}
\begin{proof}
Suppose $\mF^{-1}(\mF(B^\HH_t(x))) \cap \mF^{-1}(\mF(B^\HH_t(x'))) \neq \emptyset$ for some distinct $x,x'\in\mC$, and let $y$ lie in that intersection. This implies $$\mF(y) \in \mF(B^\HH_t(x)) \cap \mF(B^\HH_t(x')) = \Omega(x) \cap \Omega(x').$$ In other words, $\mC$ is not unambiguous for the channel $\Omega[\mN,\mA,\mF,S\to T,\mU_S,t]$. For the other direction, it is straightforward to see that disjointness of the preimages implies disjointness of the fan-out sets.
\end{proof}

Following the notation of 
 Lemma \ref{lem:goodup}, for $n=2$
 and an unambiguous $\mC$ for the channel~$\Omega[\mN,\mA,\mF,S\to T,\mU_S,t]$ we obtain the ``packing'' result
$$\sum_{x \in \mC} |\mF^{-1}(\mF(B^\HH_t(x)))| \le |\mA|^{a_1 + a_2}.$$ 
By extending this idea to more than two intermediate nodes, one obtains the following upper bound.

\begin{theorem}[Second Packing Bound]
\label{thm:up}
Following Notation~\ref{not:s5},
suppose that $a_i \le b_i$ for $1\leq i \leq r$. Let $t \ge 0$ and let $\mC$ be an unambigious code for $\Omega[\mN,\mA,\mF,S \to T,\mU_S,t]$. Then
$$\sum_{\substack{t_1,\ldots,t_r \ge 0 \\ t_1+\ldots+t_r \le t}}\prod_{i=1}^r\binom{a_i}{t_i}(|\mA|-1)^{t_i} \sum_{x \in \mC} \prod_{j=r+1}^n |\mF_j^{-1}(\mF_j(B^\HH_{t-(t_1+\ldots+t_r)}(\pi_j(x))))| \le |\mA|^{a_1+a_2+\cdots a_n}.$$
\end{theorem}

The proof of the previous result follows from similar steps to those in the proof of Theorem~\ref{thm:down}, and we omit it here. 
As we did for the case of 
Theorem~\ref{thm:down},
we also spell out the case~$(n,r)=(2,1)$ for Theorem \ref{thm:up}.

\begin{corollary}
\label{down_n2}
Following Notation~\ref{not:s5}, suppose $n=2$ and $a_1\le b_1$.
Let $t \ge 1$ and let
$\mC$ be an unambiguous code for the channel $\Omega[\mN,\mA,\mF,S \to T,\mU_S,t]$. Then
$$\sum_{t_1=0}^t \binom{a_1}{t_1}(|\mA|-1)^{t_1} \sum_{x \in \mC} |\mF_2^{-1}(\mF_2(B^\HH_{t-t_1}(\pi_2(x))))| \le |\mA|^{a_1 + a_2}.$$
\end{corollary}

\begin{remark}
Although Theorem~\ref{thm:up} acts similarly to
the Hamming Bound from classical coding theory, we find it significantly more difficult to apply than our 
Theorem~\ref{thm:down}
in the networking context.
The reason behind this lies in the fact that, in general,
\begin{equation} \label{nono}\mF^{-1}(\mF(B^\HH_t(x))) \neq \mF_1^{-1}(\mF_1(\pi_1(B^\HH_t(x)))) \times  \mF_2^{-1}(\mF_2(\pi_2(B^\HH_t(x)))).
\end{equation}
This makes it challenging to evaluate the quantities in the statement of 
Theorem~\ref{thm:up}, even in the case $n=2$ (Corollary~\ref{down_n2}).
We give an example to illustrate the inequality in~\eqref{nono}.

Consider the Diamond Network $\mathfrak{A}_1$ with $\mA=\F_3$; see Section~\ref{sec:diamond} and Subsection~\ref{sec:families}. We know from Theorem~\ref{thm:diamond_cap}
that its capacity is $\log_3 2$.
A capacity-achieving pair $(\mF,\mC)$
is given by~$\mC=\{(1,1,1),(2,2,2)\}$ and $\mF=\{\mF_1,\mF_2\}$, where 
$\mF_1(a) = a$ for all $a \in \mA$
and $$\mF_2(x,y)=
\begin{cases}
      1  & \text{if} \ \ x = y = 1 \\
      2  & \text{if} \ \ x = y = 2 \\
      0  & \text{otherwise.} \\
\end{cases}$$
However,
\begin{align*}
\mF^{-1}(\mF(B^\HH_1((1,1,1))))
    &=\mF^{-1}(\{(0,1),(1,1),(2,1),(1,0)\}) \\
    &\neq \mF_1^{-1}(\{0,1,2\}) \times \mF_2^{-1}(\{0,1\})\\
    &=\mF_1^{-1}(\mF_1(B^\HH_1(1)))\times \mF_2^{-1}(\mF_2(B^\HH_1(1,1))).
\end{align*}
\end{remark}

We next present an upper bound on the capacity of the networks of Family~\ref{fam:e},
showing that the Generalized Network Singleton Bound
of Corollary~\ref{cor:sing} is not met for this family, no matter what the alphabet size is.
Notice that the number of indices 
$i$ for which 
$a_{i}\leq b_{i}$ 
for a network 
of Family~\ref{fam:e}
is 0 whenever~$t>1$.
In particular, the strategy behind the proofs of Theorems~\ref{thm:down} and~\ref{thm:up}
is of no help in this case.

\begin{theorem}
\label{thm:mete}
Let $\mathfrak{E}_t=(\mV,\mE,S,\{T\})$ be a member of
Family~\ref{fam:e}. Let $\mA$ be any network alphabet and let $\mU_S$ be the  set of edges of $\mathfrak{E}_t$ directly connected to $S$.
We have
$$\CC_1(\mathfrak{E}_t,\mA,\mU_S,t) <1.$$
In particular, the Generalized Network Singleton Bound of Corollary~\ref{cor:sing} is not met.
\end{theorem}
\begin{proof}
Let $q:=|\mA|$.
Suppose by way of contradiction that there exists an unambiguous code~$\mC$ of size $q$ for $\Omega[\mathfrak{E}_{t},\mA,\mF,S \to T,\mU_S,t]$, for some network code  $\mF=\{\mF_{1},\mF_{2}\}$.  Since $\mC$ is unambiguous, it must be the case that $\mC$ has minimum distance equal to at least $2t+1$, making $\mC$ an MDS code of length and minumum distance $2t+1$.
Thus, $\pi_1(\mC)$ must contain $q$ distinct elements. We also observe that the restriction of $\mF_1$ to $\pi_1(\mC)$ must be injective, since otherwise the intersection of fan-out sets would be nonempty for some pair of codewords in $\mC$. Therefore we have that $\mF_1(\pi_1(\mC))=\mA$. 

Putting all of the above together, we conclude that there must exist $x, y \in \mC$ and $e\in \mA^{2t+1}$ such that $x\neq y$, $\mF_1(\pi_1(x))=\mF_1(\pi_1(e))$, and $d^H(\pi_1(e),\pi_1(y))\le t-1$. We thus have
\begin{align*}
x'&:=(x_1,\ldots,x_t,x_{t+1},y_{t+2},y_{t+3},\ldots,y_{2t+1}) \in B^\HH_t(x),\\
y'&:=(e_1,\ldots,e_t,x_{t+1},y_{t+2},y_{t+3},\ldots,y_{2t+1}) \in B^\HH_t(y).  \end{align*}
Finally, observe that $(\mF_1(\pi_1(x')),\mF_2(\pi_2(x')))=(\mF_1(\pi_1(y')),\mF_2(\pi_2(y')))\in \Omega(x)\cap \Omega(y)$, a contradiction.
\end{proof}

We now turn 
to the networks of Family~\ref{fam:a}, which are probably the most natural generalization of the Diamond Network of Section~\ref{sec:diamond}.
We show that none of the members of 
Family~\ref{fam:a} meet the Generalized Network Singleton Bound of Corollary~\ref{cor:sing} with equality, no matter what the underlying alphabet is.

\begin{theorem}
\label{thm:meta}
Let $\mathfrak{A}_t=(\mV,\mE,S,\{T\})$ be a member of
Family~\ref{fam:a}. Let $\mA$ be any network alphabet and let $\mU_S$ be the  set of edges of $\mathfrak{A}_t$ directly connected to $S$.
We have
$$\CC_1(\mathfrak{A}_t,\mA,\mU_S,t) <t.$$
In particular, the Generalized Network Singleton Bound of Corollary~\ref{cor:sing} is not met.
\end{theorem}
\begin{proof}
Let $q:=|\mA|$. Towards a contradiction, assume that
the rate $t$ is achievable. That is, there exists an unambiguous code $\mC$ of size $q^t$ for $\Omega[\mathfrak{A}_{t},\mA,\mF,S \to T,\mU_S,t]$, where $\mF=\{\mF_{1},\mF_{2}\}$ is a network code for $(\mathfrak{A}_t,\mA)$. Note that $|\pi_2(\mC)|=|\mC|$ and that the restriction of
$\mF_2$ to $\pi_2(\mC)$
is injective, as otherwise
the intersection of fan-out sets would be nonempty for some pair of codewords in $\mC$, as one can easily check.

We claim that there must exist two distinct codewords $x,y \in \mC$ such that~$B^\HH_t(\pi_2(x)) \cap B^\HH_t(\pi_2(y)) \neq \emptyset$.
Observe that
\begin{align}
\sum_{x \in \mC} |B^\HH_t(\pi_2(x))| &= q^t\left(\sum_{k=0}^t \binom{2t}{k}(q-1)^k\right) \nonumber \\
&> q^t\left(\sum_{k=0}^t \binom{t}{k}(q-1)^k\right) \nonumber \\
&= q^{2t}, \label{eqn:factorial}
\end{align}
where Equation \eqref{eqn:factorial} follows from the Binomial Theorem. If $B^\HH_t(\pi_2(x)) \cap B^\HH_t(\pi_2(y)) = \emptyset$ for all distinct $x,y\in\mC$, then we would have $$\sum_{x \in \mC} |B^\HH_t(\pi_2(x))| = \left| \bigcup_{x \in \mC} B_t^\HH(\pi_2(x)) \right| \le q^{2t}.$$ Therefore it must be the case that some such intersection is nonempty.
In other words, there exist some distinct $x,y\in \mC$ and $e \in \mA^{2t}$ such that $e \in B^\HH_t(\pi_2(x)) \cap B^\HH_t(\pi_2(y))$. 

By the fact that the restriction of
$\mF_2$ to $\pi_2(\mC)$
is injective and its codomain has size~$q^t=|\pi_2(\mC)|$,
the restriction of
$\mF_2$ to $\pi_2(\mC)$
must be surjective as well.
Thus 
$\mF_2(e) = \mF_2(\pi_2(a))$ for some $a \in \mC$.
If $a=x$, then $\mF_2(e)=\mF_2(\pi_2(x))$ and we have~$(\mF_1(\pi_1(y)),\mF_2(e))=(\mF_1(\pi_1(y)),\mF_2(\pi_2(x))) \in \Omega(x) \cap \Omega(y)$, the intersection of the fan-out sets, which is a contradiction to $\mC$ being unambiguous. A similar argument holds if $a=y$. Otherwise, we still have $(\mF_1(\pi_1(x)),\mF_2(\pi_2(a))) \in \Omega(x) \cap \Omega(a),$ a contradiction. We conclude that there cannot be an unambiguous code of size $q^t$.
\end{proof}

We conclude this section by mentioning that the networks of both Family \ref{ex:u} and Family \ref{fam:d} do achieve the Generalized Network Singleton Bound of Corollary~\ref{cor:sing}: both capacities will be examined in the next section.

\section{Simple 2-Level Networks: Lower Bounds}
\label{sec:2level_lower}

We devote this section to
deriving lower bounds
on the capacity of simple 
2-level networks.
In connection with Theorem \ref{thm:channel},
we note that a rate may be achievable for a simple 2-level network and yet not achievable for every corresponding simple 3-level network. 
We start by establishing the notation for this section.

\begin{notation}
Throughout this section we follow Notation~\ref{not:s5}.
In particular, we work with simple 2-level networks
$\mN=(\mV,\mE,S,\{T\})=([a_1,a_2,\ldots,a_n],[b_1,b_2,\ldots,b_n])$, denoting by $\mU_S$ the set of edges directly connected to the source $S$.
For any $t \ge 1$, we partition the vertices of such a network into the following (disjoint and possibly empty) sets: 
\begin{align*}
    I_{1}(\mN,t)&=\{i \st a_i\geq b_i+2t\},\\
    I_{2}(\mN,t)&=\{i \st a_i\leq b_i\},\\
    I_{3}(\mN,t)&=\{i \st b_i + 1\leq a_i \leq b_i+2t-1\}.
\end{align*}
\end{notation}

\begin{theorem}
\label{thm:lowbound}
Suppose that $\mA$ is a sufficiently large finite field.
Let $t \ge 1$ and define the set~$\Tilde{I}_{3}(\mN,t) =\{i \in I_{3}(\mN,t) \st a_i > 2t\}$. 
Then,
\begin{equation*} 
   \CC_1(\mN,\mA,\mU_S,t)\geq \sum_{i \in I_1(\mN,t)} b_i + \max\left\{X,Y \right\}, 
\end{equation*}
where  
\begin{align*}
X &= \sum_{i\in \Tilde{I}_3(\mN,t)} (a_i-2t) + \max\left\{0,\left(\sum_{i\in I_2(\mN,t)} a_i\right) - 2t\right\}, \\
Y &= \max\left\{0,\left(\sum_{i\in I_2(\mN,t)} a_i + \sum_{i\in I_3(\mN,t)} b_i\right)- 2t\right\}.
\end{align*}
\end{theorem}

\begin{proof} 
We will construct a network code $\mF = \{\mF_1,\ldots,\mF_n\}$ for~$(\mN,\mA)$ and an unambiguous outer code in various steps, indicating how each intermediate node operates. The alphabet~$\mA$ should be 
a large enough finite field
that admits MDS codes with parameters as described below. 

For each $i \in I_1(\mN,t)$, let~$\mF_i$ be a minimum distance decoder for an MDS code with parameters $\left[b_i+2t,\ b_i,\ 2t+1\right]$ that is sent along the first $b_i +2t$ edges incoming to $V_i$. The source sends arbitrary symbols along the extraneous $a_i - (b_i+2t)$ edges and these will be disregarded by $V_i$. Via the intermediate nodes indexed by $I_1(\mN,t)$, we have thus achieved $\sum_{i\in I_{1}(\mN,t)} b_i$ information symbols decodable by the terminal. 

In the remainder of the proof we will show that either an extra $X$ symbols or an extra $Y$ symbols~(i.e., an extra $\max\{X,Y\}$ symbols in general) can be successfully transmitted. We show how the two quantities $X$ and $Y$ can be achieved separately, so that we may choose the better option of the two and achieve the result.

\begin{enumerate}
\item For each $i \in \Tilde{I}_3(\mN,t)$, let $\mF_i$ be a minimum distance decoder for an MDS code with parameters $\left[a_i,\ a_i-2t,\ 2t+1\right]$ that sends its output across the first $a_i -2t$ outgoing edges from $V_i$. Arbitrary symbols are sent along the remaining outgoing edges and they will be ignored at the destination. We disregard all intermediate nodes with indices from the set $I_3(\mN,t)\setminus \Tilde{I}_3(\mN,t)$ as they will not contribute to the lower bound  in this scheme. 
The symbols sent through the vertices indexed by ${I_2}(\mN,t)$ will be globally encoded via an 
MDS code with parameters
$$\left[\sum_{i\in I_{2}(\mN,t)} a_i,\ \left(\sum_{i\in I_{2}(\mN,t)} a_i\right) - 2t,\ 2t+1\right].$$
The intermediate nodes indexed by ${I_2}(\mN,t)$
will simply forward the incoming symbols along the first $a_i$ outgoing edges, sending arbitrary symbols along the other $b_i - a_i$ outgoing edges. Symbols sent on these outgoing edges will be disregarded at the destination. If~$\sum_{i\in I_{2}(\mN,t)} a_i \leq 2t$, we instead ignore the vertices indexed by $I_2(\mN,t)$ in this scheme. In conclusion, via the intermediate nodes indexed by $I_2(\mN,t)\cup I_3(\mN,t)$, we have added $X$ information symbols decodable by the terminal.

\item For $i \in I_2(\mN,t) \cup I_3(\mN,t)$, let each $\mF_i$ simply forward up to $b_i$ received symbols as follows. If $i\in I_2(\mN,t)$, then $a_i$ symbols will be forwarded along the first $a_i$ outgoing edges and arbitrary symbols will be sent along the other $b_i - a_i$ outgoing edges. These edges will be disregarded at the destination. If $i\in I_3(\mN,t)$, then $b_i$ symbols sent over the first~$b_i$ edges incoming to $V_i$ will be forwarded. The source $S$ will send arbitrary symbols along the other $a_i - b_i$ incoming edges, which will however be ignored. At the concatenation of all (non-arbitrary) coordinates to an intermediate node with index in $I_2(\mN,t)\cup I_3(\mN,t)$, the outer code will be an MDS code with parameters
$$\left[\sum_{i\in I_{2}(\mN,t)} a_i + \sum_{i\in I_{3}(\mN,t)} b_i,\ \left(\sum_{i\in I_{2}(\mN,t)} a_i + \sum_{i\in I_{3}(\mN,t)} b_i\right) -2t ,\ 2t+1\right].$$ 
The MDS code is then decoded at the terminal. If $\sum_{i\in I_{2}(\mN,t)} a_i + \sum_{i\in I_{3}(\mN,t)} b_i \leq 2t$, we ignore the coordinates corresponding to $I_2(\mN,t)\cup I_3(\mN,t)$ in this scheme.
In conclusion, via the intermediate nodes in $I_2(\mN,t)\cup I_3(\mN,t)$ we have added $Y$ information symbols decodable at the terminal.
\end{enumerate}
This concludes the proof.
\end{proof}

We note that 
Theorem \ref{thm:lowbound} does not always yield a positive value, even for networks where the capacity is positive. For example, 
for the member
$\mathfrak{A}_2$
of 
Family \ref{fam:a}, Theorem \ref{thm:lowbound} gives a lower bound of 0 for the capacity. However, the following result shows that $\CC_1(\mathfrak{A}_2,\mA,\mU_S,2)$ is, in fact, positive.

\begin{proposition}
\label{prop:atleasta}
Let $\mathfrak{A}_t=(\mV,\mE,S,\{T\})$ be a member of
Family~\ref{fam:a}. Let $\mA$ be any network alphabet and let $\mU_S$ be the  set of edges of $\mathfrak{A}_t$ directly connected to $S$.
We have
$$\CC_1(\mathfrak{A}_2,\mA,\mU_S,2) \ge 1.$$
\end{proposition}

\begin{proof}
Fix two distinct alphabet symbols $*, *'\in \mA$.
The source $S$ encodes each element of~$\mA$ using a six-times repetition code. Vertex~$V_1$ 
simply forwards the received symbols, while vertex~$V_2$ proceeds as follows.
If, on the four incoming edges, an alphabet symbol appears at least three times, $V_2$ forwards that symbol on both outgoing edges.
Otherwise $V_2$ outputs~$*$ on one edge, and~$*'$ on the other. At destination, if the incoming symbols from~$V_{2}$ match, then the terminal $T$ decodes to that symbol. Otherwise, it decodes to the alphabet symbol sent from~$V_{1}$. All symbols from $\mA$ can be sent with this scheme, including $*$ and $*'$, giving a capacity of at least $1$. 
\end{proof}

We also give the following less sophisticated lower bound on the capacity of simple 2-level networks. The idea behind the proof is to construct a subnetwork of the original one where a lower bound for the capacity can be easily established.

\begin{proposition}
\label{prop:lin} 
Suppose that $\mA$ is a sufficiently large finite field.
For all $t \ge 0$ we have
\begin{equation} \label{albe}
    \CC_1(\mN,\mA,\mU_S,t) \ge  \max\left\{0, \, \sum_{i=1}^n\min\{a_i,b_i\} -2t\right\}.
\end{equation}
\end{proposition}

\begin{proof}
We will construct a network $\mN'=(\mV,\mE',S,\{T\})$ such that $$\CC_1(\mN,\mA,\mU_S,t) \ge \CC_1(\mN',\mA,\mU'_S,t),$$
where $\mU'_S$ is the set of edges directly connected to $S$ in the new network $\mN'$.
For $i \in I_2(\mN,t)$, remove the first $b_i - a_i$ outgoing edges from $V_i.$ Similarly, for each~$i \in \{1,\ldots,n\} \setminus I_2(\mN,t)$, remove the first $a_i-b_i$
incoming edges to $V_i$. Denote the new network obtained in this way by~$\mN'$ and observe that all intermediate nodes $V_i'$ in~$\mN'$ have $\degin(V_i')=\degout(V_i')$. This implies that~$I_2(\mN',t)=\{1,\ldots,n\}.$ Since removing edges cannot increase the capacity (technically, this is a consequence of Proposition~\ref{prop:finer}), we have $\CC_1(\mN,\mA,\mU_S,t) \ge \CC_1(\mN',\mA,\mU_S,t).$ We will give a coding scheme for $\mN'$ that achieves the right-hand side of~\eqref{albe} to conclude the proof.
Consider the edge-cut 
$$\mE'= \bigcup_{i=1}^{n} \inn(V_{i}').$$
 By the definition of $\mN'$, we have $|\mE'| = \sum_{i=1}^n\min\{a_i,b_i\}$. 
 Since the left-hand side of~\eqref{albe} is always non-negative, we shall assume
 $|\mE'| > 2t$ without loss of generality.
 Under this assumption, we will prove that the rate 
 $|\mE'| -2t$
 is achievable.
 
 We choose the outer code $\mC$ to be an MDS code with parameters $[|\mE'|,|\mE'|-2t,2t+1]$ over the finite field~$\mA$. All of the intermediate nodes forward incoming packets 
 (the number of outgoing edges will now allow this). The terminal receives a codeword of length $|\mE'|$ and decodes it according to the chosen MDS code. \end{proof}

\begin{remark}
The lower bound of Theorem~\ref{thm:lowbound} 
is larger than the one of Proposition~\ref{prop:lin}, and it can be 
strictly larger.
For example, consider the simple 2-level network $\mN=([2,5,6],[2,2,2])$. By the Generalized Network Singleton Bound of Corollary~\ref{cor:sing}, we have ~$\CC_1(\mN,\mA,\mU_S,2) \le 4$. Theorem \ref{thm:lowbound} gives a lower bound of $3$ for $\CC_1(\mN,\mA,\mU_S,2)$, while Proposition \ref{prop:lin} gives a lower bound of 2 for the same quantity.
\end{remark}

As an application of  Theorem~\ref{thm:lowbound} we provide a family of simple 2-level networks where
the Generalized Network Singleton Bound of Corollary~\ref{cor:sing} is always met with equality.

\begin{corollary}[see \cite{beemer2021curious}]
\label{cor:conf}
Suppose that $\mA$ is a sufficiently large finite field.
Let $t \ge 0$ and suppose that
$I_3(\mN,t) = \emptyset$. Then the Generalized Network Singleton Bound of Corollary~\ref{cor:sing} is achievable.
\end{corollary}
\begin{proof}
Since $I_3(\mN,t) = \emptyset$, $\Tilde{I}_3(\mN,t)$ as defined in Theorem \ref{thm:lowbound} is also empty. Therefore Theorem~\ref{thm:lowbound} gives 
\begin{equation}
\label{eqn:eql-sb}
\CC_1(\mN,\mA,\mU_S,t)\geq \sum_{i\in I_{1}(\mN,t)}b_i+\max\left\{0,\left(\sum_{i\in I_{2}(\mN,t)} a_i\right) -2t\right\}.    
\end{equation}
Choosing $P_1 = I_1(\mN,t)$ and $P_2 = I_2(\mN,t)$ in Corollary~\ref{cor:sing} gives us that the Generalized Network Singleton (upper) Bound is equal to the right-hand side of \eqref{eqn:eql-sb}.
\end{proof}

We include an example illustrating how Corollary~\ref{cor:conf} is applied in practice.

\begin{example}
Take $\mN=([12,8,2,2,1],[5,2,4,3,1])$ and let $\mA$ be a sufficiently large finite field. We want to compute $$\CC_1(\mN,\mA,\mU_S,3).$$ We have $I_1(\mN,3) = \{1,2\}$, $I_2(\mN,t) = \{3,4,5\}$, and $I_3(\mN,t)=\emptyset$. The intermediate nodes indexed by $i \in I_1(\mN,3)$ use local MDS decoders of dimension $b_i$ to retrieve $\sum_{i\in I_{1}(\mN,t)}b_i$ information symbols. That is, $V_1$ uses an $[11,5,7]$ MDS decoder and $V_2$ uses an $[8,2,7]$ MDS decoder. The intermediate nodes indexed by $I_2(\mN,3)$ are supposed to cooperate with each other and can be seen as one intermediate node with $5$ incoming and $8$ outgoing edges. 
In the notation of the proof of Corollary~\ref{cor:conf}, we have~$\sum_{i\in I_{2}(\mN,t)} a_i < 2t$. In this case, 
the vertices indexed by $I_{2}(\mN,t)$ are disregarded. 
In total, the Generalized Network Singleton Bound of Corollary~$\ref{cor:sing}$, whose value is 7, is met with equality. 
\end{example}

We finally turn to the networks of Families~\ref{fam:c} and~\ref{fam:d}, which we did not treat in Section~\ref{sec:upper}.
The next two results compute 
the capacities of these families and prove, in particular, that the Generalized Network Singleton Bound of Corollary~\ref{cor:sing}
is attained.
 This, along with the upper bounds of Section~\ref{sec:upper} for the families \ref{fam:a},\ref{ex:s} and~\ref{fam:e}, demonstrates that emptiness of $I_3(\mN,t)$ is far from a complete characterization of the Generalized Network Singleton Bound achievability~(i.e., Corollary \ref{cor:conf} is not biconditional); see Remark \ref{rem:i3crit}.

\begin{theorem}
\label{thm:metc}
Let $\mathfrak{C}_t=(\mV,\mE,S,\{T\})$ be a member of
Family~\ref{fam:c}. Let $\mA$ be any network alphabet and let $\mU_S$ be the  set of edges of $\mathfrak{C}_t$ directly connected to $S$. We have
$$\CC_1(\mathfrak{C}_t,\mA,\mU_S,t) = 1.$$ In particular, the Generalized Network Singleton Bound of Corollary~\ref{cor:sing} is met with equality.
\end{theorem}

\begin{proof}
We let $q:=|\mA|$ for ease of notation (but note that $q$ does not need to be a prime power). We will construct
a network code $\mF=\{\mF_1,\mF_2\}$
for $(\mN,\mA)$ and an unambiguous outer code~$\mC$ for the channel $\Omega[\mathfrak{C}_{t},\mA,S \to T,\mU_S,t]$ of size~$q$. 
Recall that $t \ge 2$ by the very definition of Family~\ref{fam:c}.

\underline{Case 1}: $q=t=2$. \ Let $\mA = \{a,b\}.$ Encode each element of $\mA$ using a 5-times repetition code. It can be checked that the pair $(\mC,\mF)$ achieves the desired capacity, where~$\mC=\{(a,a,a,a,a),(b,b,b,b,b)\}$ and $\mF=\{\mF_1,\mF_2\}$ is defined as follows.
$\mF_1(v) = v$ for all~$v \in \mA^2$
and $$\mF_2(w)=
\begin{cases}
      (a,a)  & \text{if} \ \ w = (a,a,a), \\
      (b,b)  & \text{if} \ \ w = (b,b,b), \\
      (a,b)  & \text{if} \ \ w \in S^\HH_{1}(a,a,a), \\
      (b,a)  & \text{if} \ \ w \in S^\HH_{1}(b,b,b). \\
\end{cases}$$

\underline{Case 2}: $\max\{q,t\} \ge 3$. \ Encode each element of $\mA$ using a $2t+1$-times repetition code, so that the codewords of $\mC$ are given by $c_{1},c_{2},\ldots,c_{q}$. The intermediate function $\mF_1$ simply forwards its received symbols. 

We will next define the function $\mF_2$. Notice that there are $q\left(\lfloor t/2 \rfloor+1\right)$ shells $S^\HH_{i}(\pi_{2}(c_{j}))$ for $i\in \{ 0,1,\ldots, \lfloor t/2 \rfloor\}$ and $j\in \{1,\ldots,q\}$. Define $$D= \mA^{t+1}\setminus \bigcup_{j\in \{1,\ldots,q\}}B^\HH_{\lfloor t/2 \rfloor}(\pi_{2}(c_{j})),$$
where $\pi_2$ is defined in Notation~\ref{not:ballsetc}.
That is, $D$ is the set of words that do not belong to any ball of radius $\lfloor t/2 \rfloor$ centered at the projection of a codeword, so that the union of $D$ with the collection of the shells described above is all of $\mA^{t+1}$. Let $q':=q\left(\lfloor t/2 \rfloor+1\right)+1$. Since $\max\{q,t\} \ge 3$,
we have
\begin{equation*}
q'<q^{t}.
\end{equation*}
Therefore, we may define $\mF_2$ to be such that the sets $D$ and $S^\HH_{i}(\pi_{2}(c_{j}))$ for $i\in \{ 0,1,\ldots, \lfloor {t/2} \rfloor\}$ and $j\in \{1,\ldots, q\}$ each map to a single, \textit{distinct} element of $\mA^{t}$. 

Decoding at the terminal is accomplished as follows. Suppose that the terminal receives the pair $(x,y)\in \mA^{t} \times \mA^{t}$. First, the set $\mF_2^{-1}(y)$ is computed. If $\mF_2^{-1}(y) = D$, then the terminal decodes to the majority of the coordinates in $x$, 
this is guaranteed to be the transmitted symbol based on how $D$ was defined. If $\mF_2^{-1}(y) =\pi_{2}(c_{j})$ for some $j$, the terminal decodes to the repeated symbol in $c_j$. 
Finally, if $\mF_2^{-1}(y) = S^\HH_{i}(\pi_{2}(c_{j}))$ for $i\in \{1,\ldots, \lfloor t/2 \rfloor\}$, then it is not clear to the decoder if $i$ symbols were corrupted incoming to~$V_{2}$ (and up to $t-i$ symbols to~$V_{1}$), or if at least $(t+1)-i$ symbols were corrupted from~$V_{2}$ (and up to $i-1$ to~$V_{1}$). To differentiate the possibilities, the terminal looks to $x$. If at least $i$ of the symbols of $x$ are consistent with~$c_{j}$, then we must be in the first scenario (recall $i\leq \lfloor t/2\rfloor$), so the terminal decodes to $c_{j}$. Otherwise, at most $i-1$ symbols were changed to~$V_{1}$, and hence the majority of the symbols in $x$ will correspond to the transmitted codeword. In this case, the terminal decodes to the majority of symbols in $x$. 
\end{proof}

We propose an example that illustrates the scheme in the proof of Theorem \ref{thm:metc}.
\begin{example}
We will show that $\CC_1(\mathfrak{C}_4,\mA,\mU_S,4) = 1$ when $|\mA|=2$; see Family~\ref{fam:c} for the notation. Let $\mA=\{a,b\}$. Following the proof of Theorem \ref{thm:metc}, we consider the repetition code that has the codewords $c_1=(a,a,a,a,a,a,a,a,a)$ and~$c_2=(b,b,b,b,b,b,b,b,b)$. Observe that~$D=\emptyset$. We illustrate the decoding of $c_1$ case-by-case. Since the alphabet size is equal to~2, the analysis of the $\binom{9}{4}$ possible actions of the adversary can be reduced to the following~5 basic cases (which also cover the favorable scenario where the adversary might not use their full power). 
\begin{enumerate}
    \item$(a,a,a,a,a,a,a,a,a)$ is changed into $(a,a,a,a,a,b,b,b,b)$ by the adversary. Since we have $(a,b,b,b,b) \in S_1(\pi_2(c_2))$ and none of the coordinates of $(a,a,a,a)$ is $b$, the terminal decodes to $c_1$.
    \item $(a,a,a,a,a,a,a,a,a)$ is changed into $(a,a,a,b,a,a,b,b,b)$ by the adversary. Since we have~$(a,a,b,b,b) \in S_2(\pi_2(c_2))$ and only one of the coordinates of $(a,a,a,b)$ is $b$, the terminal decodes to $c_1$.
    \item $(a,a,a,a,a,a,a,a,a)$ is changed into $(a,a,b,b,a,a,a,b,b)$ by the adversary. Since we have~$(a,a,a,b,b) \in S_2(\pi_2(c_1))$ and at least 2 of the coordinates of $(a,a,b,b)$ is $a$, the terminal decodes to $c_1$.
    \item $(a,a,a,a,a,a,a,a,a)$ is changed into $(a,b,b,b,a,a,a,a,b)$ by the adversary. Since we have~$(a,a,a,a,b) \in S_1(\pi_2(c_1))$ and least 1 of the coordinates of $(a,b,b,b)$ is $a$, the terminal decodes to $c_1$.
    \item $(a,a,a,a,a,a,a,a,a)$ is changed into $(b,b,b,b,a,a,a,a,a)$ by the adversary. Since we have~$(a,a,a,a,a) = \pi_2(c_1)$, the terminal decodes to $c_1$.
\end{enumerate}
This shows that, \textit{no matter} what the action of the adversary is, one alphabet symbol can always be transmitted unambiguously. 
\end{example}

\begin{remark}
\label{rem:exclude}
The reason for excluding the case $t=1$ in the Definition of Family \ref{fam:c} is the non-achievability of the Generalized Network Singleton Bound of Corollary~\ref{cor:sing} given in Theorem~\ref{thm:diamond_cap}. It should be noted that since that case is already studied, excluding it from Family~\ref{fam:c} makes sense to exhibit a family of networks which achieve Corollary~\ref{cor:sing} as proven in Theorem~\ref{thm:metc}.
\end{remark}

We now turn to the networks of Family~\ref{fam:d}, the only family introduced in Subsection \ref{sec:families} that we have not yet considered.

\begin{theorem}
\label{thm:metd}
Let $\mathfrak{D}_t=(\mV,\mE,S,\{T\})$ be a member of
Family~\ref{fam:d}. Let $\mA$ be any network alphabet and let $\mU_S$ be the  set of edges of $\mathfrak{D}_t$ directly connected to $S$. We have $$\CC_1(\mathfrak{D}_t,\mA,\mU_S,t)= 1.$$ In particular, the Generalized Network Singleton Bound of Corollary~\ref{cor:sing} is met with equality.
\end{theorem}

\begin{proof}
Fix an alphabet symbol $* \in \mA$.
The source $S$ encodes each element of $\mA$ using a~$4t$-times repetition code.
The vertices $V_1$ and $V_2$ implement a majority-vote decoder each, unless 
two symbols occur an equal number of times over the incoming edges. In that case, the vertices output $*$. At the destination, if the incoming symbols match, then the terminal decodes to that symbol. Otherwise, it decodes to the alphabet symbol that is not equal to $*$. All symbols from~$\mA$ can be sent with this scheme, including $*$, giving a capacity of at least $1$ and establishing the theorem. 
\end{proof}

\begin{remark}
\label{rem:i3crit}

Let $\mN = ([a_1,a_2],[b_1,b_2])$ a simple 2-level network with $n=2$ intermediate nodes. Let~$\mA$ be a sufficiently large finite field.

\begin{table}[!ht]
\begin{center}
 \renewcommand{\arraystretch}{1.4}
 
 \begin{tabular}{|x{4.5cm}| x{4.5cm}| x{4.5cm}|} 
 \hline
  Size of $I_3(\mN,t)$ & Corollary \ref{cor:sing} is met & Corollary \ref{cor:sing} is not met  \\ [0.5ex] 
 \hline\hline
 0  & \text{always}  & \text{never }(Corollary \ref{cor:conf})   \\
 \hline
 1 & $\mathfrak{C}_2$ (Theorem \ref{thm:metc} ) & $\mathfrak{A}_2$ (Theorem \ref{thm:meta})   \\
 \hline
 2 & $\mathfrak{D}_2$ (Theorem \ref{thm:metd})  & $\mathfrak{E}_2$ (Theorem \ref{thm:mete} )  \\
\hline
\end{tabular}
\end{center} 
\caption{On the 1-shot capacity of simple 2-level networks with 2 intermediate nodes.\label{tablele}}
\end{table}

We present Table \ref{tablele} to illustrate that the size of $I_3(N,t)$ cannot be considered as a criterion for the achievability of the Generalized Network Singleton Bound of Corollary~\ref{cor:sing}.

\end{remark}

\section{The Double-Cut-Set Bound and Applications}
\label{sec:double-cut-bd}

In this section we illustrate how the results on 2-level and 3-level networks derived throughout the paper can be combined with each other and applied to study an arbitrarily large and complex network $\mN$.

We already stressed in Section~\ref{sec:motiv}
that known cut-set bounds are not sharp in general when considering a \textit{restricted} adversary (whereas they are sharp, under certain assumptions, when the adversary is not restricted; see Theorem~\ref{thm:mcm}).

The main idea behind the approach taken in this section is to consider \textit{pairs} of edge-cuts, rather than a single one, and study the ``information flow'' between the two. This allows one to better capture the adversary's restrictions and to incorporate them into explicit upper bounds for the capacity of the underlying network $\mN$.
All of this leads to our Double-Cut-Set Bound below; see Theorem~\ref{thm:dcsb}. In turn,
Theorem~\ref{thm:dcsb} can be used to derive an upper bound for the capacity of $\mN$ in terms of the capacity of an \textit{induced} 3-level network.
This brings the study of~3-level networks and their reduction to 2-level networks into the game; see Sections~\ref{sec:net-2-and-3} and~\ref{sec:upper}.

A concrete application of the machinery developed in this section will be illustrated later in Example~\ref{ex:tulipA}, where we will go back to our opening example of Section~\ref{sec:motiv}
and rigorously compute its capacity.
Another network is studied in Example~\ref{ex:second}. 

We begin by introducing supplementary definitions and notation specific to this section. 

\begin{definition} \label{def:imme}
Let $\mN=(\mV,\mE,\mS,\bfT)$ be a
network and let $\mE_1, \mE_2 \subseteq \mE$ be non-empty edge sets. We say that $\mE_1$ \textbf{precedes} $\mE_2$ if every path from $S$ to an edge of $\mE_2$
contains an edge of $\mE_1$.
In this situation, for $e \in \mE_2$ and $e' \in \mE_1$, we say that $e'$ is an \textbf{immediate predecessor of $e$ in~$\mE_1$} if $e' \preccurlyeq e$ and there is no $e'' \in \mE_1$ with $e' \preccurlyeq e'' \preccurlyeq e$ and $e' \neq e''$.
\end{definition}

We illustrate the previous notions with an example.

\begin{example} \label{ex:imme}
Consider the network $\mN$ and the edge sets $\mE_1$ and $\mE_2$ in Figure~\ref{tulipsolved} below. Then~$\mE_1$ precedes~$\mE_2$.
We have $e_2 \preccurlyeq e_{10}$ and $e_9 \preccurlyeq e_{10}$. Moreover,
 $e_9$ is an immediate predecessor of $e_{10}$ in $\mE_1$, while $e_2$ is not. 
\end{example}

\begin{figure}[h!]
\centering
\begin{tikzpicture}

\draw[blue, line width=1.5pt] (1,2) .. controls (2,1) .. (2,-0) ;

\draw[blue, line width=1.5pt] (6.3,1) -- (7.3,-1) ;

\draw[red, line width=1.5pt] (9.5,2.5) .. controls (9.2,1) .. (10,0);

\tikzset{vertex/.style = {shape=circle,draw,inner sep=0pt,minimum size=1.9em}}
\tikzset{nnode/.style = {shape=circle,fill=myg,draw,inner sep=0pt,minimum
size=1.9em}}
\tikzset{edge/.style = {->,> = stealth}}
\tikzset{dedge/.style = {densely dotted,->,> = stealth}}
\tikzset{ddedge/.style = {dashed,->,> = stealth}}

\node[vertex] (S1) {$S$};

\node[shape=coordinate,right=\mynodespace of S1] (K) {};

\node[nnode,above=0.6\mynodespace of K] (V1) {$V_1$};

\node[nnode,below=0.6\mynodespace of K] (V2) {$V_2$};

\node[nnode,right=\mynodespace of K] (V3) {$V_3$};

\node[nnode,right=\mynodespace of V3] (V4) {$V_4$};


\node[vertex,right=3\mynodespace of V1 ] (T1) {$T_1$};

\node[vertex,right=3\mynodespace of V2] (T2) {$T_2$};

\draw[ddedge,bend left=15] (S1)  to node[fill=white, inner sep=3pt]{\small $e_1$} (V1);

\draw[ddedge,bend right=15] (S1)  to node[fill=white, inner sep=3pt]{\small $e_2$} (V1);

\draw[ddedge,bend left=15] (S1) to  node[fill=white, inner sep=3pt]{\small $e_3$} (V2);

\draw[ddedge,bend right=15] (S1) to  node[fill=white, inner sep=3pt]{\small $e_4$} (V2);

\draw[ddedge,bend left=0] (V1) to  node[fill=white, inner sep=3pt]{\small $e_6$} (V3);

\draw[edge,bend left=0] (V4)  to node[fill=white, inner sep=3pt]{\small $e_{10}$} (T1);

\draw[edge,bend left=0] (V4)  to node[fill=white, inner sep=3pt]{\small $e_{11}$} (T2);

\draw[edge,bend left=0] (V1)  to node[fill=white, inner sep=3pt]{\small $e_{5}$} (T1);

\draw[edge,bend left=0] (V2)  to node[fill=white, inner sep=3pt]{\small $e_{8}$} (T2);

\draw[ddedge,bend left=0] (V2) to  node[fill=white, inner sep=3pt]{\small $e_7$} (V3);

\draw[ddedge,bend left=0] (V3) to  node[fill=white, inner sep=3pt]{\small $e_{9}$} (V4);

\node[text=blue] (E1) at (2.3,0.0) {$\mE_1$};
\node[text=blue] (E11) at (7.6,-1) {$\mE_1$};

\node[text=red] (E2) at (10.3,0.1) {$\mE_2$};
\end{tikzpicture} 
\caption{{{Network $\mN$ for Examples~\ref{ex:imme}, \ref{extransf}, and ~\ref{ex:tulipA}}}. \label{tulipsolved}}
\end{figure}

The following notion of channel will be crucial in our approach. It was formally defined in~\cite{RK18} using a recursive procedure.

\begin{notation} \label{notat:specialtransfer}
Let $\mN$, $\mE_1$ and $\mE_2$ be as in Definition~\ref{def:imme}. If $\mA$ is a network alphabet, $\mF$ is a network code for $(\mN,\mA)$, $\mU \subseteq \mE$, and $t \ge 0$, then we denote by 
\begin{equation} \label{chtd}
\Omega[\mN,\mA,\mF,\mE_1 \to \mE_2,\mU \cap \mE_1,t]: \mA^{|\mE_1|} \dashrightarrow \mA^{|\mE_2|}
\end{equation}
the channel that describes the transfer from the edges of $\mE_1$ to those of $\mE_2$, when an adversary can corrupt up to $t$ edges from $\mU \cap \mE_1$.
\end{notation}

In this paper, we will not formally recall the definition of the channel introduced in Notation~\ref{notat:specialtransfer}. We refer to~\cite[page 205]{RK18} for further details. We will however illustrate the channel with an example.

\begin{example} \label{extransf}
Consider again the network $\mN$ and the edge sets $\mE_1$ and $\mE_2$ in Figure~\ref{tulipsolved}. Let~$\mU=\{e_1,e_2,e_3,e_4,e_6,e_7,e_9\}$ be the set of dashed (vulnerable) edges in the figure. Let $\mA$ be a network alphabet and let~$\mF$ be a network code for $(\mN,\mA)$. We want to describe the channel in~\eqref{chtd} for $t=1$, which we denote by $\Omega: \mA^3 \dashrightarrow \mA^2$ for convenience. We have
$$\Omega(x_1,x_2,x_3)=\{(\mF_{V_1}(y_1,y_2),\mF_{V_4}(y_9)) \mid y=(y_1,y_2,y_9) \in \mA^3, \, \dH(x,y) \le 1\},$$
where $\dH$ is the Hamming distance on $\mA^3$. Note that the value of each edge of $\mE_2$ 
depends only on the values of the edges that are its immediate predecessors in $\mE_1$. For example,
when computing the values that $e_{10}$ can take, the channel only considers the values that $e_9$ can take, even though both $e_1$ and $e_2$ precede $e_{10}$.
This follows from the definition of~\eqref{chtd} proposed in~\cite{RK18}, which we adopt here.
\end{example}

\begin{remark} \label{rmk:immediate}
Note that we do not require the edge-cuts $\mE_1$ and $\mE_2$ to be minimal or \textit{antichain} cuts (i.e., cuts where any two different edges cannot be compared with respect to the order $\preccurlyeq$). Furthermore, the channel~$\Omega[\mN,\mA,\mF,\mE_1 \to \mE_2,\mU \cap \mE_1,t]$ considers the immediate predecessors first in the network topology. In other words, the channel~$\Omega[\mN,\mA,\mF,\mE_1 \to \mE_2,\mU \cap \mE_1,t]$
expresses the value of each edge of $\mE_2$ as a \textit{function} of the values of its immediate predecessors in $\mE_1$. 
\end{remark}

We are now ready to state the main result of this section.

\begin{theorem}[Double-Cut-Set Bound]
\label{thm:dcsb}
Let $\mN=(\mV,\mE, S, \bfT)$ be a network, $\mA$ a network alphabet, $\mU \subseteq \mE$ a set of edges and $t \ge 0$. Let $T \in \bd{T}$ and let $\mE_1$ and $\mE_2$ be edge-cuts between~$S$ and~$T$ with the property that~$\mE_1$ precedes~$\mE_2$. We have 
$$\CC_1(\mN,\mA,\mU,t) \le \max_{\mF} \,  \CC_1(\Omega[\mN,\mA,\mF,\mE_1 \to \mE_2,\mU \cap \mE_1,t]),$$
where the maximum is taken over all 
the network codes $\mF$ for $(\mN,\mA)$.
\end{theorem}
\begin{proof}
Fix a network code $\mF$ for
$(\mN,\mA)$. We consider the (fictitious) scenario where up to $t$ errors can occur only on the edges from $\mU \cap \mE_1$. This scenario is modeled by the concatenation of channels
\begin{multline} \label{con}
    \Omega[\mN,\mA,\mF,\out(S) \to \mE_1,\mU \cap \out(S),0]  \, \blacktriangleright \, \Omega[\mN,\mA,\mF, \mE_1 \to \mE_2,\mU \cap \mE_1,t] \\  
\blacktriangleright \,  \Omega[\mN,\mA,\mF,\mE_2 \to \inn(T),\mU \cap \mE_2,0],
\end{multline}
where the three channels 
in~\eqref{con} are of the type introduced in Notation~\ref{notat:specialtransfer}.
Note moreover that the former and the latter channels in~\eqref{con} are deterministic( see Definition \ref{dd1}) as we consider an adversarial power of 0. They describe the transfer from the source to $\mE_1$ and from $\mE_2$ to $T$, respectively.
We set $\hat{\Omega}:=\Omega[\mN,\mA,\mF,\out(S) \to \mE_1,\mU,0]$ and 
$\overline{\Omega}:=\Omega[\mN,\mA,\mF,\mE_2 \to \inn(T),\mU,0]$ to simplify the notation throughout the proof. 

The channel $\Omega[\mN,\mA,\mF,S \to T,\mU,t]$ is coarser (Definition~\ref{deffiner}) than the channel in~\eqref{con}, since in the latter the errors can only occur on a subset of $\mU$. In symbols, using Proposition~\ref{prop:11} we have
\begin{equation*} \Omega[\mN,\mA,\mF,S \to T,\mU,t] \, \ge \,  \hat{\Omega} \, \blacktriangleright \,  \Omega[\mN,\mA,\mF, \mE_1 \to \mE_2,\mU \cap \mE_1,t] \, \blacktriangleright \,  \overline{\Omega}.
\end{equation*}
By Propositions~\ref{prop:finer} and~\ref{dpi}, 
this implies
that
\begin{equation}
\label{almost}
    \CC_1(\Omega[\mN,\mA,\mF,S \to T,\mU,t]) \le \CC_1(\Omega[\mN,\mA,\mF,\mE_1 \to \mE_2, \mU \cap \mE_1,t]).
\end{equation}
Since~\eqref{almost} holds for any $\mF$, 
Proposition~\ref{prop:aux} finally gives
\begin{equation*} \label{mmm}
    \CC_1(\mN,\mA,\mU,t) \le
\max_{\mF} \, \CC_1(\Omega[\mN,\mA,\mF,\mE_1 \to \mE_2,\mU \cap \mE_1,t]),
\end{equation*}
concluding the proof.
\end{proof}

Our next step is to make the Double-Cut-Set Bound of Theorem~\ref{thm:dcsb} more explicit and~``easy'' to apply. More in detail, we now explain how Theorem~\ref{thm:dcsb} can be used to construct a simple~3-level network from a larger (possibly more complex) network $\mN$, whose capacity is an upper bound for the capacity of $\mN$. This strategy reduces the problem of computing an upper bound for the capacity of $\mN$ to that of estimating the capacity of the corresponding simple~3-level network. In turn, Subsection~\ref{sec:3to2reduc} often reduces the latter problem to that
of computing an upper bound for a simple 2-level network, a problem we have studied extensively throughout the paper.

\begin{corollary} \label{cor:tothree}
Let $\mN=(\mV,\mE, S, \bfT)$ be a network, $\mA$ a network alphabet, $\mU \subseteq \mE$ a set of edges and $t \ge 0$. Let $T \in \bd{T}$ and let $\mE_1$ and $\mE_2$ be edge-cuts between~$S$ and~$T$ with the property that~$\mE_1$ precedes~$\mE_2$. Consider a simple 3-level network $\mN'$ with source $S$, terminal~$T$, and vertex layers~$\mV_1$ and~$\mV_2$. The vertices of~$\mV_1$ are in bijection with the edges of $\mE_1$ and the vertices of~$\mV_2$ with the edges of $\mE_2$. A vertex $V \in \mV_1$ is connected to vertex $V' \in \mV_2$ if and only if the edge of~$\mE_1$ corresponding to $V$ is an immediate predecessor of the edge of~$\mE_2$ corresponding to $V'$; see Definition~\ref{def:imme}. Denote by~$\mE'_S$ the edges directly connected with the source of $\mN'$, which we identify with the edges of~$\mE_1$ (consistently with how we identified these with the vertices in $\mV_1$).
We then have
$$\CC_1(\mN,\mA,\mU,t) \le \CC_1(\mN',\mA, \mU \cap \mE'_S,t).$$
\end{corollary}

Before proving Corollary~\ref{cor:tothree}, we show how to apply it to the opening example of this paper. This will give us a sharp {upper}
bound for its capacity, as we will show.

\begin{example} \label{ex:tulipA}
Consider the network $\mN$ of Figure~\ref{fig:introex1}, where the adversary can corrupt at most~$t=1$ of the dashed edges in $\mU=\{e_1,e_2,e_3,e_4,e_6,e_7,e_9\}$.
We focus on terminal $T_1$ (a similar approach can be taken for $T_2$, since the network is symmetric) and consider the two edge-cuts~$\mE_1$ and $\mE_2$
depicted in Figure~\ref{tulipsolved}. Clearly, $\mE_1$ precedes $\mE_2$.

Following Corollary~\ref{cor:tothree}, we construct a simple 3-level network $\mN'$ with source $S$, terminal~$T$, and vertex sets $\mV_1$ and $\mV_2$ of cardinalities 3 and 2, respectively. We depict the final outcome in 
Figure~\ref{fig:3lev}, where we label the vertices and some of the edges according to the edges of $\mN$ they are in bijection with.

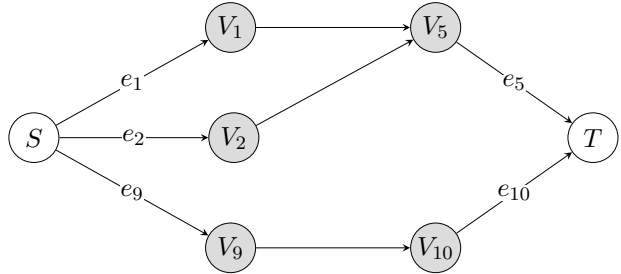
\begin{figure}[htbp]
\centering
\scalebox{0.90}{
\begin{tikzpicture}

\tikzset{vertex/.style = {shape=circle,draw,inner sep=0pt,minimum size=1.9em}}
\tikzset{nnode/.style = {shape=circle,fill=myg,draw,inner sep=0pt,minimum
size=1.9em}}
\tikzset{edge/.style = {->,> = stealth}}
\tikzset{dedge/.style = {densely dotted,->,> = stealth}}
\tikzset{ddedge/.style = {dashed,->,> = stealth}}

\node[vertex] (S1) {$S$};

\node[shape=coordinate,right=\mynodespace of S1] (K) {};

\node[nnode,above=0.5\mynodespace of K] (V1) {$V_1$};

\node[nnode,right=-0.13\mynodespace of K] (V2) {$V_2$};

\node[nnode,below=0.5\mynodespace of K] (V9) {$V_9$};

\node[nnode,right=0.9\mynodespace of V1] (V5) {$V_5$};

\node[nnode,right=0.9\mynodespace of V9 ] (V10) {$V_{10}$};

\node[vertex,right=1.8\mynodespace of V2] (T) {$T$};

\draw[edge,bend left=0] (S1)  to node[fill=white, inner sep=0pt]{$e_1$} (V1);

\draw[edge,bend right=0] (S1)  to node[fill=white, inner sep=0pt]{$e_2$} (V2);

\draw[edge,bend left=0] (S1) to  node[fill=white, inner sep=0pt]{$e_9$} (V9);

\draw[edge,bend left=0] (V1) to  node{} (V5);

\draw[edge,bend left=0] (V2) to  node{} (V5);

\draw[edge,bend left=0] (V9) to  node{} (V10);

\draw[edge,bend left=0] (V5) to  node[fill=white, inner sep=0pt]{$e_5$} (T);

\draw[edge,bend left=0] (V10) to  node[fill=white, inner sep=0pt]{$e_{10}$} (T);

\end{tikzpicture} 

}
\caption{The 3-level network $\mN'$ induced by
the network $\mN$ of Figure~\ref{tulipsolved}. \label{fig:3lev}}
\end{figure}

The next step is to make the edges of $\mU \cap \mE'_S=\{e_1,e_2,e_9\}$
vulnerable and consider an adversary capable of corrupting at most $t=1$ of them. We thus consider the network in Figure~\ref{fig:3levB} after renumbering the edges and vertices.

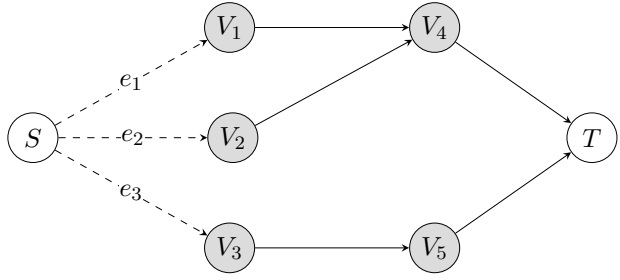
\begin{figure}[htbp]
\centering
\scalebox{0.90}{
\begin{tikzpicture}

\tikzset{vertex/.style = {shape=circle,draw,inner sep=0pt,minimum size=1.9em}}
\tikzset{nnode/.style = {shape=circle,fill=myg,draw,inner sep=0pt,minimum
size=1.9em}}
\tikzset{edge/.style = {->,> = stealth}}
\tikzset{dedge/.style = {densely dotted,->,> = stealth}}
\tikzset{ddedge/.style = {dashed,->,> = stealth}}

\node[vertex] (S1) {$S$};

\node[shape=coordinate,right=\mynodespace of S1] (K) {};

\node[nnode,above=0.5\mynodespace of K] (V1) {$V_1$};

\node[nnode,right=-0.13\mynodespace of K] (V2) {$V_2$};

\node[nnode,below=0.5\mynodespace of K] (V9) {$V_3$};

\node[nnode,right=0.9\mynodespace of V1] (V5) {$V_4$};

\node[nnode,right=0.9\mynodespace of V9 ] (V10) {$V_{5}$};

\node[vertex,right=1.8\mynodespace of V2] (T) {$T$};

\draw[ddedge,bend left=0] (S1)  to node[fill=white, inner sep=0pt]{$e_1$} (V1);

\draw[ddedge,bend right=0] (S1)  to node[fill=white, inner sep=0pt]{$e_2$} (V2);

\draw[ddedge,bend left=0] (S1) to  node[fill=white, inner sep=0pt]{$e_3$} (V9);

\draw[edge,bend left=0] (V1) to  node{} (V5);

\draw[edge,bend left=0] (V2) to  node{} (V5);

\draw[edge,bend left=0] (V9) to  node{} (V10);

\draw[edge,bend left=0] (V5) to  node{} (T);

\draw[edge,bend left=0] (V10) to  node{} (T);

\end{tikzpicture} 

}
\caption{The 3-level network $\mN'$ induced by
the network $\mN$ of Figure~\ref{tulipsolved} where the vulnerable edges are dashed. \label{fig:3levB}}
\end{figure}

We finally apply the procedure described in Subsection~\ref{sec:3to2reduc} to obtain a 2-level network from~$\mN'$, whose capacity is an upper bound for that of~$\mN'$. It is easy to check that the 
2-level network obtained from the network in Figure~\ref{fig:3levB} is precisely the Diamond Network $\mathfrak{A}_1$ introduced in Section~\ref{sec:diamond}; see Figure~\ref{fig:diamond}.
Therefore by combining Theorem~\ref{thm:diamond_cap}, Theorem~\ref{thm:channel}, and Corollary~\ref{cor:tothree}, we obtain
\begin{equation} \label{finalestimate}
    \CC_1(\mN,\mA,\mU,1) \le \log_{|\mA|}(|\mA|-1).
\end{equation}
In Theorem~\ref{computC} below, we will prove that the above bound is met with equality. In particular, the procedure described in this example to obtain the bound in~\eqref{finalestimate} is sharp, and actually leads to the exact capacity value of the opening example network from Section \ref{sec:motiv}.
\end{example}

\begin{proof}[Proof of Corollary~\ref{cor:tothree}]
We will prove that
\begin{equation*}
    \CC_1(\Omega[\mN,\mA,\mF,\mE_1 \to \mE_2,\mU \cap \mE_1,t]) \le \CC_1(\mN',\mA,\mU\cap \mE'_S,t)
\end{equation*} for every network code $\mF$ for $(\mN,\mA)$, which in turn establishes the  corollary thanks to Theorem~\ref{thm:dcsb}.
We fix $\mF$ and consider the auxiliary channel
$\Omega:=\Omega[\mN,\mA,\mF,\mE_1 \to \mE_2,\mU \cap \mE_1,0]$, which is deterministic.
By Remark~\ref{rmk:immediate}, the channel $\Omega$
expresses the value of each edge of $\mE_2$ as a function of the values of its immediate predecessors in $\mE_1$. 
By the construction of $\mN'$, there exists a network code $\mF'$ (which depends on $\mF$) for $(\mN',\mA)$ with the property that
\begin{equation} \label{cc1}
\Omega=\Omega[\mN',\mA,\mF',\mE'_S \to \inn(T),\mU \cap \mE'_S,0],
\end{equation}
where the edges of $\mE_1$ and $\mE_2$ are identified with those of $\mE'_S$ and $\inn(T)$ in $\mN'$ as explained in the statement. 
Now observe that the channel $\Omega[\mN,\mA,\mF,\mE_1 \to \mE_2,\mU \cap \mE_1,t]$ can be written as the concatenation
\begin{equation} \label{cc2}
\Omega[\mN,\mA,\mF,\mE_1 \to \mE_2,\mU \cap \mE_1,t] = \Omega[\mN,\mA,\mF,\mE_1 \to \mE_1,\mU \cap \mE_1,t] \blacktriangleright \Omega,
\end{equation}
where the first channel in the concatenation simply describes the action of the adversary on the edges of $\mU \cap \mE_1$ (in the terminology of~\cite{RK18}, the channel is called of \textit{Hamming type}; see~\cite[Sections~III and~V]{RK18}).
By combining~\eqref{cc1} with~\eqref{cc2} and using  the identifications between $\mE_1$ and $\mE_S'$,
we can write
\begin{align}
    \Omega[\mN',\mA,\mF',\mE'_S \to \inn(T),\mU \cap \mE'_S,t] &= \Omega[\mN',\mA,\mF',\mE'_S \to \mE'_S,\mU \cap \mE'_S,t]  \nonumber  \\ 
    &  \qquad \qquad \qquad \qquad \quad \blacktriangleright
\Omega[\mN',\mA,\mF',\mE'_S \to \inn(T),\mU \cap \mE'_S,0] \nonumber \\
&= \Omega[\mN,\mA,\mF,\mE_1 \to \mE_1,\mU \cap \mE_1,t] \blacktriangleright \Omega \nonumber \\
&=\Omega[\mN,\mA,\mF,\mE_1 \to \mE_2, \mU \cap \mE_1,t]. \label{lll}
\end{align}
Note that, by definition,
$\CC_1(\mN',\mA,\mU \cap \mE'_S,t) \ge 
\CC_1(\Omega[\mN',\mA,\mF',\mE'_S \to \inn(T),\mU \cap \mE'_S,t])$,
which, combined with~\eqref{lll},
leads to
$$\CC_1(\mN',\mA,\mU \cap \mE'_S,t) \ge 
\CC_1(\Omega[\mN,\mA,\mF,\mE_1 \to \mE_2,\mU \cap \mE_1,t]).$$
Since $\mF$ was an arbitrary network code for $(\mN,\mA)$, this is precisely what we wanted to show, concluding the proof of the corollary.
\end{proof}

Next, we give a capacity-achieving scheme for the network depicted in Figure \ref{fig:introex1}, proving that the estimate in~\eqref{finalestimate} is sharp. 

\begin{theorem} \label{computC}
Let $\mN$ and $\mU$ be as in Example~\ref{ex:tulipA}; see also Figure~\ref{fig:introex1}. Then for all network alphabets $\mA$ we have
$$\CC_1(\mN,\mA,\mU,1) = \log_{|\mA|}(|\mA|-1).$$
\end{theorem}

\begin{proof}
The fact that $\CC_1(\mN,\mA,\mU,1) \le \log_{|\mA|}(|\mA|-1)$ has already been shown in Example~\ref{ex:tulipA} when illustrating how to apply Corollary~\ref{cor:tothree}.
We will give a scheme that achieves the desired capacity value. Reserve an alphabet symbol $* \in \mA$. The 
source $S$ emits any symbol from~$\mA \setminus \{*\}$
 via a 4-times repetition code.
 Vertices $V_1$ and $V_2$ proceed as follows: If the symbols on their incoming edges are equal, they forward that symbol; otherwise they output $*$.
 Vertex~$V_3$ proceeds as follows:
 If one of the two received symbols is different from $*$, then it forwards that symbol. If both received symbols are different from $*$, then it outputs $*$ over $e_9$. The vertex $V_4$ just forwards. Decoding is done as follows. $T_1$ and $T_2$ look at the edges~$e_5$ and~$e_8$, respectively. If they do not receive  $*$ over those edges, they trust the received symbol. If one of them is~$*$, then the corresponding terminal trusts the outgoing edge from $V_4$. For example, if $e_5$ carries $*$, then $T_1$ trusts $e_{10}.$ It is not difficult to see that this scheme defines a network code~$\mF$ for~$(\mN,\mA)$ and an unambiguous outer code $\mC$ of cardinality $|\mA|-1$, establishing the theorem.
\end{proof}

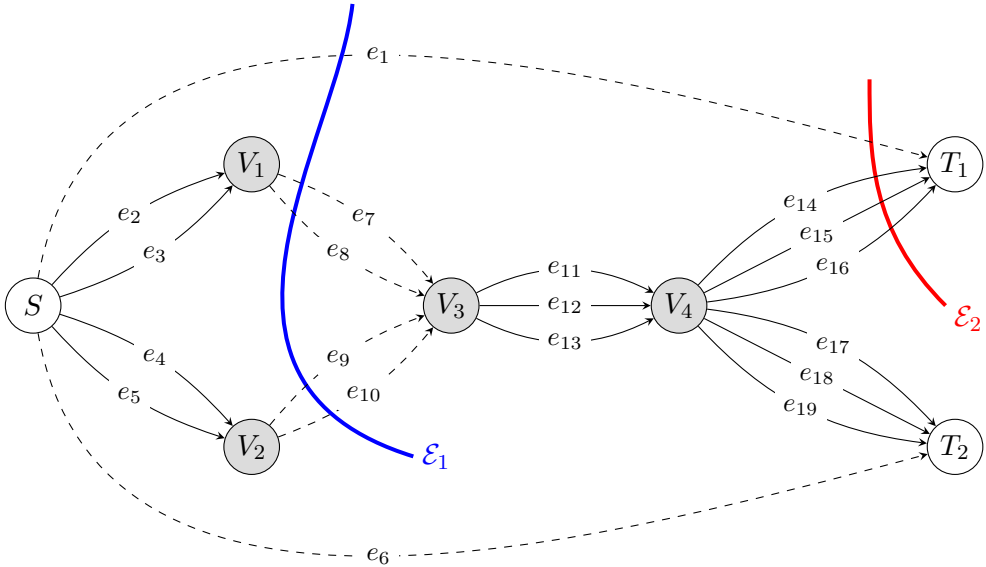
\begin{figure}[h!]
\centering
\begin{tikzpicture}

\draw[blue, line width=1.5pt] (4.2,4) .. controls (4,2.2) and (1.8,-1) .. (5,-2);

\draw[red, line width=1.5pt] (11,3) .. controls (11,2) and (11,1) .. (12,0);

\tikzset{vertex/.style = {shape=circle,draw,inner sep=0pt,minimum size=1.9em}}
\tikzset{nnode/.style = {shape=circle,fill=myg,draw,inner sep=0pt,minimum
size=1.9em}}
\tikzset{edge/.style = {->,> = stealth}}
\tikzset{dedge/.style = {densely dotted,->,> = stealth}}
\tikzset{ddedge/.style = {dashed,->,> = stealth}}

\node[vertex] (S1) {$S$};

\node[shape=coordinate,right=\mynodespace of S1] (K) {};

\node[nnode,above=0.6\mynodespace of K] (V1) {$V_1$};

\node[nnode,below=0.6\mynodespace of K] (V2) {$V_2$};

\node[nnode,right=0.9\mynodespace of K] (V3) {$V_3$};

\node[nnode,right=0.9\mynodespace of V3] (V4) {$V_4$};


\node[vertex,right=3.4\mynodespace of V1 ] (T1) {$T_1$};

\node[vertex,right=3.4\mynodespace of V2] (T2) {$T_2$};

\draw[edge,bend left=15] (S1)  to node[fill=white, inner sep=3pt]{\small $e_2$} (V1);

\draw[edge,bend right=15] (S1)  to node[fill=white, inner sep=3pt]{\small $e_3$} (V1);

\draw[edge,bend left=15] (S1) to  node[fill=white, inner sep=3pt]{\small $e_4$} (V2);

\draw[edge,bend right=15] (S1) to  node[fill=white, inner sep=3pt]{\small $e_5$} (V2);

\draw[ddedge,bend left=15] (V1) to  node[fill=white, inner sep=3pt]{\small $e_7$} (V3);

\draw[ddedge,bend right=15] (V1) to  node[fill=white, inner sep=3pt]{\small $e_8$} (V3);

\draw[edge,bend left=20] (V4)  to node[fill=white, inner sep=3pt]{\small $e_{14}$} (T1);

\draw[edge,bend left=0] (V4)  to node[fill=white, inner sep=3pt]{\small $e_{15}$} (T1);

\draw[edge,bend right=20] (V4)  to node[fill=white, inner sep=3pt]{\small $e_{16}$} (T1);

\draw[edge,bend left=20] (V4)  to node[fill=white, inner sep=3pt]{\small $e_{17}$} (T2);

\draw[edge,bend right=0] (V4)  to node[fill=white, inner sep=3pt]{\small $e_{18}$} (T2);

\draw[edge,bend right=20] (V4)  to node[fill=white, inner sep=3pt]{\small $e_{19}$} (T2);

\draw[ddedge,bend left=15] (V2) to  node[fill=white, inner sep=3pt]{\small $e_9$} (V3);

\draw[ddedge,bend right=15] (V2) to  node[fill=white, inner sep=3pt]{\small $e_{10}$} (V3);

\draw[edge,bend left=27] (V3) to  node[fill=white, inner sep=3pt]{\small $e_{11}$} (V4);

\draw[edge,bend left=0] (V3) to  node[fill=white, inner sep=3pt]{\small $e_{12}$} (V4);

\draw[edge,bend right=27] (V3) to  node[fill=white, inner sep=3pt]{\small $e_{13}$} (V4);

\draw[ddedge,out=80,in=165] (S1) to  node[fill=white, inner sep=3pt]{\small $e_1$} (T1);

\draw[ddedge,out=-80,in=-165] (S1) to  node[fill=white, inner sep=3pt]{\small $e_{6}$} (T2);

\node[text=blue] (E11) at (5.3,-2) {$\mE_1$};

\node[text=red] (E2) at (12.3,-0.17) {$\mE_2$};

\end{tikzpicture} 

\caption{{{Network $\mN$ for Example \ref{ex:second}.}}\label{fig:secondex}}
\end{figure}

We conclude this section by illustrating with another example how the results of this paper can be combined and  applied to derive upper bounds for the capacity of a large network.

\begin{example}
\label{ex:second}
Consider the network $\mN$  and the edge sets $\mE_1$ and $\mE_2$ depicted in Figure \ref{fig:secondex}. Both $\mE_1$ and $\mE_2$ are edge-cuts between $S$ and $T_1$. Moreover, $\mE_1$ precedes~$\mE_2$.

We start by observing that if there is no adversary present, then the capacity of the network~$\mN$ of Figure \ref{fig:secondex} is at most 4 since the min-cut between $S$ and any terminal~$T \in \{T_1,T_2\}$ is 4. 
It is straightforward to design a strategy that achieves this rate.

When the adversary is allowed to change \textit{any} of the network edges, then 
Theorem \ref{thm:mcm} gives that the capacity is equal to 2, under certain assumptions on the alphabet.

Now consider an adversary able to corrupt at most~$t=1$ of the  edges from the set  $\mU=\{e_1,e_7,e_8,e_9,e_{10}\}$, which are dashed in Figure~\ref{fig:secondex}.
In this situation, the capacity expectation increases from the fully vulnerable case, and the Generalized Network Singleton Bound of Theorem \ref{sbound} predicts 3 as the largest achievable rate. 
Using the results of this paper, we will show that a rate of 3 is actually not achievable. Following Corollary~\ref{cor:tothree}, we construct a simple~3-level network $\mN'$ induced from $\mN$. We depict the final outcome in 
Figure~\ref{fig:3rc}.

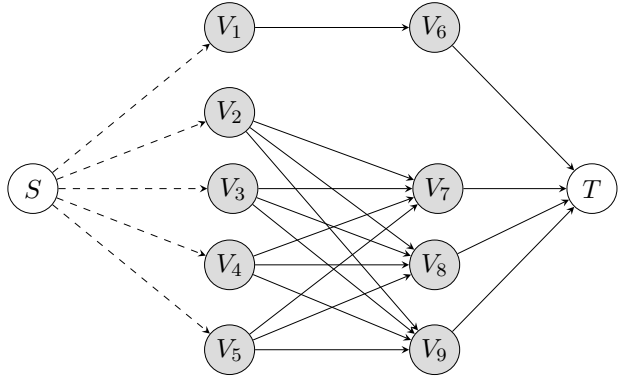
\begin{figure}[htbp]
\centering
\scalebox{0.90}{
\begin{tikzpicture}

\tikzset{vertex/.style = {shape=circle,draw,inner sep=0pt,minimum size=1.9em}}
\tikzset{nnode/.style = {shape=circle,fill=myg,draw,inner sep=0pt,minimum
size=1.9em}}
\tikzset{edge/.style = {->,> = stealth}}
\tikzset{dedge/.style = {densely dotted,->,> = stealth}}
\tikzset{ddedge/.style = {dashed,->,> = stealth}}

\node[vertex] (S1) {$S$};

\node[shape=coordinate,right=\mynodespace of S1] (K) {};

\node[nnode,above=0.8\mynodespace of K] (V1) {$V_1$};

\node[nnode,above=0.3\mynodespace of K] (V2) {$V_2$};

\node[nnode,right=-0.13\mynodespace of K] (V3) {$V_3$};

\node[nnode,below=0.3\mynodespace of K] (V4) {$V_4$};

\node[nnode,below=0.8\mynodespace of K] (V5) {$V_5$};

\node[nnode,right=0.9\mynodespace of V1] (V6) {$V_6$};

\node[nnode,right=0.9\mynodespace of V3 ] (V7) {$V_{7}$};

\node[nnode,right=0.9\mynodespace of V4] (V8) {$V_8$};

\node[nnode,right=0.9\mynodespace of V5 ] (V9) {$V_{9}$};

\node[vertex,right=1.8\mynodespace of V3] (T) {$T$};

\draw[ddedge,bend left=0] (S1)  to node{} (V1);

\draw[ddedge,bend right=0] (S1)  to node{} (V2);

\draw[ddedge,bend left=0] (S1) to  node{} (V3);

\draw[ddedge,bend right=0] (S1)  to node{} (V4);

\draw[ddedge,bend left=0] (S1) to  node{} (V5);

\draw[edge,bend left=0] (V1) to  node{} (V6);

\draw[edge,bend left=0] (V2) to  node[]{} (V7);

\draw[edge,bend left=0] (V2) to  node[]{} (V8);

\draw[edge,bend left=0] (V2) to  node[]{} (V9);

\draw[edge,bend left=0] (V3) to  node[]{} (V7);

\draw[edge,bend left=0] (V3) to  node[]{} (V8);

\draw[edge,bend left=0] (V3) to  node[]{} (V9);

\draw[edge,bend left=0] (V4) to  node[]{} (V7);

\draw[edge,bend left=0] (V4) to  node[]{} (V8);

\draw[edge,bend left=0] (V4) to  node[]{} (V9);

\draw[edge,bend left=0] (V5) to  node[]{} (V7);

\draw[edge,bend left=0] (V5) to  node[]{} (V8);

\draw[edge,bend left=0] (V5) to  node[]{} (V9);

\draw[edge,bend left=0] (V6) to  node{} (T);

\draw[edge,bend left=0] (V7) to  node[]{} (T);

\draw[edge,bend left=0] (V8) to  node{} (T);

\draw[edge,bend left=0] (V9) to  node[]{} (T);

\end{tikzpicture} 

}
\caption{The 3-level network $\mN'$ induced by
the network $\mN$ of Figure~\ref{fig:secondex}. Vulnerable edges are dashed. \label{fig:3rc}}
\end{figure}

Lastly, we apply the procedure described in Subsection~\ref{sec:3to2reduc} to obtain a 2-level network from~$\mN'$, whose capacity will be an upper bound for that of~$\mN'$. It can easily be seen that the~2-level network obtained is precisely the network $\mathfrak{B}_3$ of Family \ref{fam:b} introduced in Section~\ref{sec:families}. This is depicted in Figure \ref{fig:2rc}.
Therefore, by combining Theorem~\ref{thm:channel}, Theorem \ref{thm:notmet} and Corollary~\ref{cor:tothree}, we finally obtain $$\CC_1(\mN,\mA,\mU,1) < 3.$$
At the time of writing this paper we cannot give an exact expression for the value of $\CC_1(\mN,\mA,\mU,1)$ for an arbitrary alphabet $\mA$.
This remains an open problem.
\end{example}

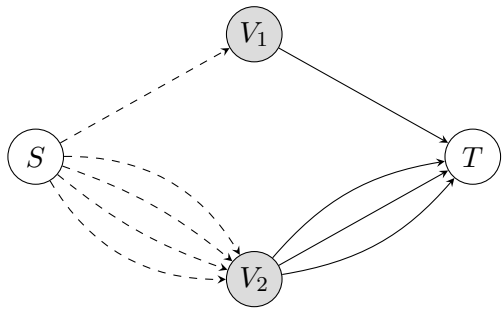
\begin{figure}[htbp]
\centering
\begin{tikzpicture}
\tikzset{vertex/.style = {shape=circle,draw,inner sep=0pt,minimum size=1.9em}}
\tikzset{nnode/.style = {shape=circle,fill=myg,draw,inner sep=0pt,minimum
size=1.9em}}
\tikzset{edge/.style = {->,> = stealth}}
\tikzset{dedge/.style = {densely dotted,->,> = stealth}}
\tikzset{ddedge/.style = {dashed,->,> = stealth}}

\node[vertex] (S1) {$S$};

\node[shape=coordinate,right=\mynodespace of S1] (K) {};

\node[nnode,above=0.5\mynodespace of K] (V1) {$V_1$};
\node[nnode,below=0.5\mynodespace of K] (V2) {$V_2$};

\node[vertex,right=2\mynodespace of S1] (T) {$T$};

\draw[ddedge,bend left=0] (S1)  to node[]{} (V1);
\draw[ddedge,bend left=30] (S1)  to node[]{} (V2);
\draw[ddedge,bend left=10] (S1) to  node[]{} (V2);
\draw[ddedge,bend right=10] (S1) to  node[]{} (V2);
\draw[ddedge,bend right=30] (S1) to  node[]{} (V2);

\draw[edge,bend right=0] (V1) to  node[]{} (T);
\draw[edge,bend left=20] (V2) to  node[]{} (T);

\draw[edge,bend left=0] (V2) to  node[]{} (T);
\draw[edge,bend right=20] (V2) to  node[]{} (T);

\end{tikzpicture} 

\caption{{{The simple 2-level network associated to the network 
of Figure~\ref{fig:3rc}.}}}\label{fig:2rc}
\end{figure}

\section{Linear Capacity}
\label{sec:linear}

As mentioned in Sections \ref{sec:motiv} and \ref{sec:channel}, 
in the presence of an ``unrestricted'' adversarial noise
the~(1-shot) capacity of a network can be achieved by combining a rank-metric (outer) code with a~\textit{linear} network code;
see~\cite{SKK,MANIAC,RK18,KK1}. In words, this means that the intermediate nodes of the network focus on spreading information, while decoding is performed in an end-to-end fashion.

In this section, we show that
the strategy outlined above is far from being optimal when the adversary is restricted to operate 
on a proper subset of the network edges.
In fact, we establish
some strong separation results 
between the capacity (as defined in Section~\ref{sec:channel}) and the ``linear'' capacity of a network, 
which we define by imposing that the intermediate nodes combine packets linearly. This indicates that implementing 
network \textit{decoding} becomes indeed necessary to achieve capacity in the scenario where the adversary is restricted.
The following definitions make the concept of linear capacity rigorous.

\begin{definition}
Let $\mN=(\mV,\mE, S, \bfT)$ be a network and $\mA$ an alphabet. Consider a \textbf{network code} $\mF$ for  $(\mN,\mA)$ as in Definition \ref{def:nc}. 
We say that $\mF$ is a 
\textbf{linear} network code if $\mA$ is a finite field
and each function~$\mF_V$ is $\mA$-linear.
\end{definition}

We can now define the linear version of the 1-shot capacity of an adversarial network, i.e., the analogue of Definition~\ref{def:capacities}.

\begin{definition} \label{def:lin_capacities}
Let $\mN=(\mV,\mE, S, \bfT)$ be a network, $\mA$ a finite field,
$\mU \subseteq \mE$ an edge set,
and~$t \ge 0$ an integer. 
The (\textbf{1-shot}) \textbf{linear capacity}
of $(\mN,\mA,\mU,t)$
is the largest 
real number~$\kappa$ for which there exists an \textbf{outer code} $$\mC \subseteq \mA^{\degout(S)}$$
and a linear network code $\mF$ for~$(\mN,\mA)$ 
with $\kappa=\log_{|\mA|}(|\mC|)$ such that
$\mC$ is unambiguous for each channel
$\Omega[\mN,\mA,\mF,S \to T,\mU,t]$, $T \in \bd{T}$. The notation for such largest $\kappa$ is
$$\CC^\lin_1(\mN,\mA,\mU,t).$$
\end{definition}

Note that in the definition of linear  capacity we do not require $\mC$ to be a linear code, but only that the network code $\mF$ is linear.

The first result of this section shows that the linear capacity of any member of Family~\ref{fam:d} is zero. This is in sharp contrast with Theorem~\ref{thm:metd}.

\begin{theorem} \label{thm:linmirr}
Let $\mathfrak{D}_t=(\mV,\mE,S,\{T\})$ be a member of
Family~\ref{fam:d}. Let $\mA$ be any {finite field} 
and let $\mU_S$ be the  set of edges of $\mathfrak{D}_t$ directly connected to $S$. We have $$\CC^\lin_1(\mathfrak{D}_t,\mA,\mU_S,t)= 0.$$ 
In particular, the linear capacity of the Mirrored Diamond Network of Figure~\ref{fig:mirrored} is zero.
\end{theorem}
\begin{proof}
{Let $q:=|\mA|$.} Fix any {linear} network code $\mF=\{\mF_1,\mF_2\}$
for $(\mathfrak{D}_t,\mA)$ and
let $\mC$ be an unambigious code for the channel
$\Omega[\mathfrak{D}_t,\mA,\mF,S \to T,\mU_S,t]$.
Suppose that $|\mC| \ge 2$ and let~$x,a \in \mC$ with $x \neq a$ such that
$$x = (x_1,\ldots,x_{2t},x_{2t+1},\ldots,x_{4t}), \quad a = (a_1,\ldots,a_{2t},a_{2t+1},\ldots,a_{4t}),$$
and
$$\mF_1(u_1,\ldots,u_{2t}) = \sum_{i=1}^{2t}\lambda_iu_i, \quad \mF_2(u_{2t+1},\ldots,u_{4t}) = \sum_{i=2t+1}^{4t}\lambda_iu_i,$$ where $\lambda_r \in \mathbb{F}_q$ for $1 \le r \le 4t$ and $u \in \mA^{4t}$. 
We let $\Omega := \Omega[\mathfrak{D}_t,\mA,\mF,S \to T,\mU_S,t]$ to simplify the notation throughout the remainder of the proof. 

We start by observing that $\lambda_1,\ldots,\lambda_{2t}$ cannot all be 0. Similarly, $\lambda_{2t+1},\ldots,\lambda_{4t}$ cannot all be 0 (it is easy to see that the adversary can 
cause ambiguity otherwise). Therefore we shall assume $\lambda_1 \ne 0$ and $\lambda_{4t} \ne 0$ without loss of generality. 
We will now construct vectors $y,b \in \mA^{4t}$
such that $\dH(x,y) = \dH(a,b)= 1$.
Concretely, let
\begin{itemize}
    \item $y_i = x_i \mbox{ for } 1\le i \le 4t-1$,
    \item $y_{4t} = a_{4t} + \sum_{i=2t+1}^{4t-1} \lambda_{4t}^{-1}\lambda_i(a_i-x_i)$,
    \item $b_{1} = x_{1} + \sum_{i=2}^{2t} \lambda_{1}^{-1}\lambda_i(x_i-a_i)$,
    \item $b_i = a_i \mbox{ for } 2\le i \le 4t$.
\end{itemize}
It follows from the definitions that $\dH(x,y) = \dH(a,b)= 1$ and that
$$z_x:=\left(\sum_{r=1}^{2t} \lambda_ry_r,\ \sum_{r=2t+1}^{4t} \lambda_ry_r \right) \in \Omega(x), \qquad z_a:=\left(\sum_{r=1}^{2t} \lambda_rb_r,\ \sum_{r=2t+1}^{4t} \lambda_rb_r\right) \in \Omega(a).$$ 
However, one easily checks that
$z_x=z_a$, which in turn implies that
$\Omega(x) \cap \Omega(a) \neq \emptyset$. Since $x$ and $a$ were arbitrary elements of $\mC$, this 
establishes the theorem.
\end{proof}

By proceeding as in the proof of Theorem~\ref{thm:linmirr}, one can check that the linear capacity
of any member of Family~\ref{fam:e} is zero as well. This can also be established by observing that $\mathfrak{E}_t$ is~a ``subnetwork'' of $\mathfrak{D}_t$ for all $t$.

\begin{theorem} \label{thm:8.4}
Let $\mathfrak{E}_t=(\mV,\mE,S,\{T\})$ be a member of
Family~\ref{fam:e}. Let $\mA$ be any finite field and let $\mU_S$ be the  set of edges of $\mathfrak{E}_t$ directly connected to $S$. We have  $$\CC^\lin_1(\mathfrak{E}_t,\mA,\mU_S,t)= 0.$$ 
In particular, the linear capacity of the Diamond Network of Section~\ref{sec:diamond} is zero.
\end{theorem}

We conclude this section by observing
that the proof of 
Proposition \ref{prop:lin} 
actually uses a linear network code.
In particular, the following holds.

\begin{proposition}
\label{cor:ll}
Let $\mN=([a_1,\ldots,a_n],[b_1,\ldots,b_n])=(\mV,\mE,S,\{T\})$ be a simple 2-level network, $\mA$ a sufficiently large finite field, $t \ge 0$.
Let $\mU_S$ denote the set of edges directly connected to $S$. Then
$$\CC^\lin_1(\mN,\mA,\mU_S,t) \ge  \max\left\{0,\sum_{i=1}^n\min\{a_i,b_i\} -2t\right\}.$$
\end{proposition}

Finally, by combining 
 Proposition \ref{cor:ll} and Theorem \ref{thm:notmet}, we obtain the following result on the capacities of the members of Family~\ref{fam:b}.

\begin{corollary}
\label{cor:sbs}
Let $\mathfrak{B}_s=(\mV,\mE,S,\{T\})$ be a member of
Family~\ref{fam:b}. Let $\mA$ be a sufficiently large finite field and let $\mU_S$ be the  set of edges of $\mathfrak{B}_s$ directly connected to $S$. We have  $$s >\CC_1(\mathfrak{B}_s,\mA,\mU_S,1) \ge \CC^\lin_1(\mathfrak{B}_s,\mA,\mU_S,1) \ge s-1.$$ 
\end{corollary}

\section{Conclusions and Future Research Directions}
\label{sec:open}

In this paper, we considered the 
1-shot capacity of multicast networks 
affected by adversarial noise
restricted to a proper subset of vulnerable edges.
We introduced a formal framework to study
these networks based on the notions of  adversarial channels and fan-out sets.
We defined five families 
of 2-level networks that 
play the role of fundamental 
stepping stones in our developed theory,
and derived upper and lower bounds for the capacities of these families.
We also showed that upper 
bounds for 2-level and 3-level networks can be ported to arbitrarily large networks via a Double-Cut-Set Bound.
Finally, we analyzed the capacity of certain partially vulnerable networks under the assumption that the intermediate nodes combine packets linearly.

The results presented in this paper show that
classical approaches to estimate or achieve capacity in multicast communication networks affected by an unrestricted adversary (cut-set bounds, linear network coding, rank-metric codes) are far from being optimal when the adversary is restricted to operate on a proper subset of the network edges. Moreover, the non-optimality
comes precisely from limiting the adversary to operate on a certain region of the network, which in turn forces the intermediate nodes to partially \textit{decode} information before forwarding it towards the terminal. This is in strong contrast with the typical scenario within network coding, where capacity can be achieved in an \textit{end-to-end} fashion using (random) linear network coding combined with a rank-metric code.

We conclude this paper by mentioning three research directions that naturally originate from our work.

\begin{enumerate}
    \item It remains an open problem to compute the capacities of three of the five fundamental families of networks we introduced in Subsection~\ref{sec:families} for arbitrary values of the parameters.
    We believe that this problem is challenging and requires developing new coding theory methods of combinatorial flavor that extend traditional packing arguments.
    
    \item A very natural continuation of this paper lies in the study of the scenario where a network can be used multiple times for communication, which we excluded from this first treatment.
    The multi-shot scenario is modeled by the \textit{power channel} construction; see~\cite{RK18,shannon_zero}.
    
    \item Most of our results and techniques extend to networks having multiple sources. This is another research direction that arises from this paper very naturally.
   
\end{enumerate}

\bigskip

\bibliographystyle{abbrv}
\bibliography{ADV}

\end{document}